\documentclass[11pt,a4paper]{article}
\usepackage{epsfig}
\usepackage{graphics}
\usepackage{subfigure}

\usepackage{epstopdf}
\usepackage{hyperref}

\usepackage[english]{babel}
\usepackage{mdwlist}
\usepackage{graphicx}
\usepackage{hyperref}
\usepackage{appendix}
\usepackage{dsfont}
\usepackage{tikz}
\usepackage{amsmath}
\usepackage{amsbsy}
\usepackage{amsfonts}
\usepackage{amssymb}
\usepackage{amscd}
\usepackage{amsthm}
\usepackage{empheq}
\usepackage{placeins}

\usepackage{bm}
\usepackage{xspace}

\makeatletter
\def\XS{\xspace}
\DeclareMathAlphabet{\mathb}{OML}{cmm}{b}{it}
\def\sbm#1{\ensuremath{\mathb{#1}}}                
\def\scu#1{\ensuremath{\mathcal{#1\XS}}}           
\def\sbl#1{\ensuremath{\mathbb{#1}}}              

  \def\wb{{\sbm{w}}\XS}

\def\Fc{{\scu{F}}\XS}

\def\Ic{{\scu{I}}\XS}   
   
\def\Kc{{\scu{K}}\XS}

\def\Cbb{{\sbl{C}}\XS}  
  
\def\Ebb{{\sbl{E}}\XS}

\def\Pbb{{\sbl{P}}\XS}  
  
\def\Rbb{{\sbl{R}}\XS}

\def\sgn{{\mathrm{sign}}}							

\newsavebox{\fminibox}
\newlength{\fminilength}

  \def\+{^\dagger}

\def\nequiv{\not\kern-.05em\equiv}
\def\egal{\kern-.5em=\kern-.5em}        
\def\propt{\kern-.2em\propto\kern-.2em} 

\def\intdouble{\int\kern-0.3em\int}
\def\inttriple{\int\kern-0.3em\int\kern-0.3em\int}

\def\rond#1{\overset{\kern-0.33em~_\circ}{#1}}
\def\rondit[#1]#2{\overset{\kern#1~_\circ}{#2}}

\newcommand{\Id}{\ensuremath{\mathrm{Id}}}

\def\qed{\ifmmode\hbox{\hfill\sqb}\else{\ifhmode\unskip\fi%
\nobreak\hfil
\penalty50\hskip1em\null\nobreak\hfil$\blacksquare$
\parfillskip=0pt\finalhyphendemerits=0\endgraf}\fi}

\RequirePackage{amsmath}
\RequirePackage{xspace}

\def\XS{\xspace}
\DeclareMathAlphabet{\mathb}{OML}{cmm}{b}{it}
\def\sbm#1{\ensuremath{\mathb{#1}}}                
\def\scu#1{\ensuremath{\mathcal{#1\XS}}}           

  \def\wb{{\sbm{w}}\XS}

\def\Fc{{\scu{F}}\XS}

\def\Ic{{\scu{I}}\XS}   
   
\def\Kc{{\scu{K}}\XS}

\def\MAP{^{\kern1pt{\rm MAP}\kern-1pt}}

\usepackage{color,soul}

\def\Cbb{{\sbl{C}}\XS}  
  
\def\Ebb{{\sbl{E}}\XS}

\def\Pbb{{\sbl{P}}\XS}  
  
\def\Rbb{{\sbl{R}}\XS}

\def\PM{\kern0pt^{\textrm{{\scriptsize PM}}}\kern0pt}
\def\MMAP{\kern1pt^{\textrm{{\tiny MMAP}}}\kern-1pt} 

\def\rem#1{}                    

 \def\btabu{\begin{tabular}}             \def\etabu{\end{tabular}}
 
\makeatother

\newtheorem{thmchapter}{Theorem}[section]
\newtheorem{defchapter}[thmchapter]{Definition}
\newtheorem{prop}[thmchapter]{Proposition}
\newtheorem{corollary}[thmchapter]{Corollary}
\newtheorem{lemme}[thmchapter]{Lemma}
\newtheorem{remark}[thmchapter]{Remark}

\title{Compressed sensing with structured sparsity and structured acquisition}

\author{Claire Boyer$^{(1)}$, J\'er\'emie Bigot$^{(2)}$ and Pierre Weiss$^{(1,3)}$
\\
\\
$^{(1)}$ Institut de Math\'ematiques de Toulouse (UMR 5219), CNRS, Universit\'e de Toulouse, France\\
{\small {claire.boyer}@math.univ-toulouse.fr} \\
$^{(2)}$ Institut de Math\'ematiques de Bordeaux (UMR 5251), CNRS, Universit\'e de Bordeaux, France\\
{\small {jeremie.bigot}@math.u-bordeaux1.fr} \\
$^{(3)}$ Institut des Technologies Avanc\'ees du Vivant (USR 3505), CNRS, Toulouse, France\\
{\small {pierre.armand.weiss}@gmail.com} 
}

\begin{document}
\maketitle

\begin{abstract}
Compressed Sensing (CS) is an appealing framework for applications such as Magnetic Resonance Imaging (MRI). However, up-to-date, the sensing schemes suggested by CS theories are made of random isolated measurements, which are usually incompatible with the physics of acquisition. To reflect the physical constraints of the imaging device, we introduce the notion of blocks of measurements: the sensing scheme is not a set of isolated measurements anymore, but a set of groups of measurements which may represent any arbitrary shape (parallel or radial lines for instance).
Structured acquisition with blocks of measurements are easy to implement, and provide good reconstruction results in practice.
However, very few results exist on the theoretical guarantees of CS reconstructions in this setting.
In this paper, we derive new CS results for structured acquisitions and signals satisfying a prior structured sparsity.
The obtained results provide a recovery probability of sparse vectors that explicitly depends on their support. 
Our results are thus support-dependent and offer the possibility for flexible assumptions on the sparsity structure. 
Moreover, the results are drawing-dependent, since we highlight an explicit dependency between the probability of reconstructing a sparse vector and the way of choosing the blocks of measurements.
Numerical simulations show that the proposed theory is faithful to experimental observations.
\end{abstract}



\section{Introduction}

Since its introduction in \cite{candes2006stable,donoho2006compressed}, compressed sensing triggered a massive interest in fundamental and applied research. However, despite recent progresses, existing theories are still insufficient to explain the success of compressed acquisitions in many practical applications. Our aim in this paper is to extend the applicability of the theory by combining two new ingredients: structured sparsity and acquisition structured by blocks.

\subsection{A brief history of compressed sensing}

In this section, we provide a brief history of the compressed sensing evolution, with a particular emphasis on Fourier imaging, in order to better highlight our contribution.

\subsubsection{Sampling with matrices with i.i.d. entries}

Compressed sensing - as proposed in \cite{candes2006near} - consists in recovering a signal $x\in \Cbb^n$, from a vector of measurements $y=Ax$, where $A\in \Cbb^{m\times n}$ is the sensing matrix. Typical theorems state that if $A$ is an i.i.d.\ Gaussian matrix, $x$ is $s$-sparse, and $m \gtrsim  s\log(n)$, then $x$ can be recovered exactly from $y$ by solving the following $\ell_1$ minimization problem:
\begin{equation}\label{pb:minL1}
\min_{x\in \Cbb^n, Ax=y} \|x\|_1.
\end{equation}
Moreover, it can be shown that the recovery is robust to noise if the constraint  in \eqref{pb:minL1} is penalized. An important fact about this theorem is that the number of measurements mostly depends on the intrinsic dimension $s$ rather than the ambient dimension $n$.

\subsubsection{Uniform sampling from incoherent bases}
\label{sec:uniformsampling}

Nearly at the same time, the theory was extended to random linear projections from orthogonal bases \cite{candes2006stable,rauhut2010compressive,candes2011probabilistic,foucart2013mathematical}. 
Let $A_{0} \in \Cbb^{n \times n}$ denote an orthogonal matrix with rows $(a_i^*)_{1\leq i \leq n} \in \Cbb^n$. 
A sensing matrix $A$ can be constructed by randomly drawing rows as follows
\begin{equation}
\label{eq:unisampintro}
A = \frac{1}{\sqrt{m}} \left( \frac{1}{\sqrt{\pi_{J_\ell}}} a_{J_\ell}^* \right)_{1\leq \ell \leq m}, 
\end{equation}
where $\left( J_\ell \right)_{1\leq \ell \leq m}$ are i.i.d.\ copies of a uniform random variable $J$ with $\Pbb(J = j) = \pi_j =  1/n$,  for all $ 1 \leq j \leq n$.
The coherence of matrix $A_0$ can be defined by
$$ \kappa(A_0) = n \cdot \max_{1\leq i \leq n} \| a_i\|_\infty^2.
$$

A typical result in this setting states that if $m\gtrsim \kappa(A_0) s \ln (n /\varepsilon)$ then an $s$-sparse vector $x$ can be exactly recovered  using the $\ell^1$-minimization problem \eqref{pb:minL1} with probability exceeding $1-\varepsilon$.
This type of theorem is particularly helpful to explain the success of recovery of sparse signals (spikes) from Fourier measurements, since in that case $\kappa(A_0)=1$. 


\subsubsection{The emergence of variable density sampling}

Unfortunately, in most applications, the sensing matrix $A_0$ is coherent, meaning that $\kappa(A_0)$ is large. 
In pratice, uniformly drawn measurements lead to very poor reconstructions.
A natural idea to reduce the coherence consists in drawing the highly coherent rows of $A_0$ more often than the others. 

A byproduct of standard compressed sensing results \cite{candes2011probabilistic} implies that variable density sampling \cite{puy2011variable,chauffert2013variable,krahmer2013stable} allows perfect reconstruction with a limited (but usually too high) number of measurements. This idea is captured by the following result. 

Let $A_{0} \in \Cbb^{n \times n}$ denote an orthogonal matrix with rows $(a_i^*)_{1\leq i \leq n} \in \Cbb^n$. 
Let $A$ denote the random matrix
\begin{equation}
\label{eq:vdssampintro}
A = \frac{1}{\sqrt{m}} \left( \frac{1}{\sqrt{\pi_{J_\ell}}} a_{J_\ell}^* \right)_{1\leq \ell \leq m}, 
\end{equation}
where $\left( J_\ell \right)_{1\leq \ell \leq m}$ are i.i.d.\ copies of a random variable $J$ with $\Pbb(J = j) = \pi_j = \frac{\|a_j\|^2_\infty}{\sum_{j=1}^n \|a_j\|_\infty^2}$,  for all $ 1 \leq j \leq n$.

Let $x$ denote an $s$-sparse vector and set $m\gtrsim \left(\sum_{j=1}^n\|a_j\|_\infty^2\right) s \ln (n /\varepsilon)$. 
Then, the minimizer of \eqref{pb:minL1} coincides with $x$, with probability larger than $1-\varepsilon$.

Unfortunately, it is quite easy to show experimentally, that this principle cannot explain the success of CS in applications such as Fourier imaging. The flip test proposed in \cite{adcock2013breaking} is a striking illustration of this fact.

\subsubsection{Variable density sampling with structured sparsity}

A common aspect of the above results is that they assume no structure - apart from sparsity - in the signals to recover. Recovering arbitrary sparse vectors is a very demanding property that precludes the use of CS in many practical settings. 
Exact recovery conditions for sparse vectors with a structured support appeared quite early, with the work of Tropp \cite{tropp2006just}. To the best of our knowledge, the work \cite{adcock2013breaking} is the first to provide explicit constructions of random matrices allowing to recover sparse signals with a structured support. 
The theory in \cite{adcock2013breaking}, also suggests variable density sampling strategies. 
There is however one major difference compared to the previously mentioned contributions: the density should depend both on the sensing basis \emph{and} the sparsity structure. 
The authors develop a comprehensive theory for Fourier sampling, based on isolated measurements under a sparsity-by-levels assumption in the wavelet domain. They illustrate through extensive numerical experiments in \cite{adcock2014quest} that sampling structured signals in coherent bases can significantly outperform i.i.d.\ Gaussian measurements - usually considered as an optimal sampling strategy. This theory will be reviewed and compared to ours in Section \ref{sec:Applications}.

\subsubsection{An example in MRI}

To fix the ideas, let us illustrate the application of the previously described theory in the context of Magnetic Resonance Imaging (MRI).
In MRI, images are sampled in the Fourier domain and can be assumed to be sparse in the wavelet domain. 
Figure \eqref{fig:reconstructionIso} (a) illustrates a variable density sampling pattern: the white dots indicate which Fourier coefficients are probed.
Figure \eqref{fig:reconstructionIso} (b) is the reconstruction of a phantom image from the measurements in (a) via $\ell^1$-minimization.
Figure \eqref{fig:reconstructionIso} (c) is a zoom on the reconstruction. 
As can be seen, only $4.6\%$ of the coefficients are enough to reconstruct a well resolved image.

\begin{figure}[h!]
\begin{center}
\btabu{@{}cccc}
\includegraphics[width=0.3\linewidth]{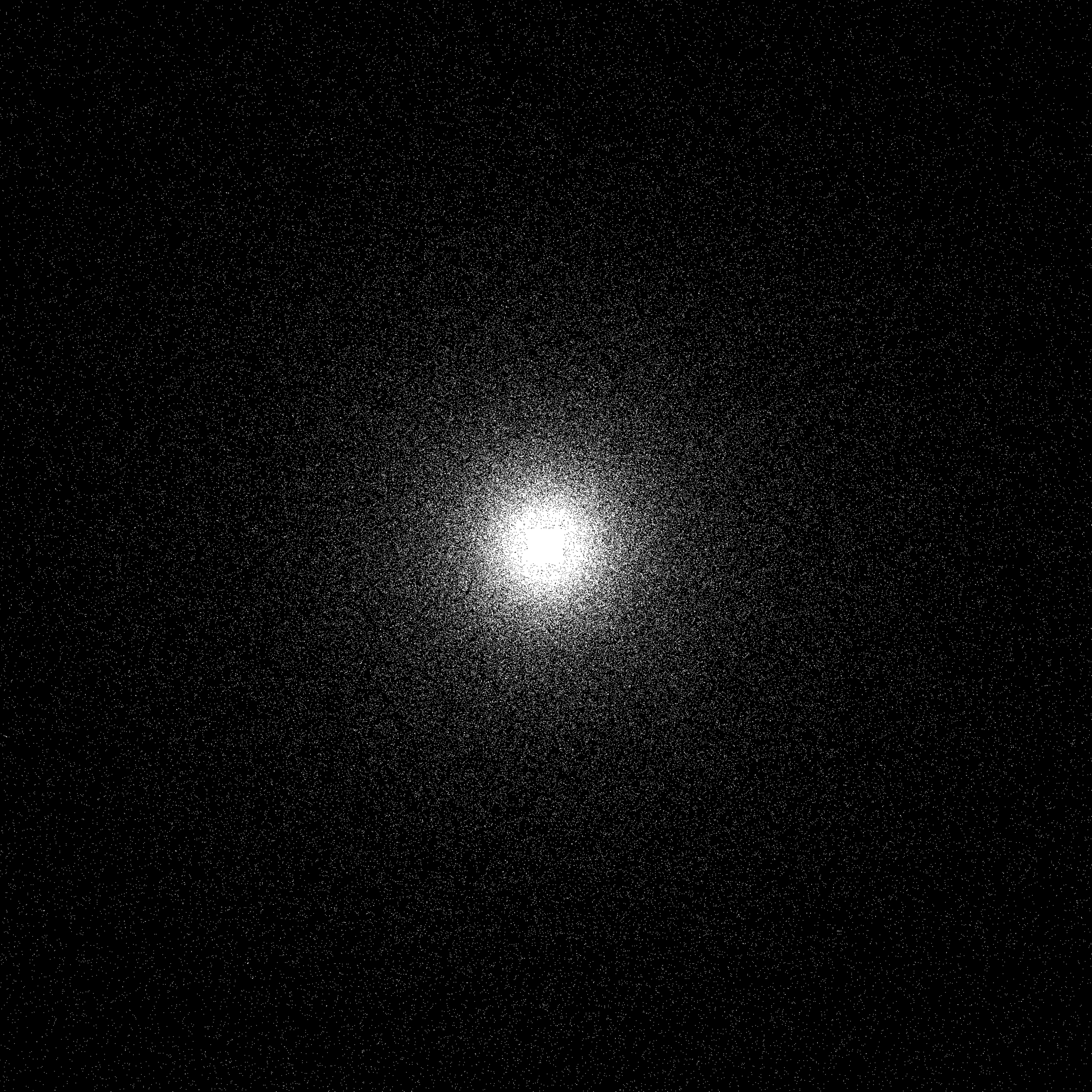} &
\hspace{-0.3cm}\includegraphics[width=0.3\linewidth]{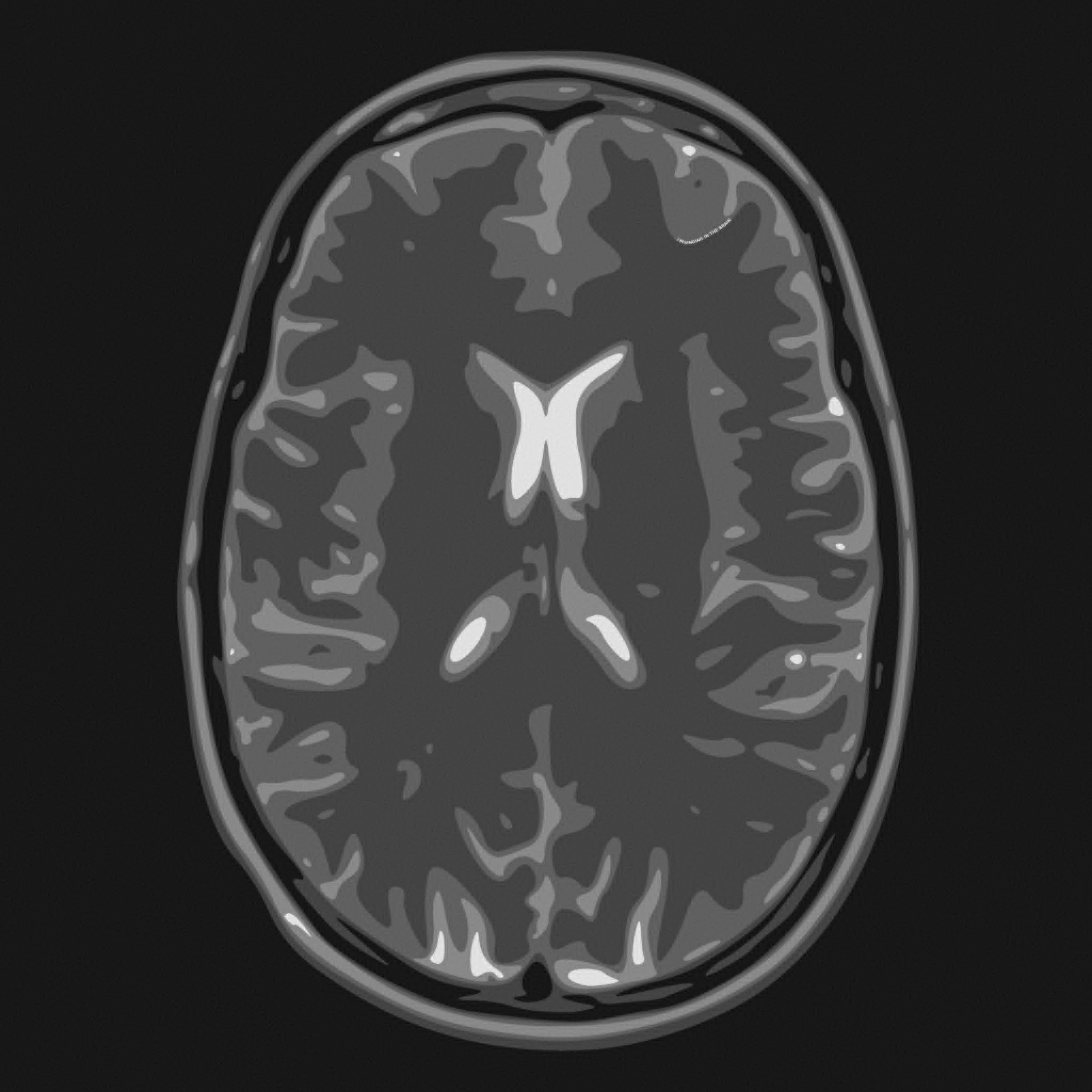} &
\hspace{-0.3cm}\includegraphics[width=0.3\linewidth]{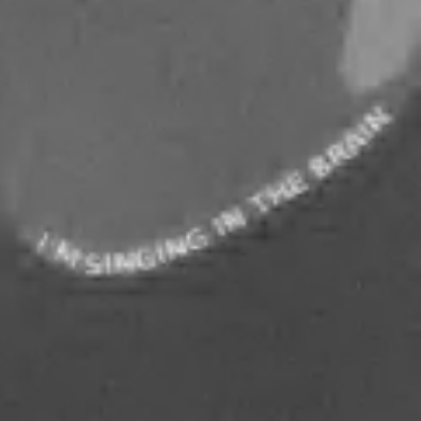} 
\\
{\small (a)}&{\hspace{-0.3cm}\small (b) SNR = 24.2 dB} & {\hspace{-0.3cm}\small (c)} \\
&\hspace{-0.3cm}\includegraphics[width=0.3\linewidth]{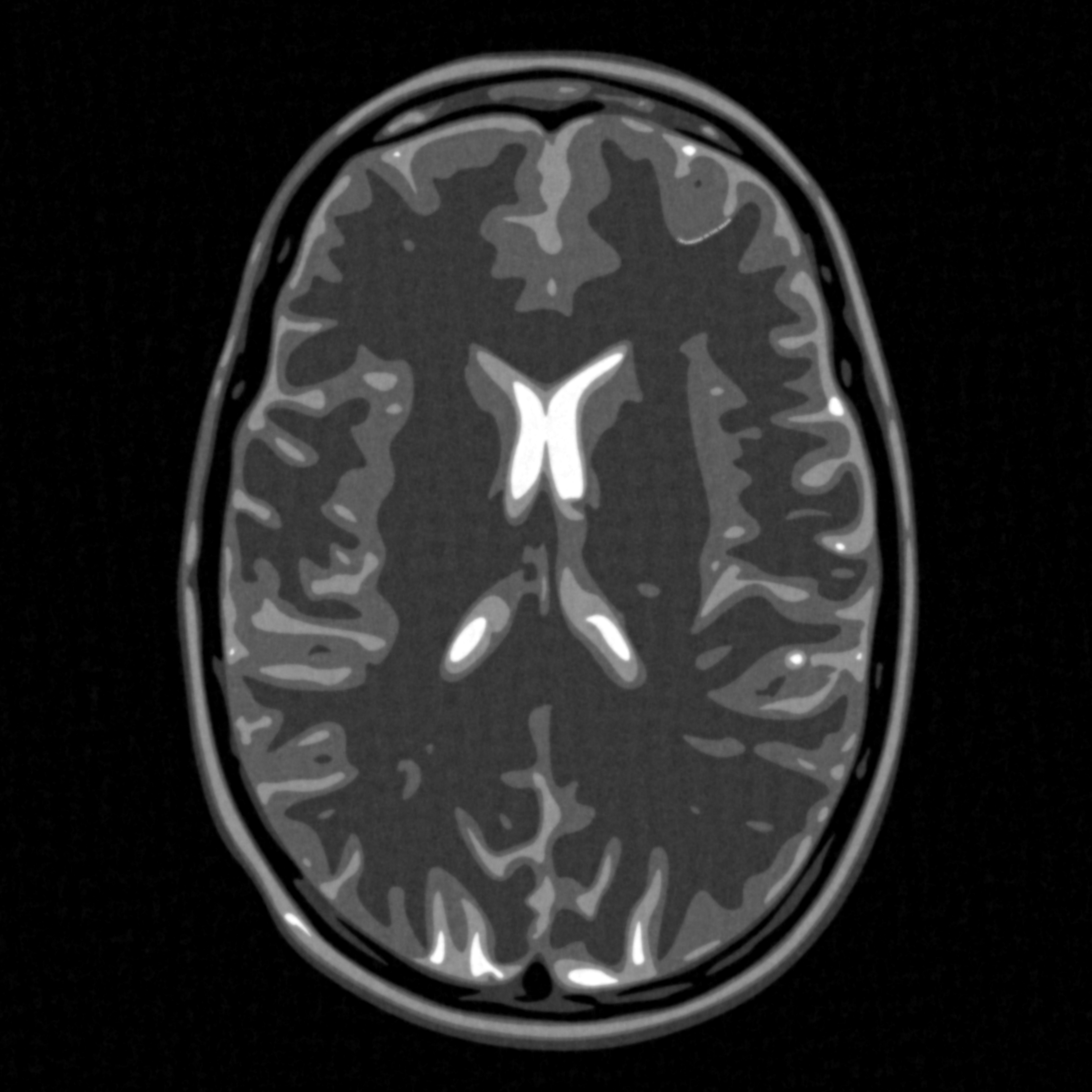} &
\hspace{-0.3cm}\includegraphics[width=0.3\linewidth]{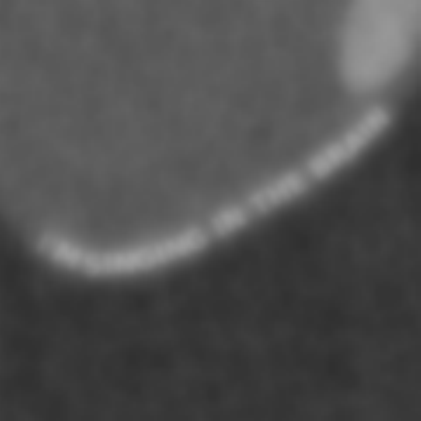} \\
 &  {\hspace{-0.3cm}\small (d) SNR = 21 dB} & {\hspace{-0.3cm}\small (e)}
\etabu
\caption{\label{fig:reconstructionIso} An example of reconstruction of a $2048\times 2048$  MR image from isolated measurements. (a)  Sensing pattern from a variable density sampling strategy (with $4.6\%$ measurements). (b) Corresponding reconstruction via $\ell_1$-minimization. (c) A zoom on a part of the reconstructed image. (d) Image obtained by using the pseudo-inverse transform. (e) A zoom on a part of this image.}
\end{center}
\end{figure}

\subsection{The need for new results}

Probing measurements independently at random is infeasible - or at least impractical - in most measuring instruments. 
This is the case in MRI, where the samples have to lie on piecewise smooth trajectories \cite{lustig2007sparse,chauffert2014,chauffert2016TMI}. 
The same situation occurs in a number of other devices such as Electron \cite{leary2013compressed} and X-ray Tomography \cite{pan2009commercial}, radio-interferometry \cite{wiaux2009compressed}, mobile sensing \cite{taubock2008compressed}, ... 
As a result, concrete applications of CS often rely on sampling schemes that strongly deviate from theory, to account for physical constraints intrinsic to each instrument. 
Despite having no solid theoretical foundation, these heuristic strategies work very well in practice.

This fact is illustrated in Figure \ref{fig:reconstructionLines}. 
In this numerical experiment, parallel lines drawn indepently at random generate a very structured sampling pattern in the Fourier domain, see Figure \ref{fig:reconstructionLines} (a). 
As can be seen in Figure  \ref{fig:reconstructionLines} (b) and (c), this highly structured pattern makes it possible to recover well resolved images using an $\ell^1$-minimization reconstruction. 

To the best of our knowledge, there currently exists no theory able to explain this favorable behavior. The only works dealing with such an acquisition are \cite{polak2012performance,bigot2014analysis}. They assume no structure in the sparsity and we showed in \cite{bigot2014analysis} that structure was crucially needed to explain results such as those in Figure \ref{fig:reconstructionLines}. We will recall this result in Section \ref{subsec:limits}.

\begin{figure}[h!]
\begin{center}
\btabu{@{}ccc}
\includegraphics[width=0.3\linewidth]{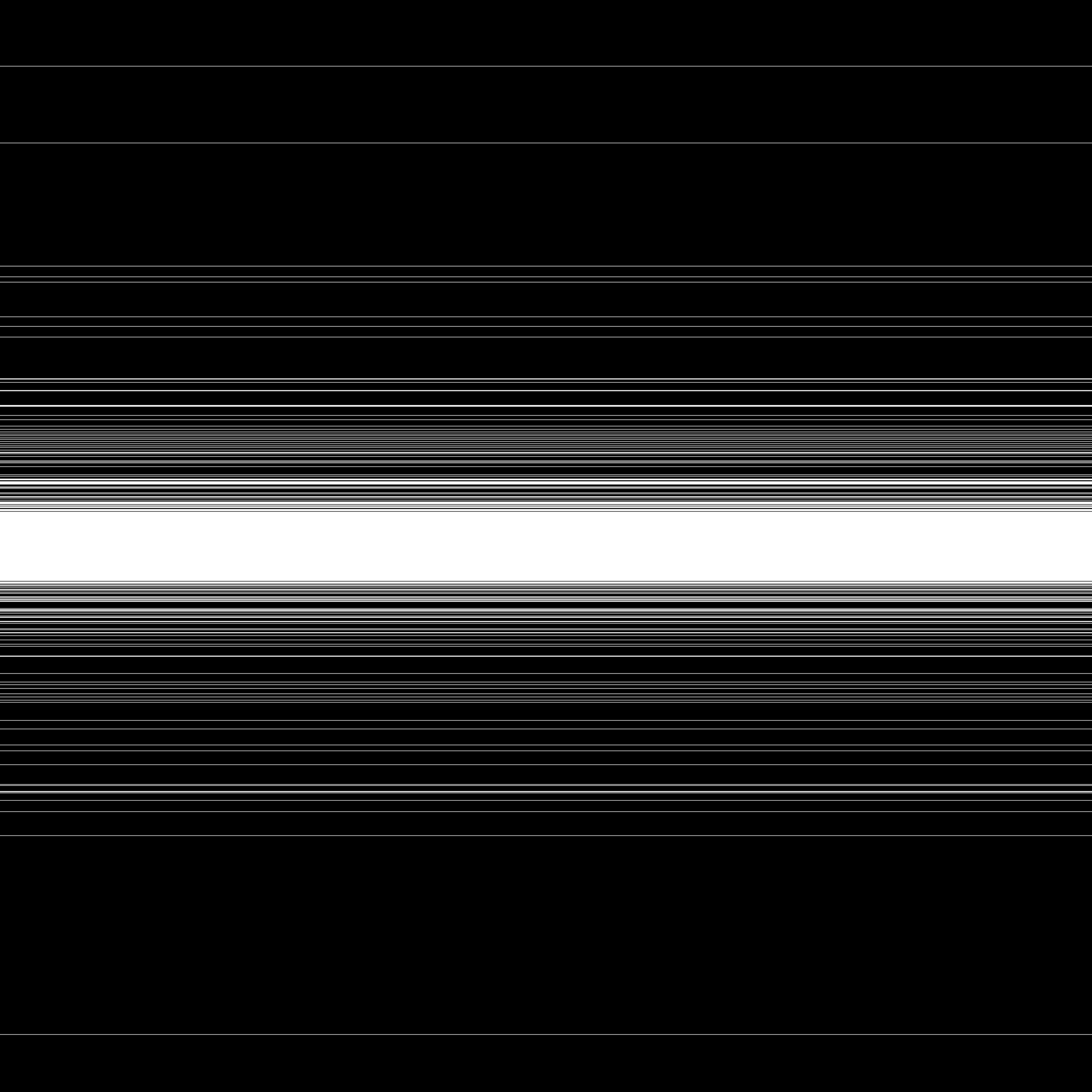} &
\hspace{-0.3cm}\includegraphics[width=0.3\linewidth]{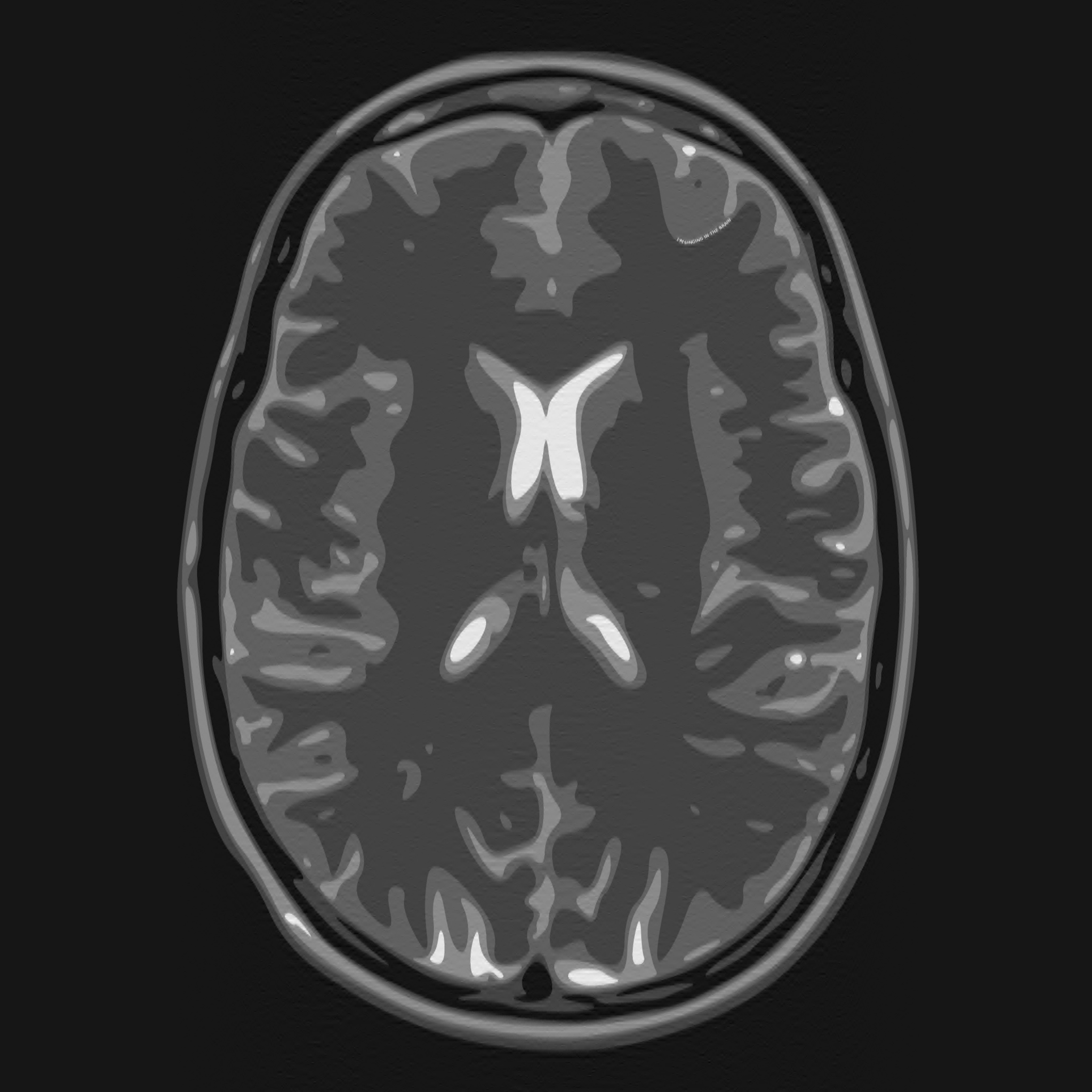} &
\hspace{-0.3cm}\includegraphics[width=0.3\linewidth]{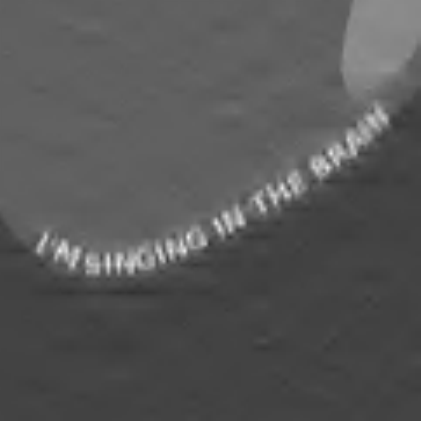} 
\\
{\small (a)}&{\hspace{-0.3cm}\small (b) SNR = 24.1 dB} & {\hspace{-0.3cm}\small (c)} \\
&\hspace{-0.3cm}\includegraphics[width=0.3\linewidth]{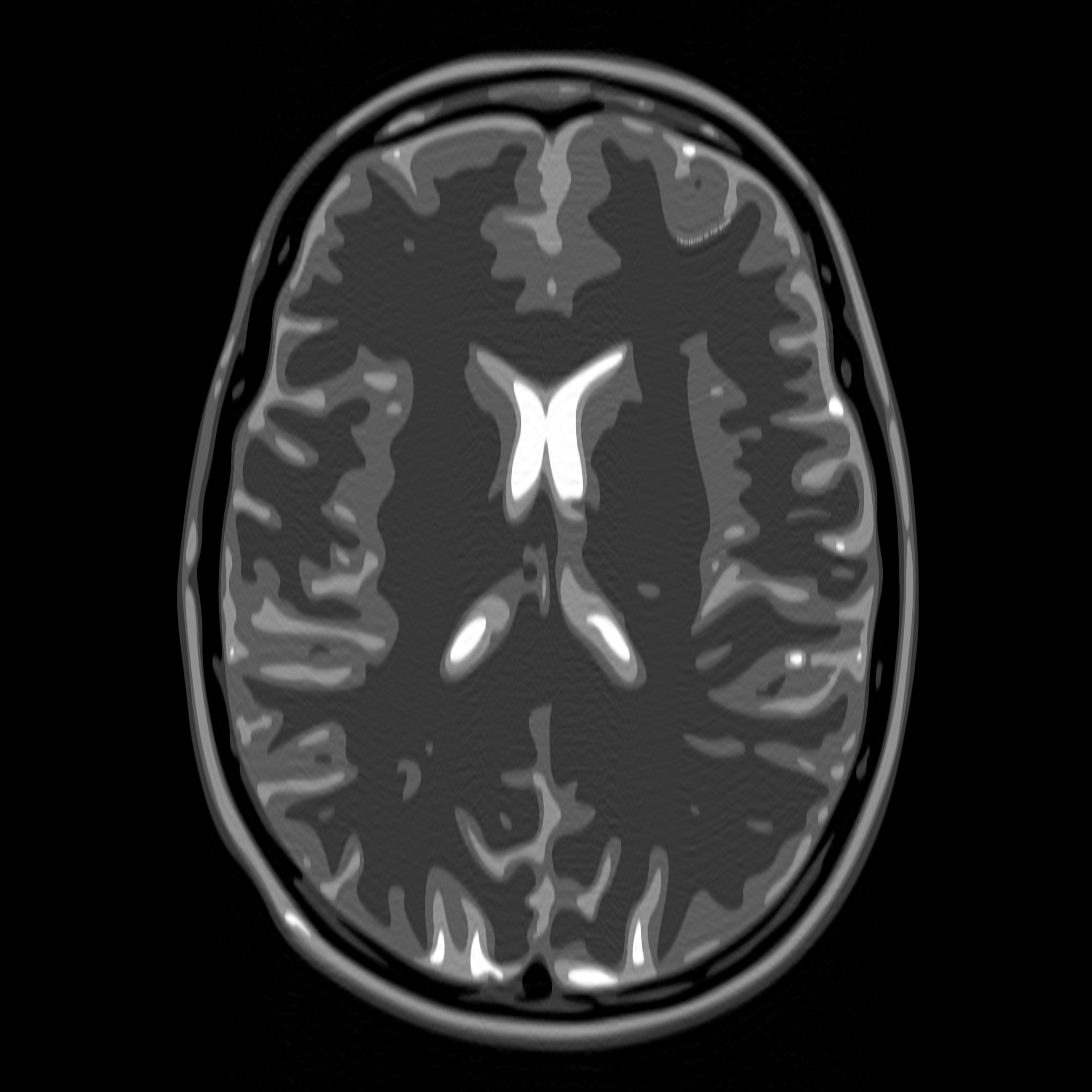} &
\hspace{-0.3cm}\includegraphics[width=0.3\linewidth]{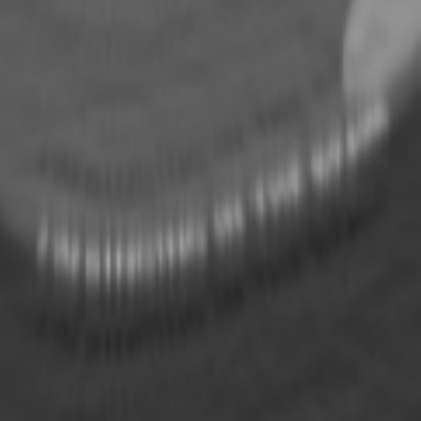} \\
 &  {\hspace{-0.3cm}\small (d) SNR = 21 dB} & {\hspace{-0.3cm}\small (e)}
\etabu
\caption{\label{fig:reconstructionLines} An example of reconstruction of a $2048\times 2048$  MR image from blocks of measurements. (a)  Sampling pattern horizontal lines ($13\%$ of measurements). (b) Corresponding reconstruction via $\ell_1$-minimization. (c) A zoom on a part of the reconstructed image. (d) Image obtained by using the pseudo-inverse transform. (e) A zoom on a part of this image.}
\end{center}
\end{figure}

\subsection{Contributions}

The main contribution of this paper is to derive a new compressed sensing theory:
\begin{enumerate}
\item giving recovery guarantees with an explicit dependency on the support of the vector to reconstruct, 
\item based on block-structured acquisition. 
\end{enumerate}

Informally, our main result (Theorem  \ref{thm:recovery}) reads as follows. 
Let $x\in \Cbb^n$ denote a vector with support $S\subset \{1,\hdots,n\}$. 
Draw $m$ blocks of measurements with a distribution $\pi \in \Rbb^M$, where $M$ denotes the number of available blocks.
If $m\gtrsim \Gamma(S,\pi) \ln\left(\frac{n}{\varepsilon}\right)$, the vector $x$ is recovered by $\ell^1$ minimization with probability greater than $1-\varepsilon$.

The proposed theory has a few important consequences:
\begin{itemize}
    \item The block structure proposed herein enriches the family of sensing matrices available for CS. 
    Existing theories for structured sampling do not take constraints of the sampling device into account. 
    Therefore, the proposed theory gives keys to design realistic structured sampling.
   
    \item Our theorem significantly departs from most works that consider reconstruction of any $s$-sparse vector. 
This is similar in spirit to the works \cite{adcock2015generalized,adcock2013breaking}
However, this is the first time that the dependency on the support $S$ and the drawing probability $\pi$ is made explicit through the quantity $\Gamma(S,\pi)$. 
This provides many possibilities such as optimizing the drawing probability $\pi$ or identifying the classes of supports recoverable with block sampling strategies.

     \item The proposed approach generalizes most existing compressed sensing theories. In particular, it allows recovering all the results mentioned in the introduction. 
     
     \item The provided theory seems to predict accurately practical Fourier sampling experiments, which is quite rare in this field.
     The example given in Figure \ref{fig:reconstructionLines} can be analyzed precisely. In particular, we show that a block structured acquisition can be used, only if the support structure is adapted to it. The resulting structures are more complex than the sparsity by levels of \cite{adcock2013breaking}.
     
  	\item The proposed theory allows envisioning the use of CS in situations that were not possible before. The use of incoherent transforms is not necessary anymore, given that the support $S$ has some favorable properties. 
  	
  	\item The usual restricted isometry constant or coherence are replaced by the quantity $\Gamma(S,\pi)$, which seems to be much more adapted to describe the practical success of CS.
\end{itemize}

\subsection{Related notions in the literature}

In this paper, structured acquisition denotes the constraints imposed by the physics of the acquisition, that are modeled using blocks of measurements extracted from a full deterministic matrix $A_0$. This notion of structured acquisition differs from the notion of structured random matrices, as described in \cite{rauhut2010compressive} and \cite{duarte2011structured}. 
Indeed, this latter strategy is based on acquiring isolated measurements randomly drawn from the rows of a deterministic matrix. The resulting sensing matrix has thus some inherent structure, which is not the case of random matrices with i.i.d.\ entries, that were initially considered in CS.
In our paper, the sensing matrix $A$ is even more structured, in the sense that the full sampling matrix $A_0$ has been partitioned into blocks of measurements.

We also focus on obtaining RIPless results by combining structured acquisition and structured sparsity.
RIPless results \cite{candes2011probabilistic} refer to CS approaches that are non-uniform in the sense that they hold for a given sensing matrix $A$ and a given support $S$ of length $s$, but not for all $s$-sparse vectors. Nevertheless, existing RIPless results in the literature are only based on the degree of sparsity $s=|S|$. 
A main novelty of this paper is to develop RIPless results that explicitly depend on the support $S$ (and not only on its cardinality $s$) of the signal to reconstruct. This strategy allows to incorporate any kind of \textit{prior} information on the structure of $S$ to study its influence on the quality of CS reconstructions. 

Structured sparsity is a concept that appeared early in the history of compressed sensing.
The works \cite{tropp2006just,gribonval2008beyond,herzet2013exact} provide sufficient conditions to recover structured sparse signals by using orthogonal matching pursuit or basis pursuit algorithms. 
Similar conditions (inexact dual certificate) are used in our work.
The main novelty and difficulty in our contribution is to show that very structured sampling matrices satisfy these conditions.

Other authors \cite{eldar2009robust,baraniuk2010model,duarte2011structured,bach2012optimization} proposed to change the recovery algorithm, when a prior knowledge of structured sparsity is available. Their study is usually restricted to random sub-Gaussian matrices which have no structure at all. At this point, we do not know if better recovery guarantees could be obtained by using structured recovery algorithms with structured sampling.

Finally, let us mention that a few papers recently considered the problem of mobile sampling \cite{unnikrishnan2013sampling,unnikrishnan2013signal,grochenig2014minimal}. In these papers, the authors provide theoretical guarantees for the exact reconstruction of bandlimited functions in the spirit of Shannon's sampling theorem. These papers thus strongly differ from our compressed sensing perspective.

\subsection{Organization of the paper}
The paper organization is as follows. Section \ref{sec:setting} gives the formal setting of structured acquisition. Section \ref{sec:mainResults} gives the main results, with a precise definition of $\Gamma(S,\pi)$. Applications of our main theorem to various settings are presented in Section \ref{sec:Applications}. Technical appendices contain the proofs of the main results of this paper.


\section{Preliminaries}
\label{sec:setting}

\subsection{Notation}

In this paper, $n$ denotes the dimension of the signal to reconstruct. 
The notation $S\subset \{1, \hdots, n\}$ refers to the support of the signal to reconstruct.
The vectors $\left( e_i \right)_{1\leq i \leq p}$ denote the vectors of the canonical basis of $\Rbb^d$, where $d$ will be equal to $n$ or $\sqrt{n}$, depending on the context. 
In the sequel, we set $P_S \in \Rbb^{n\times n}$ to be the projection matrix onto $\text{span} \left( \left\lbrace e_i , i\in S \right\rbrace \right)$, i.e.\ the diagonal matrix with the $j$-th diagonal entry equal to 1 if $j\in S$, and 0 otherwise. We will use the shorthand notation $M_S\in \Cbb^{n\times n}$ and $v_S\in \Cbb^n$ to denote the matrix $M P_S$ and the vector $P_S v$ for $M \in \Cbb^{n\times n }$ and $v\in \Cbb^n$.
Similarly, if $M_k$ denotes a matrix indexed by $k$, then $M_{k,S}=M_k P_S$.
For any matrix $M$, for any $1\leq p,q \leq \infty$, the operator norm $\| M \|_{p\rightarrow q}$ is defined as
$$ \| M \|_{p\rightarrow q} = \sup_{\|v\|_p \leq 1} \| Mv \|_q,
$$
with $\|\cdot\|_p$ and $\|\cdot\|_q$ denoting the standard $\ell_p$ and $\ell_q$ norms. Note that for a matrix $M\in \Rbb^{n\times n}$,
$$ \| M \|_{\infty \rightarrow \infty} = \max_{1\leq i \leq n} \| e_i^* M \|_1.
$$
The function $\sgn : \Rbb^n \rightarrow \Rbb^n$ is defined by
\[ \left( \sgn ( x) \right)_i = \left\lbrace
\begin{array}{cc}
1 & \text{if}  \quad x_i >0 \\
-1 & \text{if} \quad  x_i < 0 \\
0 & \text{if} \quad x_i=0,
\end{array}
\right.
\]
and $\Id_n$ will denote the $n$-dimensional identity matrix.

\subsection{Sampling strategy}


In this paper, we assume that we are given some orthogonal matrix $A_{0} \in \Cbb^{n \times n}$ representing the set of possible linear measurements imposed by a specific sensor device. Let $\left( \Ic_k \right)_{1\leq k \leq M}$ denote a partition of the set $\{1, \hdots , n \}$. The rows $(a_i^*)_{1\leq i \leq n} \in \Cbb^n$ of $A_0$ are partitioned into the following blocks dictionary $\left( B_k \right)_{1 \leq k \leq M}$, such that
$$ 
B_k = \left( a_i^* \right)_{i\in \Ic_k} \in \Cbb^{|\Ic_k| \times n} \qquad {s.t.} \qquad \Ic_k \subset \{1, \hdots , n \},
$$
with $\sqcup_{k=1}^M \Ic_k = \{1,\hdots , n\}$. 
The sensing matrix $A$ is then constructed by randomly drawing blocks as follows
\begin{equation}
\label{eq:modsamp}
A = \frac{1}{\sqrt{m}} \left( \frac{1}{\sqrt{\pi_{K_\ell}}} B_{K_\ell} \right)_{1\leq \ell \leq m}, 
\end{equation}
where $\left( K_\ell \right)_{1\leq \ell \leq m}$ are i.i.d.\ copies of a random variable $K$ such that
$$
\Pbb(K = k) = \pi_{k},
$$ 
for all $ 1 \leq k \leq M$.
Moreover, thanks to the renormalization of the blocks $B_{K_\ell}$  by the weights $1/\sqrt{\pi_{K_\ell}}$ in model \eqref{eq:modsamp}, the random block $ B_{K} $ satisfies
\begin{align}
\label{isotropyCondition}
\Ebb\left( \frac{B_{K}^*B_K}{\pi_K} \right)= \sum_{k=1}^M B_k^*B_k = \Id,
\end{align}
since $A_0$ is orthogonal and $\left( B_k\right)_{1\leq k \leq M}$ is a partition of the rows of $A_0$. 


\begin{remark}
The case of overlapping blocks can also be handled. 
To do so, we may define the blocks $\left( B_k \right)_{1\leq k \leq M}$ as follows:
$$ B_k = \left( \frac{1}{\sqrt{\alpha_i}} a_i^* \right)_{i \in \Ic_k},  \qquad \text{for} \quad 1\leq k \leq M,
$$
where $\displaystyle \bigcup_{k=1}^M \Ic_k = \{1, \hdots  , n\}$. The coefficients $\left( \alpha_i\right)_{1\leq i \leq n}$ denotes the multiplicity of the row $a_i^*$, namely the number of appearances $\alpha_i = |\{k, i \in \Ic_k \}|$ of this row in different blocks.
This renormalization is sufficient to ensure the isotropy condition $\Ebb \left( \frac{B_{K}^*B_K}{\pi_K} \right) = \Id$ where $K$ is defined as above. 
\end{remark}

Note that our block sampling strategy encompasses the standard acquisition based on isolated measurements. Indeed, isolated measurements can be considered as blocks of measurements consisting of only one row of $A_0$.

\begin{remark}
More generally, the theorems could be extended - with slight adaptations - to the case where the sensing matrix is
$$
A= \frac{1}{\sqrt{m}}
  \begin{pmatrix}
B_{K_1} \\ \vdots \\ B_{K_m}  
   \end{pmatrix}
$$
where $B_{K_1},\hdots, B_{K_m}$ are i.i.d.\ copies of a random matrix $B\in \Cbb^{b \times n}$ satisfying 
$$
\Ebb(B^*B)=\Id.
$$
The integer $b$ is itself random and $\Id$ is the $n\times n$ identity matrix. 
Assuming that $B$ takes its value in a countable family $\left( B_k  \right)_{k\in \Kc}$, this formalism covers a large number of applications described in \cite{bigot2014analysis}: (i) blocks with i.i.d.\ entries, (ii) partition of the rows of orthogonal transforms, (iii) cover of the rows of orthogonal transforms, (iv) cover of the rows from tight frames.
\end{remark}


\section{Main Results}
\label{sec:mainResults}

\subsection{Fundamental quantities}

Before introducing our main results, we need to define some quantities (reminiscent of the coherence) that will play a key role in our analysis.

\begin{defchapter}
\label{def:quantities}
Consider a blocks dictionary $\left( B_k\right)_{1\leq k \leq M}$. Let $S\subset \{1, \hdots, n \}$ and $\pi$ be a probability distribution on $\{1,\hdots ,M\}$.
Define
\begin{align}
\label{ineq:Theta}
 \Theta(S,\pi) &:= \max_{1\leq k \leq M}  \frac{1}{\pi_k} \| B_{k}^*B_{k,S}  \|_{\infty \rightarrow \infty} = \max_{1\leq k \leq M} \max_{1\leq i \leq n} \frac{ \| e_i^* B_k^* B_{k,S} \|_1}{\pi_k},
 \\
\label{ineq:Upsi}
 \Upsilon(S,\pi)  &:=\max_{1\leq i \leq n} \sup_{\|v\|_\infty \leq 1 }\sum_{k=1}^M   \frac{1}{\pi_k}  \left| e_i^*  B_{k}^* B_{k,S} v  \right|^2,\\ 
\label{Gamma}
 \Gamma(S,\pi) &:= \max \left( \Upsilon(S,\pi) , \Theta(S,\pi)  \right).
\end{align}
\end{defchapter}
For the sake of readability, we will sometimes use the shorter notation $\Theta, \Upsilon$ and $\Gamma$ to denote  $\Theta(S,\pi), \Upsilon(S,\pi)$ and $\Gamma(S,\pi)$.
In Definition \ref{def:quantities}, $\Theta$ is related to the local coherence and the degree of sparsity, when the blocks are made of only one row (the case of isolated measurements). Indeed, in such a case, $\Theta$ reads as follows
$$ \Theta(S,\pi) := \max_{1\leq k \leq n} \frac{\| a_k \|_\infty \| a_{k,S} \|_1}{\pi_k} \leq s\cdot \max_{1\leq k \leq n} \frac{\| a_k \|_\infty^2 }{\pi_k}.
$$
The quantity $ \max_{1 \leq k \leq n} \frac{\| a_k \|_\infty^2 }{\pi_k}$ refers to the usual notion of coherence described in \cite{candes2011probabilistic}.
The quantity $\Upsilon$ is new and it is more delicate to interpret. 
It reflects an inter-block coherence.
A rough upper-bound for $\Upsilon$ is
$$  \Upsilon(S,\pi)  \leq \sum_{k=1}^M   \frac{1}{\pi_k}  \left\| B_{k}^* B_{k,S}   \right\|^2_{\infty \rightarrow \infty}.
$$
by switching the maximum and supremum with the sum in the definition of $\Upsilon$.
However, it is important to keep this order (maximum, supremum and sum) to measure interferences between blocks. 
In Section \ref{sec:Applications}, we give more precise evaluations of $\Theta(S,\pi)$ and $\Upsilon(S,\pi)$ in particular cases.

\begin{remark}[Support-dependency and drawing-dependency]
In Definition \ref{def:quantities}, the quantities $\Theta$ and $\Upsilon$ are drawing-dependent and support-dependent. Indeed, $\Gamma$ does not only depend on the degree of sparsity $s = |S|$. To the best of our knowledge, existing theories in CS only rely on $s$, see \cite{candes2006robust, candes2011probabilistic}, or on degrees of sparsity structured by levels, see \cite{adcock2013breaking}. Since $\Gamma$ is explicitly related to $S$, this allows to incorporate \textit{prior} assumptions on the structure of $S$. Besides, the dependency on $\pi$ (i.e.\ the way of drawing the measurements) is also explicit in the definition of $\Gamma$.  This offers the flexibility to analyze the influence of $\pi$ on the required number of measurements. We therefore believe that the introduced quantities might play an important role in the future analysis of CS.  
\end{remark}

\subsection{Exact recovery guarantees}

Our main result reads as follows.
\begin{thmchapter}
\label{thm:recovery}
Let $S \subset \{ 1 ,\hdots , n \}$ be a set of indices of cardinality $s\geq 16$  and suppose that $x \in \Cbb^{n}$ is an $s$-sparse vector supported on $S$.  Fix $\varepsilon \in (0,1)$. Suppose that the sampling matrix $A$ is constructed as in \eqref{eq:modsamp}.  Suppose that ${\Gamma(S ,\pi)} \geq 1$. If 
\begin{eqnarray}\label{eq:mainthm}
m \geq 73 \cdot {\Gamma(S ,\pi)} \ln(64s) \left(   \ln \left( \frac{9 n }{\varepsilon}\right) + \ln \ln(64s) \right),
\end{eqnarray}
then $x$ is the unique solution of \eqref{pb:minL1}  with probability larger than $1-\varepsilon$. 
\end{thmchapter}

\begin{remark}
In the sequel, we will simplify condition \eqref{eq:mainthm} by writing:
\begin{equation*}
 m\geq C \cdot \Gamma(S ,\pi) \ln(s)\ln \left( \frac{n }{\varepsilon}\right)
\end{equation*}
where $C$ is a universal constant.
\end{remark}

The proof of Theorem \ref{thm:recovery} is contained in Appendix\ref{sec:proof1}. It relies on the construction of an inexact dual certificate satisfying appropriate properties that are described in Lemma \ref{lem:inexactDuality}. Then our proof is based on the so-called golfing scheme introduced in \cite{gross2011recovering} for matrix completion and adapted by \cite{candes2011probabilistic} for compressed sensing from isolated measurements. 
In the golfing scheme, the main difficulty is to control operator norms of random matrices extracted from the sensing matrix $A$. 
In \cite{candes2011probabilistic}, it is proposed to control (in probability) the operator norms $\|\cdot \|_{\infty \rightarrow 2}$ and $\| \cdot \|_{2\rightarrow 2}$. However, this technique only gives results depending on the degree of sparsity $s$. 
In order to include an explicit dependency on the support $S$, one has to modify the golfing scheme in \cite{candes2011probabilistic}, by controlling the operator norm $\| \cdot \|_{\infty \rightarrow \infty}$, instead of  controlling the operator norms $\|\cdot \|_{\infty \rightarrow 2}$ and $\| \cdot \|_{2\rightarrow 2}$. 
A similar idea has been developed in \cite{adcock2013breaking}.

\begin{remark}
\label{rem:balancingProperty}
Compared to most compressed sensing results, the condition required in Theorem \ref{thm:recovery} involves the extra multiplicative factor $\ln(64s)$. This factor does not appear in \cite{gross2011recovering,candes2011probabilistic}, but this is due to a mistake that was detected and corrected in \cite{adcock2015generalized}. 
Following the proofs proposed in \cite{adcock2015generalized}, we could in fact obtain a bound of type:
\begin{equation*}
 m\geq C' \cdot \Gamma(S ,\pi) \ln \left( \frac{n }{\varepsilon}\right),
\end{equation*}
with $C'>C$. To the best of our knowledge, the ratio $C'/C$ obtained using the proof in \cite{adcock2015generalized} is of order $24$.
This means that the new bound becomes interesting only for $s>4\cdot 10^8$, i.e. in an asymptotic regime. 
In this paper, we therefore stick to the bound in Theorem \ref{thm:recovery} for (i) simplifying the proof of the main result and (ii) obtain the best results in a non asymptotic regime.
\end{remark}


\subsection{Consequences for stochastic signal models}

The explicit dependency of $\Gamma$ in $S$ allows us to consider the case of a random support $S$. 
\begin{prop}
\label{prop:stochasticModelS}
Let $S\subset \{1,\hdots , n\}$ denote a random support. For some real positive $\gamma$, suppose that the event $\Gamma(S,\pi)\leq \gamma$ occurs with probability larger than $1-\varepsilon'(\gamma)$. If $ m \gtrsim \gamma \ln(s) \ln(n/\varepsilon)$, then $x$ is the unique solution of Problem \ref{pb:minL1} with probability higher than $1-\varepsilon-\varepsilon \varepsilon'(\gamma)$.
\end{prop}
\begin{proof}
Set $m \gtrsim \gamma \ln(s) \ln(n/\varepsilon)$. Define the event $R$ ``$x$ is the unique solution of Problem \ref{pb:minL1}" where $R$ stands for ``reconstruction of the signal". Define also $A$ the event ``$\Gamma(S,\pi)\geq \gamma$". The hypothesis of Proposition \ref{prop:stochasticModelS} and Theorem \ref{thm:recovery} give that $\Pbb\left( R | A \right) \geq 1-\varepsilon$. To prove Proposition \ref{prop:stochasticModelS}, we must quantify 
\begin{align*}
\Pbb\left( R  \right) &= \Pbb\left( R \cap A \right)+  \Pbb\left( R \cap A^c \right) = \Pbb \left( R| A\right) \Pbb(A) + \Pbb\left(R\cap A^c \right) \\
&\geq (1-\varepsilon)\left( 1-\varepsilon'(\gamma)\right) = 1-\varepsilon-\varepsilon \varepsilon'(\gamma),
\end{align*}
which concludes the proof.
\end{proof}

\subsection{Choice of the drawing probability}

The choice of a drawing probability $\pi$ minimizing the required number of block measurements in Theorem \ref{thm:recovery},  is a delicate issue. The distribution $\pi^\star$ minimizing $\Theta(S,\pi)$ in Equation \eqref{ineq:Theta} can be obtained explicitly:
\begin{align}
\label{eq:choicePi}
 \pi^\star_k = \frac{\| B_{k}^*B_{k,S}  \|_{\infty \rightarrow \infty}  }{\sum_{\ell=1}^M \| B_{\ell}^*B_{\ell,S}  \|_{\infty \rightarrow \infty}  }, \quad \text{for} \quad 1\leq k \leq M.
\end{align}
Unfortunately, the minimization of $\Upsilon(S,\pi)$ with respect to $\pi$ seems much more involved and we leave this issue as an open question in the general case. 

Note however that in all the examples treated in the paper, we derive upper bounds depending on $(S,\pi)$ for $\Upsilon(S,\pi)$ and $\Theta(S,\pi)$ that coincide. The distribution $\pi^\star$ is then set to minimize the latter upper bound.

Note also that optimizing $\pi$ independently of $S$ will result in a sole dependence to the degree of sparsity $s=|S|$ which is not desirable if one wants to exploit structured sparsity. 


\section{Applications}

\label{sec:Applications}

In this section, we first show that Theorem \ref{thm:recovery} can be used to recover state of the art results in the case of isolated measurements \cite{candes2011probabilistic}. 
We then show that it allows recovering recent results when a prior on the sparsity structure is available.  
The proposed setting however applies to a wider setting even in the case of isolated measurements. 
Finally, we illustrate the consequences of our results when the acquisition is constrained by blocks of measurements. In the latter case, we show that the sparsity structure should be adapted to the sampling structure for exact recovery.

\subsection{Isolated measurements with arbitrary support}

First, we focus on an acquisition based on isolated measurements which is the most widespread in CS. 
This case corresponds to choose blocks of form $B_k = a_k^*$ for $1\leq k \leq n$ with $M=n$, where $a_k^*$ are the rows of an orthogonal matrix.
In such a setting, the sensing matrix can be written as follows
\begin{align}
\label{eq:sensingMatrixIsolated}
A = \frac{1}{\sqrt{m}}\left( \frac{1}{\sqrt{\pi_{K_\ell}}} a_{K_\ell}^* \right)_{1\leq \ell \leq m},
\end{align}
where $\left( K_\ell \right)_{1\leq \ell \leq m}$ are i.i.d.\ copies of $K$ such that $\Pbb \left( K = k\right) = \pi_k$, for $1\leq k \leq n$.

We apply Theorem \ref{thm:recovery} when only the degree of sparsity $s$ of the signal to reconstruct is known. 
This is the setting considered in most CS papers (see e.g. \cite{candes2006near,rauhut2010compressive,candes2011probabilistic}).
In this context, our main result can be rewritten as follows.
\begin{corollary}
\label{corol:isolatedUniform}
Let $S \subset \{ 1 ,\hdots , n \}$ be a set of indices of cardinality $s$  and suppose that $x \in \Cbb^{n}$ is an $s$-sparse vector.  Fix $\varepsilon \in (0,1)$. Suppose that the sampling matrix $A$ is constructed as in \eqref{eq:sensingMatrixIsolated}.  If 
\begin{eqnarray}\label{eq:bbbbbbb}
m &\geq & C \cdot s \cdot  \max_{1\leq k \leq n} \frac{ \|a_k \|_\infty^2}{\pi_k} \ln(s) \ln \left( \frac{ n }{\varepsilon}\right),
\end{eqnarray}
then $x$ is the unique solution of \eqref{pb:minL1}  with probability at least $1-\varepsilon$. 

Moreover, the drawing distribution minimizing \eqref{eq:bbbbbbb} is $\pi_k = \frac{  \|a_k \|_\infty^2 }{\sum_{\ell=1}^n \|a_\ell \|_\infty^2 }$, which leads to
\begin{eqnarray*}
m &\geq & C \cdot s  \cdot \sum_{k=1}^n \|a_k \|_\infty^2 \ln(s) \ln \left( \frac{n }{\varepsilon}\right).
\end{eqnarray*}
\end{corollary}
The proof is given in Appendix \ref{app:proofIsolatedCandes}.

Note that Corollary \ref{corol:isolatedUniform} is identical to Theorem 1.1 in \cite{candes2011probabilistic} up to a logarithmic factor. 
This result is usually used to explain the practical success of variable density sampling. 
It is the core of papers such as \cite{puy2011variable,krahmer2013stable,chauffert2014}.


\subsection{Isolated measurements with structured sparsity}

When using coherent transforms, meaning that the term $\max_{1\leq k \leq n} \frac{ \|a_k \|_\infty^2}{\pi_k}$ in Equation \eqref{eq:bbbbbbb} is an increasing function of $n$, Corollary \ref{corol:isolatedUniform} is unsufficient to justify the use of CS in applications. 
In this section, we show that the proposed results allow justifying the use of CS even in the extreme case where the sensing is performed with the canonical basis.

 \subsubsection{A toy example: sampling isolated measurements from the Identity matrix and knowing the support $S$}
 
Suppose that the signal $x$ to reconstruct is $S$-sparse where $S\subseteq \{1, \hdots , n \}$ is a fixed subset. Consider the highly coherent case where $A_0=\Id$. All current CS theories would give the same unsatisfactory conclusion: it is not possible to use CS since $A_0$ is a perfectly coherent transform. 
Indeed, the bound on the required number of isolated measurements given by standard CS theories \cite{candes2011probabilistic} reads as follows
$$ m \geq C\cdot  s \cdot\max_{1\leq k \leq n} \frac{\| e_k^* \|_\infty^2}{\pi_k} \ln\left(n  / \varepsilon\right) = C \cdot s \cdot \max_{1\leq k \leq n} \frac{1}{\pi_k} \ln\left(n / \varepsilon\right) .
$$
Without any assumption on the support $S$, one can choose to draw the measurements uniformly at random, i.e. $\pi_k = 1/n$ for $1\leq k \leq n$. This particular choice leads to a required number of measurements of the order
$$ m \geq C \cdot  s \cdot n \ln\left(n  / \varepsilon\right),
$$
which corresponds to fully sampling the acquisition space several times. 

Let us now see what conclusion can be drawn with Theorem \ref{thm:recovery}.
\begin{corollary}
\label{corol:isolatedIdentity}
Let $S \subseteq \{ 1 ,\hdots , n \}$ of cardinality $s$.  Suppose that $x \in \Cbb^{n}$ is an $S$-sparse vector.  
Fix $\varepsilon \in (0,1)$. 
Suppose that the sampling matrix $A$ is constructed as in \eqref{eq:sensingMatrixIsolated} with $A_0=\Id$.  
Set $\pi_k = \frac{\delta_{k,S}}{s}$ for $1\leq k \leq n$ where $\delta_{k,S} =1$ if $k\in S$, $0$ otherwise. 
Suppose that
\begin{eqnarray*}
m &\geq & C\cdot s \cdot \ln(s) \ln \left( \frac{n }{\varepsilon}\right).
\end{eqnarray*}
then $x$ is the unique solution of \eqref{pb:minL1}  with probability at least $1-\varepsilon$. 
\end{corollary}

With this new result, $O(s \ln(s) \ln(n))$ measurements are sufficient to reconstruct the signal via a totally coherent. 
The least amount of measurements necessary to recover $x$ is of order $O(s \ln(s))$, by an argument of coupon collector effect \cite[p.262]{feller2008introduction}.
Therefore, Corollary \ref{corol:isolatedIdentity} is near-optimal up to logarithmic factors.

\begin{proof}
The result ensues from a direct evaluation of $\Gamma$. 
Indeed,
\begin{align*}
\|e_k e_{k,S}^* \|_{\infty \rightarrow \infty} = \max_{1\leq i \leq n} \sup_{\| v\|_\infty \leq 1} | \left\langle e_i , e_k e_{k,S}^* v \right\rangle |  = \sup_{\| v\|_\infty \leq 1} |e_{k,S}^* v| = \delta_{k,S},
\end{align*}
where $\delta_{k,S}=1$ if $k\in S$, $0$ otherwise. 
Therefore
$$ 
\Theta = \max_{1\leq k \leq n} \frac{\delta_{k,S}}{\pi_k}.
$$
Then, we can write that
\begin{align*}
\Upsilon(S,\pi) &= \max_{1\leq i \leq n}\sup_{\|v\|_\infty\leq 1} \sum_{k=1}^n \frac{1}{\pi_k} | e_i^* e_k e_{k,S}^* v |^2 = \max_{1\leq i \leq n} \sup_{\|v\|_\infty\leq 1} \frac{|e_{i,S}^* v |^2 }{\pi_i}\\
&= \max_{1\leq i \leq n} \frac{\delta_{i,S}}{\pi_i}.
\end{align*}
To conclude the proof it suffices to apply Theorem \ref{thm:recovery}.
\end{proof}

\subsubsection{Isolated measurements when the degree of sparsity is structured by levels}
 
In this part, we consider a partition of $\{1 , \hdots ,n \}$ into levels 
$\left( \Omega_i \right)_{i=1,\hdots, N} \subset \{1 , \hdots ,n \} $ such that 
$$ 
\mathop{\bigsqcup}_{1\leq i\leq N} \Omega_i = \{1 , \hdots ,n \} \qquad \text{and} \qquad | \Omega_i | = N_i.
$$
We consider that $x$ is $S$-sparse with $|S\cap \Omega_i | = s_i$ for $1\leq i \leq N$ meaning that restricted to the level $\Omega_i$, the signal $P_{\Omega_i}x$ is $s_i$-sparse.
This setting is studied extensively in the recent papers \cite{adcock2013breaking,roman2014asymptotic,bastounis2014absence}. 
Theorem \ref{thm:recovery} provides the following guarantees.

\begin{corollary}
\label{corol:isolatedBlocksSparsity}
Let $S \subset \{ 1 ,\hdots , n \}$ be a set of indices of cardinality $s$, such that $|S\cap \Omega_i | = s_i$ for $1\leq i \leq N$.  
Suppose that $x \in \Cbb^{n}$ is an $S$-sparse vector.  
Fix $\varepsilon \in (0,1)$. 
Suppose that the sampling matrix $A$ is constructed as in \eqref{eq:sensingMatrixIsolated}.  
Set 
\begin{eqnarray}
m &\geq & C \left(\max_{1\leq k \leq n} \frac{\sum_{\ell=1}^N s_\ell \|a_{k,\Omega_\ell} \|_\infty \|a_k \|_\infty}{\pi_k}  \right)\ln(s) \ln \left( \frac{n }{\varepsilon}\right), \label{eq:bbnn} \\
m &\geq & C \left(\max_{1\leq i \leq n}\sup_{\|v\|_\infty \leq 1 }\sum_{k=1}^n   \frac{1}{\pi_k}  \left| e_i^*  a_{k}\right|^2 \left| a_{k,S}^* v  \right|^2\right) \ln(s) \ln \left( \frac{ n }{\varepsilon}\right), \label{eq:bbmm}
\end{eqnarray}
then $x$ is the unique solution of \eqref{pb:minL1}  with probability at least $1-\varepsilon$. 
\end{corollary}

The proof of Corollary \ref{corol:isolatedBlocksSparsity} is given in Appendix \ref{app:proofIsolatedBlocksSparsity}. We show in Appendix \ref{app:comparisonAdcock} that a simple analysis leads to results that are nearly equivalent to those in \cite{adcock2013breaking}. 
It should be noted that the term $\frac{\|a_{k,\Omega_\ell} \|_\infty \|a_k \|_\infty}{\pi_k} $ is related to the notion of local coherence defined in \cite{adcock2013breaking}.
There are however a few differences making our approach potentially more interesting in the case of isolated measurements:
\begin{itemize}
 \item Our paper is based on i.i.d.\ sampling with an arbitrary drawing distribution. 
This leaves a lot of freedom for generating sampling patterns and optimizing the probability $\pi$ in order to minimize the upper-bounds \eqref{eq:bbnn} and \eqref{eq:bbmm}. 
In contrast, the results in \cite{adcock2013breaking} are based on uniform Bernoulli sampling over \emph{fixed} levels. The dependency on the levels is not explicit and it therefore seems complicated to optimize them.
 \item We can deal with a fixed support $S$, which enlarges the possibilities for structured sparsity. It is also possible to consider random supports as explained in Proposition \ref{prop:stochasticModelS}.
\end{itemize}

  \subsubsection{Isolated measurements for the Fourier-Haar transform}
 \label{subsec:isoMRI}

The bounds in Corollary \ref{corol:isolatedBlocksSparsity} are rather cryptic. 
They have to be analyzed separately for each sampling strategy.
To conclude the discussion on isolated measurements, we provide a practical example with the 1D Fourier-Haar system. 

We set $A_0 = \Fc \phi^*$, where $\Fc \in \Cbb^{n\times n}$ is the 1D Fourier transform and $\phi^* \in \Cbb^{n\times n}$ is the 1D inverse wavelet transform. 
To simplify the notation, we assume that $n=2^J$ and we decompose the signal at the maximum level $J = \log_2 (n) -1$.
In order to state our result, we introduce a dyadic partition $\left( \Omega_j \right)_{0\leq j \leq J}$ of the set $\{ 1 , \hdots , n \}$. We set $\Omega_0 = \{1\}, \, \Omega_{1} = \{2\}, \, \Omega_{3} = \{ 3, 4 \}, \hdots , \, \Omega_J =  \{ n/2+1 , \hdots ,  n\}$.
We also define the function $j : \{ 1 , \hdots , n \} \rightarrow \{ 0 , \hdots , J\}$ by $j(u) = j$ if $u\in \Omega_{j}$.
\begin{corollary}
\label{corol:isolatedMRI}
Let $S \subset \{ 1 ,\hdots , n \}$ be a set of indices of cardinality $s$, such that $|S\cap \Omega_j | = s_j$ for $0\leq j \leq J$.  Suppose that $x \in \Cbb^{n}$ is an $s$-sparse vector supported on $S$.  Fix $\varepsilon \in (0,1)$. Suppose that $A$ is constructed from the Fourier-Haar transform $A_0$.  
Choose $\pi_k$ to be constant by level, i.e. $\pi_k = \tilde{\pi}_{j(k)}$.
If
\begin{eqnarray}\label{eq:bboo}
m &\geq & C \cdot \max_{0\leq j \leq J} \frac{1}{\tilde{\pi}_j} 2^{-j} \sum_{p=0}^J 2^{-|j-p|/2} s_p \cdot \ln(s) \ln \left( \frac{ n }{\varepsilon}\right),
\end{eqnarray}
then $x$ is the unique solution of \eqref{pb:minL1}  with probability at least $1-\varepsilon$. 

In particular, the distribution minimizing \eqref{eq:bboo} is
$$\tilde{\pi}_j = \frac{2^{-j} \sum_{p=0}^J 2^{-|j-p|/2} s_p}{\sum_{\ell=1}^n 2^{-j(\ell)} \sum_{p=0}^J 2^{-|j(\ell)-p|/2} s_p},$$ 
which leads to 
\begin{eqnarray}\label{eq:bbbooo}
m &\geq & C \cdot \sum_{j=0}^J \left( s_j +  \sum_{p=0 \atop p\neq j}^J 2^{-|j-p|/2} s_p \right) \cdot  \ln(s) \ln \left( \frac{ n }{\varepsilon}\right).
\end{eqnarray}
\end{corollary}
The proof is presented in Section \ref{app:proofIsolatedMRI}. This corollary is once again similar to the results in \cite{adcock2014quest}. 
The number of measurements in each level $j$ should depend on the degree of sparsity $s_j$ but also on the degree of sparsity of the other levels which is more and more attenuated when the level is far away from the $j$-th one.

\begin{remark}
The Fourier-Wavelet system is coherent and the initial compressed sensing theories cannot explain the success of sampling strategies with such a transform. 
To overcome the coherence, two strategies have been devised.
The first one is based on variable density sampling (see e.g. \cite{puy2012spread,chauffert2014,krahmer2013stable}).
The second one is based on variable density sampling and an additional structured sparsity assumption (see e.g. \cite{adcock2013breaking} and Corollary \ref{corol:isolatedMRI}). 
First, note that the results obtained with the latter approach allow recovering signal with arbitrary supports. Indeed, $\displaystyle\sum_{j=0}^J s_j +  \sum_{p=0 \atop p\neq j}^J 2^{-|j-p|/2} s_p  \leq 2 s$.

Second, it is not clear yet - from a theoretical point of view - that the structure assumption allows obtaining better guarantees.
Indeed, it is possible to show that the sole variable density sampling leads to perfect reconstruction from $m\propto s \ln(n)^2$ measurements, which is on par with bound \eqref{eq:bbbooo}. It will become clear that structured sparsity is essential when using the Fourier-Wavelet systems with structured acquisition. Morever, the numerical experiments led in \cite{adcock2014quest} let no doubt about the fact that structured sparsity is essential to ensure good reconstruction with a low number of measurements.
\end{remark}

\subsection{Structured acquisition and structured sparsity}
 \label{sec:MRI}

In this paragraph, we illustrate how Theorem \ref{thm:recovery} explains the practical success of structured acquisition in applications. We will mainly focus on the 2D setting: the vector $x\in \Cbb^n$ to reconstruct can be seen as an image of size $\sqrt{n}\times \sqrt{n}$.

\subsubsection{The limits of structured acquisition}
\label{subsec:limits} 
 
In \cite{bigot2014analysis,polak2012performance}, the authors provided theoretical CS results when using block-constrained acquisitions. 
Moreover, the results in \cite{bigot2014analysis} are proved to be tight in many practical situations.
Unfortunately, the bounds on the number of blocks of measurements necessary for perfect reconstruction are incompatible with a faster acquisition. 

To illustrate this fact, let us recall a typical result emanating from \cite{bigot2014analysis}. 
It shows that the recovery of sparse vectors with an arbitrary support is of little interest when sampling lines of tensor product transforms. This setting is widely used in imaging. It corresponds to the MRI sampling strategy proposed in \cite{lustig2007sparse}.
\begin{prop}[\cite{bigot2014analysis}]
\label{prop:tensorlimit}
Suppose that $A_0 = \phi \otimes \phi \in \Cbb^{n\times n}$ is a 2D separable transform, where $\phi \in \Cbb^{\sqrt{n}\times \sqrt{n}}$ is an orthogonal transform. 
Consider blocks of measurements made of $\sqrt{n}$ horizontal lines in the 2D acquisition space, i.e. for $1\leq k \leq \sqrt{n}$
$$
B_k = \left( \phi_{k,1} \phi ,  \hdots , \phi_{k,\sqrt{n} } \phi \right).
$$

If the number of acquired lines $m$ is less than $\min(2s,\sqrt{n})$, then there exists no decoder $\Delta$ such that $ \Delta(A x ) = x$ for all $s$-sparse vector $x\in \Cbb^n$.

In other words, the minimal number $m$ of distinct blocks required to identify every $s$-sparse vectors is necessarily larger than $\min(2s,\sqrt{n})$. 
\end{prop}

This theoretical bound is quite surprising: it seems to enter in contradiction with the practical results obtained in Figure \ref{fig:reconstructionLines} or with one of the most standard CS strategy in MRI \cite{lustig2007sparse}. Indeed, the equivalent number of isolated measurements required by Proposition \ref{prop:tensorlimit} is of the order $O(s\sqrt{n})$. This theoretical result means that in many applications, a full sampling strategy should be adopted, when the acquisition is structured by horizontal lines.
In the next paragraphs, we show how Theorem \ref{thm:recovery} allows bridging the gap between theoretical recovery and practical experiments.

\subsubsection{Breaking the limits with adapted structured sparsity}

In this paragraph, we illustrate - through a simple example - that additional assumptions on structured sparsity is the key to explain practical results.

\begin{corollary}
\label{corol:Fourier2DLines}
Let $A_0\in \Cbb^{n\times n}$ be the 2D Fourier transform.
Assume that $x$ is a 2D signal with support $S$ concentrated on $q$ horizontal lines of the spatial plane, i.e.  
\begin{align}
\label{def:Slines}
S \subset \{ (j-1) \sqrt{n} + \{1, \hdots , \sqrt{n}\}, j\in J\}
\end{align} 
where $J\subset \{1, \hdots , \sqrt{n}\}$ and $|J|=q$. 

Choose a uniform sampling strategy among the $\sqrt{n}$ horizontal lines, i.e. $\pi_k^\star = 1/\sqrt{n}$ for $1\leq k \leq \sqrt{n}$. The number $m$ of sampled horizontal lines sufficient to reconstruct $x$ with probability $1-\varepsilon$ is
$$ 
m \geq C \cdot   q\cdot \ln(s) \ln \left( \frac{n }{\varepsilon}\right).
$$ 
\end{corollary}
The proof is given in Appendix \ref{app:proofFourier2DLines}
By Proposition \ref{corol:Fourier2DLines}, we can observe that the required number of sampled lines is of the order of non-zero lines in the 2D signal.
In comparison, Proposition 4.6 in \cite{bigot2014analysis} (with no structured sparsity) requires
$$ 
m \gtrsim    s\cdot \ln(n/\varepsilon),
$$
measurements, to get the same guarantees. 
This means that the required number of horizontal lines to sample is of the order of the non-zero coefficients.
By putting aside the logarithmic factors, we see that the gain with our new approach is considerable. 
Clearly, our strategy is able to take advantage of the sparsity structure of the signal of interest.

\subsubsection{Consequences for MRI sampling}
 \label{subsec:linesMRI}

We now turn to a real MRI application. 
We assume that the sensing matrix $A_0 \in \Cbb^{n\times n}$ is the product of the 2D Fourier transform $\Fc_{2D}$ with the inverse 2D wavelet transform $\Phi^*$. 
We aim at reconstructing a vector $x\in \Cbb^n$ that can be seen as a 2D wavelet transform with $\sqrt{n}\times \sqrt{n}$ coefficients.
Set $J = \log_2\left(\sqrt{n}\right)-1$ and let $\left( \tau_j \right)_{0\leq j\leq J}$ denote a dyadic partition of the set $ \{1,\hdots , \sqrt{n} \}$, i.e. $\tau_0 = \{1 \}, \, \tau_1 = \{2\} , \, \tau_2 = \{ 3, 4 \} , \hdots , \,  \tau_J = \{\sqrt{n}/2+1,\hdots , \sqrt{n} \}$. 
Define $j : \{ 1 , \hdots , \sqrt{n} \} \rightarrow \{ 0 , \hdots , J\}$ by $j(u)=j$ if $u\in \tau_j$.
Finally, define the sets $\Omega_{\ell,\ell'} = \tau_\ell \times \tau_{\ell'}$, for $0\leq \ell, \ell' \leq J$. 
See Figure \ref{fig:illusEns} for an illustration of these sets.


 \begin{figure}
  \begin{center}
 \includegraphics[height=10cm]{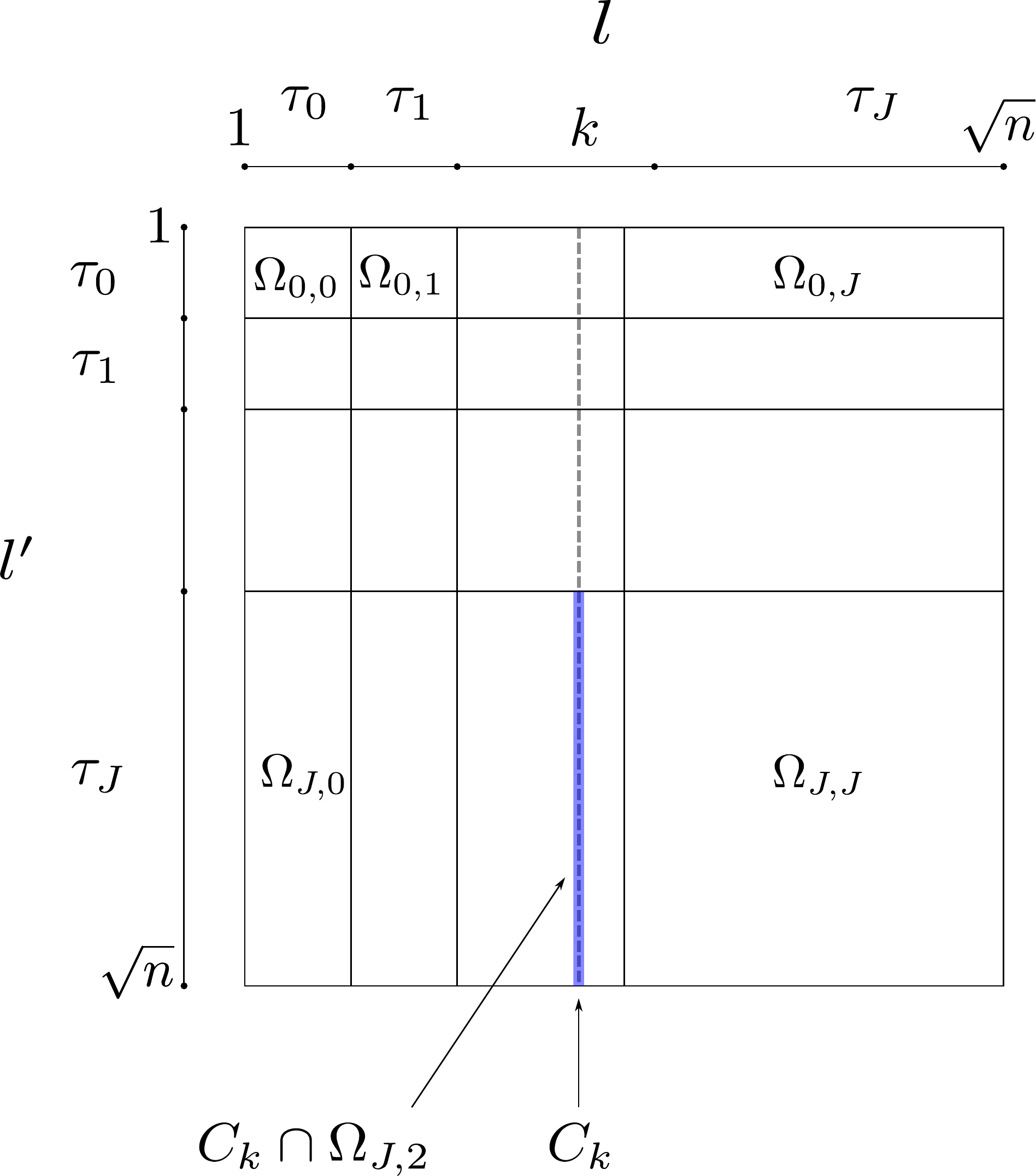}
 \caption{\label{fig:illusEns}2D view of the signal $x\in \Cbb^n$ to reconstruct. The vector $x$ can be reshaped in a $\sqrt{n}\times \sqrt{n}$ matrix. $C_k$ represents the coefficient indexes corresponding to the $k$-th vertical column. }
  \end{center}
 \end{figure}

\begin{defchapter}
Given $S = \text{supp}(x)$, define the following quantity
\begin{align}\label{eq:scol}
s^c_{\ell} := \max_{0\leq \ell'\leq J} \max_{ k \in \tau_{\ell'}  } \left| S \cap \Omega_{\ell,\ell'} \cap C_k\right| ,
\end{align}
where $C_k$ represents the set corresponding to the $k$-th vertical line (see Figure \ref{fig:illusEns}). 
\end{defchapter}
The quantity $s^c_{\ell}$ represents the maximal sparsity of $x$ restricted to columns (or vertical lines) of $\cup_{1\leq l'\leq J}\Omega_{\ell,\ell'}$. We have now settled everything to state our result.

 As a first step, we will consider the case of Shannon's wavelets, leading to a block-diagonal sampling matrix $A_0$.
 \begin{corollary}
 \label{corol:linesShannonMRI}
 Let $S \subset \{ 1 ,\hdots , n \}$ be a set of indices of cardinality $s$.  Suppose that $x \in \Cbb^{n}$ is an $s$-sparse vector supported on $S$.  Fix $\varepsilon \in (0,1)$. Suppose that $A_0$ is the product of the 2D Fourier transform with the 2D inverse Shannon's wavelets transform.   
 Consider that the blocks of measurements are the $\sqrt{n}$ horizontal lines in the 2D setting.
 Choose $\left( \pi_k \right)_{1\leq k \leq \sqrt{n}}$ to be constant by level, i.e. $\pi_k = \tilde{\pi}_{j(k)}$. If the number of horizontal lines to acquire satisfies
 \begin{eqnarray*}
 m &\gtrsim & \max_{0\leq j \leq J} \frac{1}{\tilde{\pi}_j} 2^{-j} s^c_{j} \ln(s) \ln \left( \frac{ n }{\varepsilon}\right),
 \end{eqnarray*}
 then $x$ is the unique solution of Problem \ref{pb:minL1}. Furthermore, choosing $\tilde{\pi}_j = \frac{s^c_{j} / 2^j}{ \sum_{\ell=0}^J s^c_{\ell}}$, for $0\leq j \leq J$, leads to the following upper bound
 \begin{eqnarray*}
 m &\gtrsim & \sum_{j=0}^J s^c_{j} \ln(s) \ln \left( \frac{ n }{\varepsilon}\right).
 \end{eqnarray*}
 \end{corollary}
 The proof is given in Section \ref{app:proofLinesShannonMRI}. 
 Corollary \ref{corol:linesShannonMRI} shows that the number of lines acquired at level $j$ depends only on an extra-column structure of $S$.
Now let us turn to a case where the matrix $A_0$ is not block-diagonal anymore.

\begin{corollary}
\label{corol:linesHaarMRI}
Suppose that $x \in \Cbb^{n}$ is an $S$-sparse vector.
Fix $\varepsilon \in (0,1)$. 
Suppose that $A_0$ is the product of the 2D Fourier transform with the 2D inverse Haar transform.   
Consider that the blocks of measurements are the $\sqrt{n}$ horizontal lines.
Choose $\left( \pi_k \right)_{1\leq k \leq \sqrt{n}}$ to be constant by level, i.e. $\pi_k = \tilde{\pi}_{j(k)}$. 

If the number $m$ of drawn horizontal lines satisfies
\begin{eqnarray*}
m &\gtrsim & \max_{0\leq j \leq J}\frac{2^{-j}}{\tilde{\pi}_j}  \sum_{r=0}^J 2^{-| j - r|/2} {s^c_{r}} \ln(s) \ln \left( \frac{ n }{\varepsilon}\right),
\end{eqnarray*}
then $x$ is the unique solution of Problem \ref{pb:minL1} with probablity $1-\varepsilon$. 

In particular, if
$$ {\pi}_k = \frac{2^{-j(k) \sum_{r=0}^J 2^{-|j-r|/2} s^c_{r}}}{ 
\sum_{\ell=1}^{\sqrt{n}}  
2^{-j(\ell)} \sum_{r=0}^J 
2^{-|j(\ell) -r|/2} s^c_{r}  },
$$
then 
\begin{eqnarray*}\label{eq:finalbound}
m &\gtrsim & \sum_{j=0}^J \left(  s^c_{j}  + \sum_{r=0 \atop r\neq j}^J 
2^{-|j -r|/2} s^c_{r}  \right) \cdot \ln(s) \ln \left( \frac{ n }{\varepsilon}\right)
\end{eqnarray*}
ensures perfect reconstruction with probability $1-\varepsilon$.
\end{corollary}
The proof of Corollary \ref{corol:linesHaarMRI} is given in Section \ref{app:prooflinesHaarMRI}. 

This result indicates that the number of acquired lines in the "horizontal" level $j$ should be chosen depending on the quantities $s_j^c$. 
Note that this is very different from the sparsity by levels proposed in \cite{adcock2013breaking}.
In conclusion, Corollary \ref{corol:linesHaarMRI} reveals that with a structured acquisition, the sparsity needs to be more structured in order to guarantee exact recovery. 
To the best of our knowledge, this is the first theoretical result which can explain why sampling lines in MRI as in \cite{lustig2007sparse} might work. 
In Figure \ref{fig:roseau}, we illustrate that the results in Corollary \ref{corol:linesHaarMRI} seem to indeed correspond to the practical reality. 
In this experiment, we seek reconstructing a reeds image from block structured measurements. As a test image, we chose a reeds image with vertical stripes of its rotated version. This particular geometrical structure explains that the quantities $s^c_{j}$ are much higher for the horizontal stripes than for the vertical one. As can be seen, the image with a low $s^c_{j}$ is much better reconstructed than the one with a high $s^c_{j}$. This wa predicted by our theory. 
\begin{figure}[h!]
\begin{center}
\includegraphics[height=5cm]{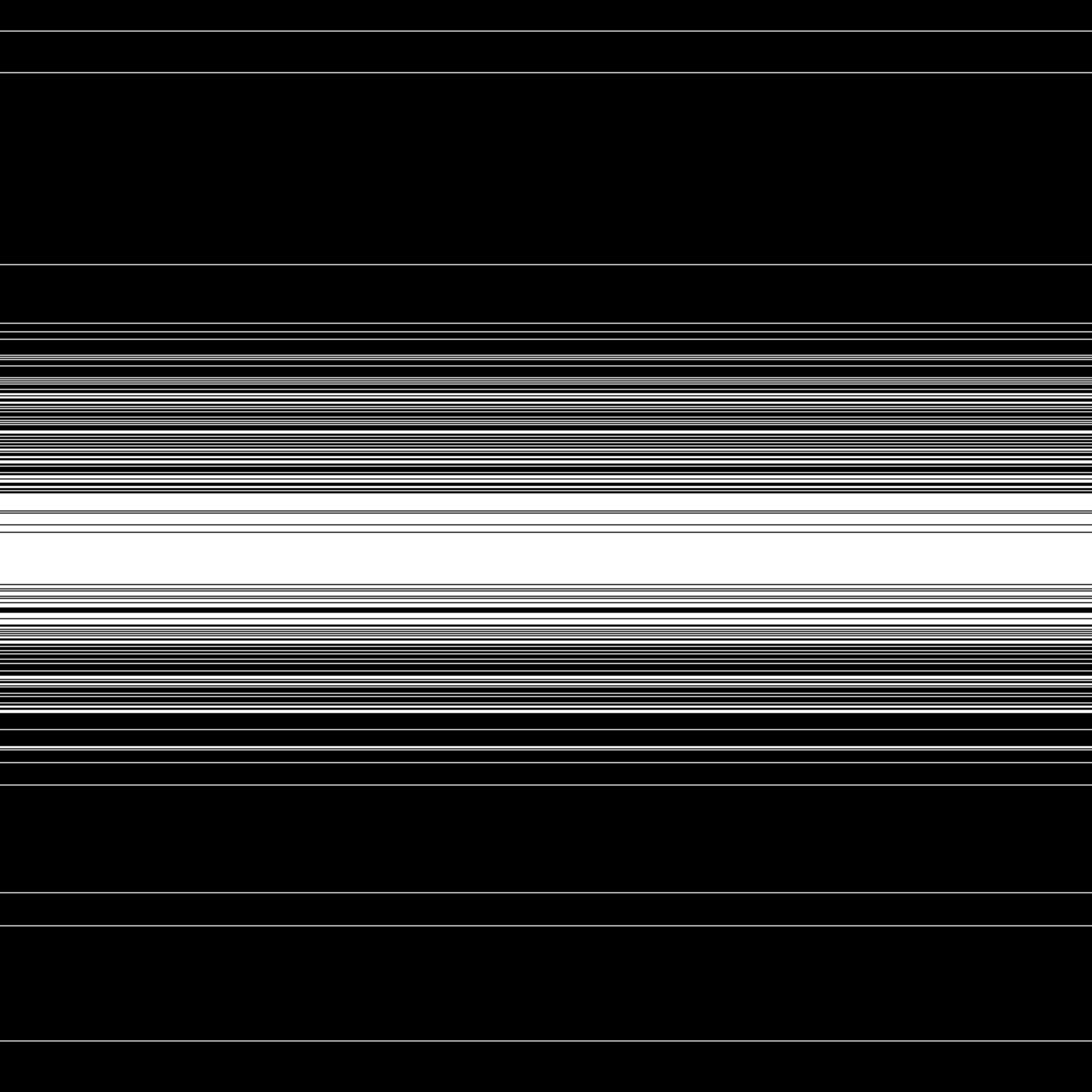} \\
{\small{Sampling scheme}}
\btabu{@{}cc}
\includegraphics[height=7cm]{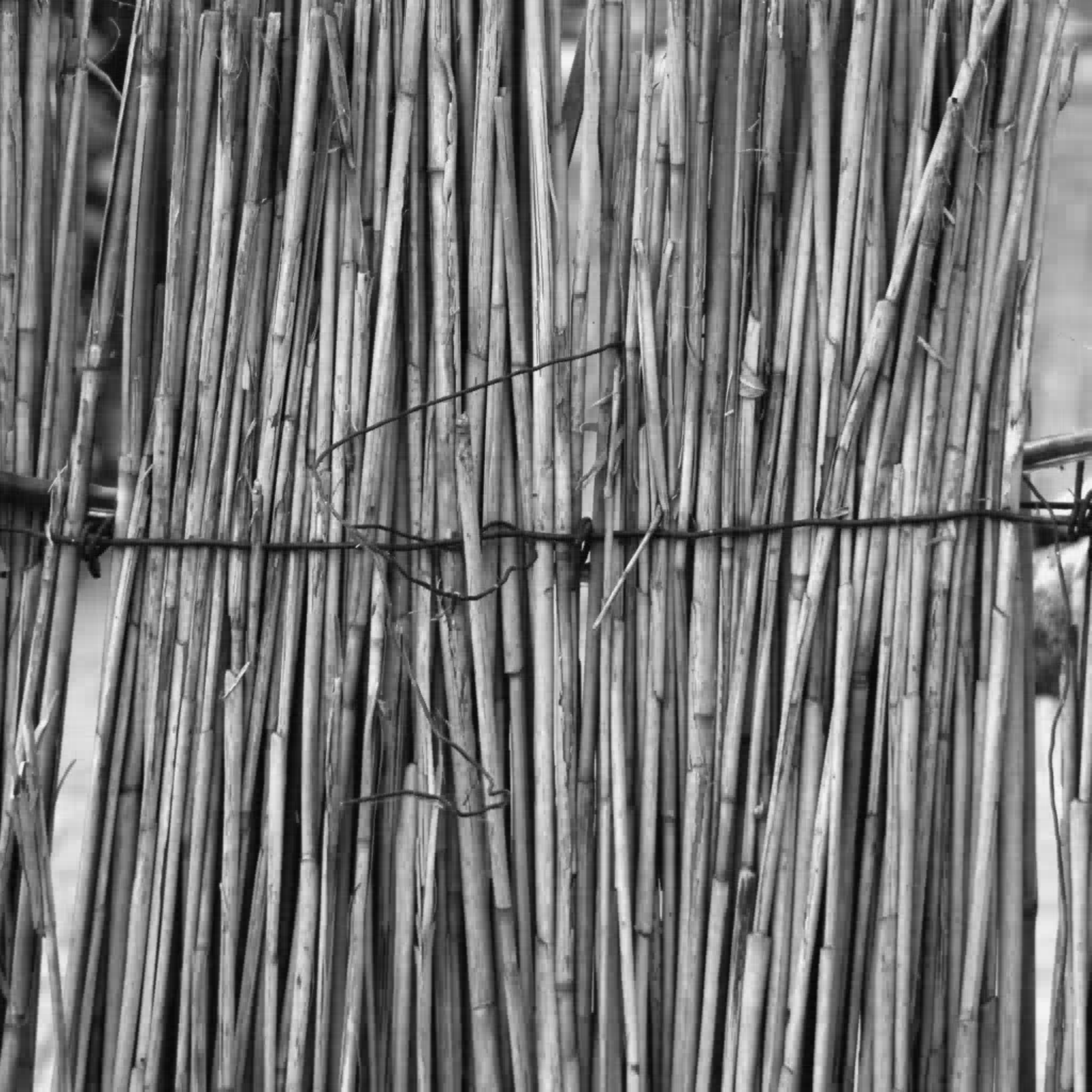}  &
\hspace{-0.3cm}\includegraphics[height=7cm]{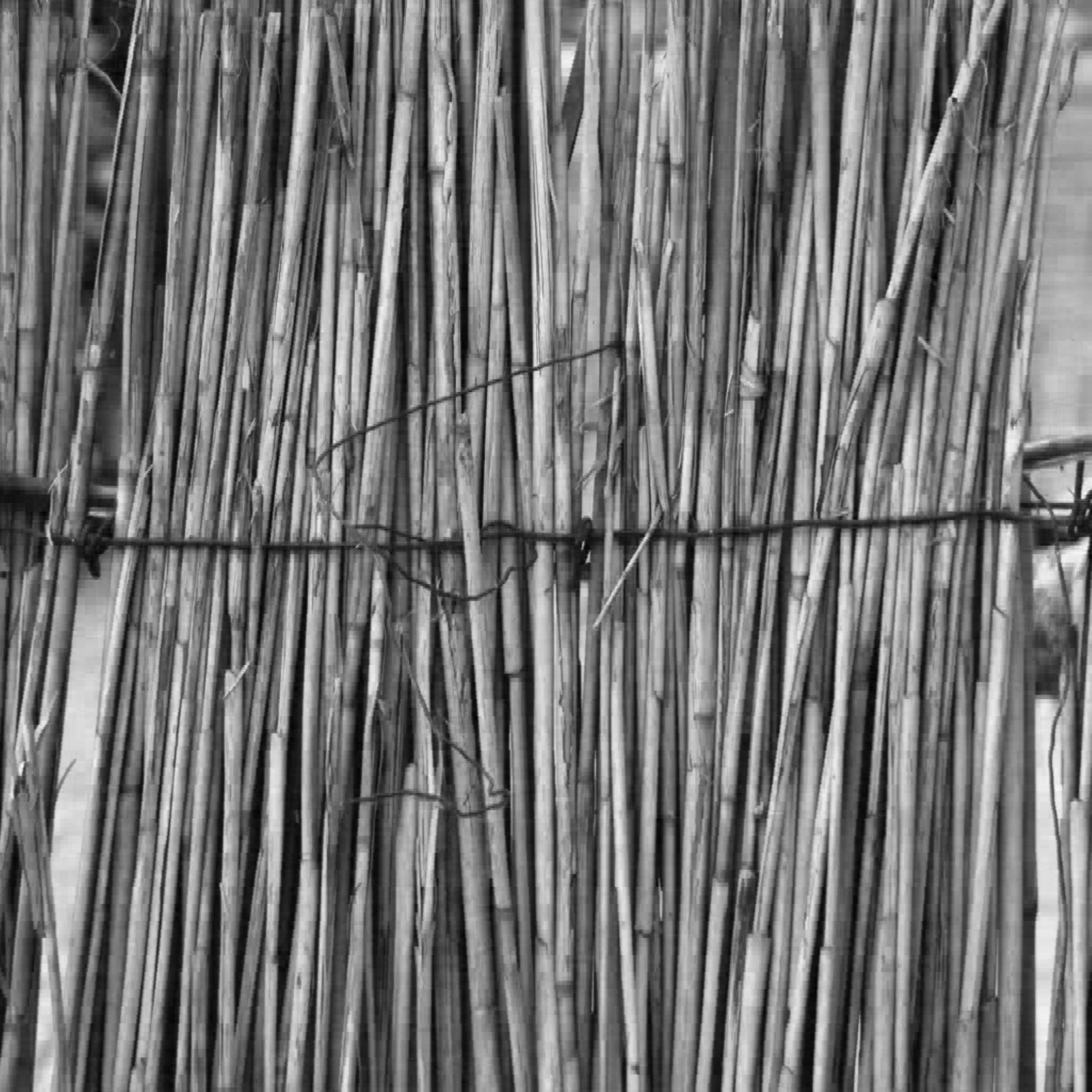} 
\\
{\small (a) Original image }& {\hspace{-0.3cm}\small (b) SNR = 27.8 dB} \\
$s^c = (16, \,16, \, 32, \, 59, \, 81, \, 75, \, 48)$ &  \\
\includegraphics[height=7cm]{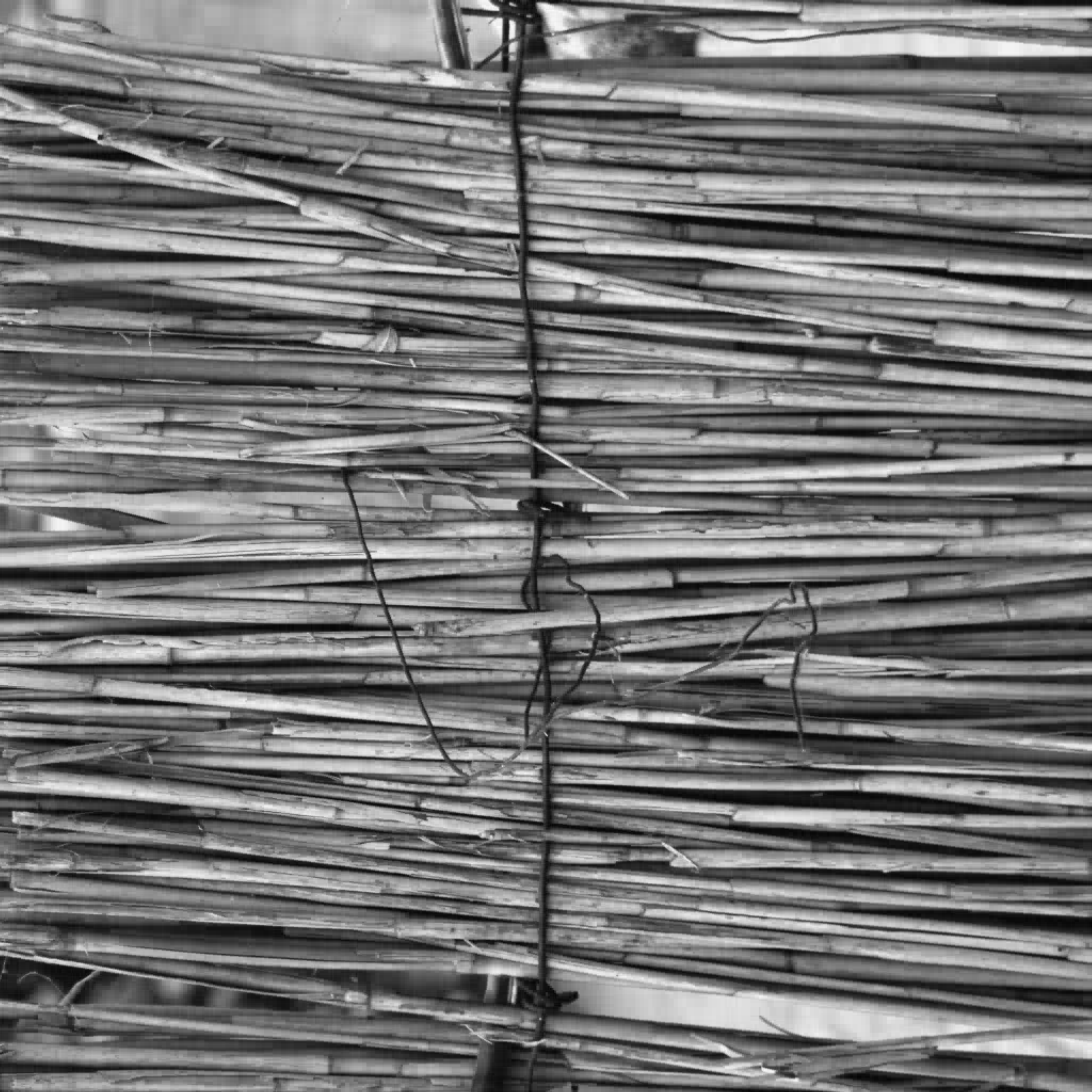} &
\hspace{-0.3cm}\includegraphics[height=7cm]{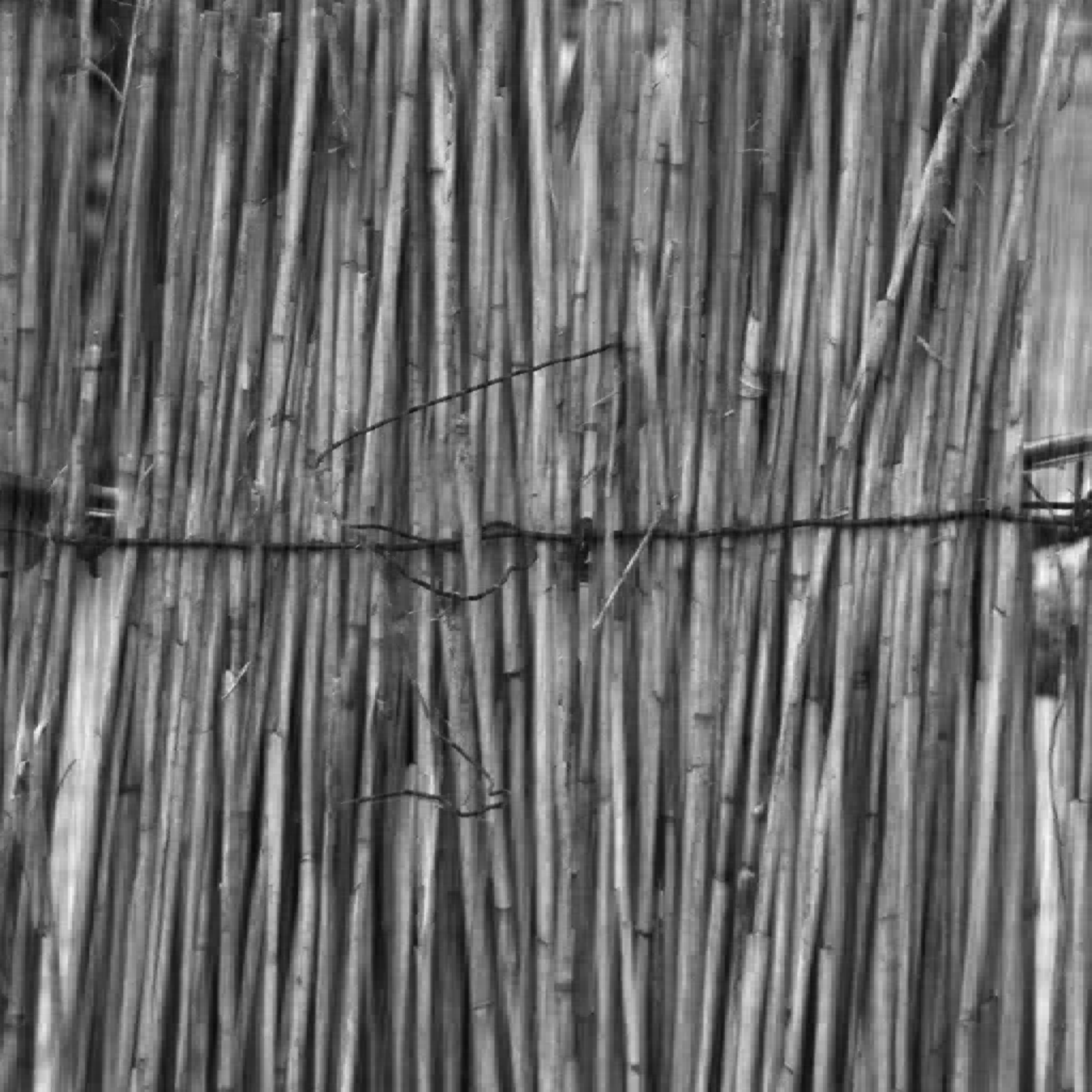} 
\\
{\small (c) Original image }  & {\hspace{-0.3cm}\small (d) SNR = 14.7 dB} \\
$s^c = (16, \,16, \, 32, \, 64, \, 124, \, 240, \, 411)$ &  \\
\etabu
\caption{\label{fig:roseau} An example of reconstruction of a $2048\times 2048$ real image sensed in the Fourier domain. In (a) (c),  Reference images to reconstruct: (c) is the same image as (a) but rotated of $90^\circ$.  We precise the value of the vector $s^c = \left( s^c_j \right)_{1\leq j \leq 7}$ for both images. Note that the quantities $s^c_j$ are larger in the case of image (b). For the reconstruction, we use the sampling scheme at the top of the Figure. It corresponds to 9.8 \% of measurements.
In (b) (d), corresponding reconstruction via $\ell_1$-minimization. We have rotated the image in (d) to facilitate the comparison between both. Note that (b) is much better reconstructed than (d). This is predicted by Corollary \ref{corol:linesHaarMRI}.}
\end{center}
\end{figure}
\FloatBarrier



\section{Extensions}

\subsection{The case of Bernoulli block sampling}
We analyzed the combination of structured acquisition and structured sparsity with i.i.d.\ drawings of random blocks. 
These results can be extended to a Bernoulli sampling setting. 
In such a setting, the sensing matrix is constructed as follows
$$ A = \left( \frac{\delta_k}{\sqrt{\pi_k}} B_k \right)_{1\leq k \leq M},
$$
where $\left( \delta_k \right)_{1\leq k \leq M}$ are independent Bernoulli random variables such that $\Pbb \left( \delta_k = 1\right) = \pi_k$, for $1\leq k \leq M$. We may set $\sum_{k=1}^M \pi_k = m$ in order to measure $m$ blocks of measurements in expectation. By considering the same definition for $\Gamma(S,\pi)$ with $\left( \pi_k \right)_{1\leq k \leq M}$ the Bernoulli weights, it is possible, for the case of Bernoulli block sampling, to give a reconstruction result that shares a similar flavor to Theorem \ref{thm:recovery}.

\subsection{Towards new sampling schemes?}

The results in Section \ref{subsec:linesMRI} lead to the conclusion that exact recovery with structured acquisition can only occur if the the signal to reconstruct possesses an adequate sparsity pattern.
We believe that the proposed theorems might help designing new efficient and feasible sampling schemes.
Ideally, this could be done by optimizing $\Gamma(S,\pi)$ assuming that $S$ belongs to some set of realistic signals. 
Unfortunately, this optimization seems unrealistic to perform numerically, owing to the huge dimensions of the objects involved.
We therefore leave this question open for future works.

However, probing the limits of a given system, as was proposed in Corollary \ref{corol:linesHaarMRI} helps designing better sampling schemes. 
To illustrate this fact, we performed a simple experiment. 
Since the quantity $s^j_{c}$ is critical to characterize a sampler efficiency, it is likely that mixing horizontal and vertical sampling lines improves the situation.
We aim at reconstructing the MR image shown in Figure \ref{fig:cerveau} and assume that it is sparse in the wavelet basis. 
In Figure \ref{fig:cerveau}(a)(d), we propose two different sampling schemes. The first one is based solely on parallel lines in the horizontal direction, while the second one is based on a combination of vertical and horizontal lines. The combination of vertical and horizontal lines provides much better reconstruction results despite a lower total number of measurements.

\begin{figure}[h!]
\begin{center}
\includegraphics[height=5cm]{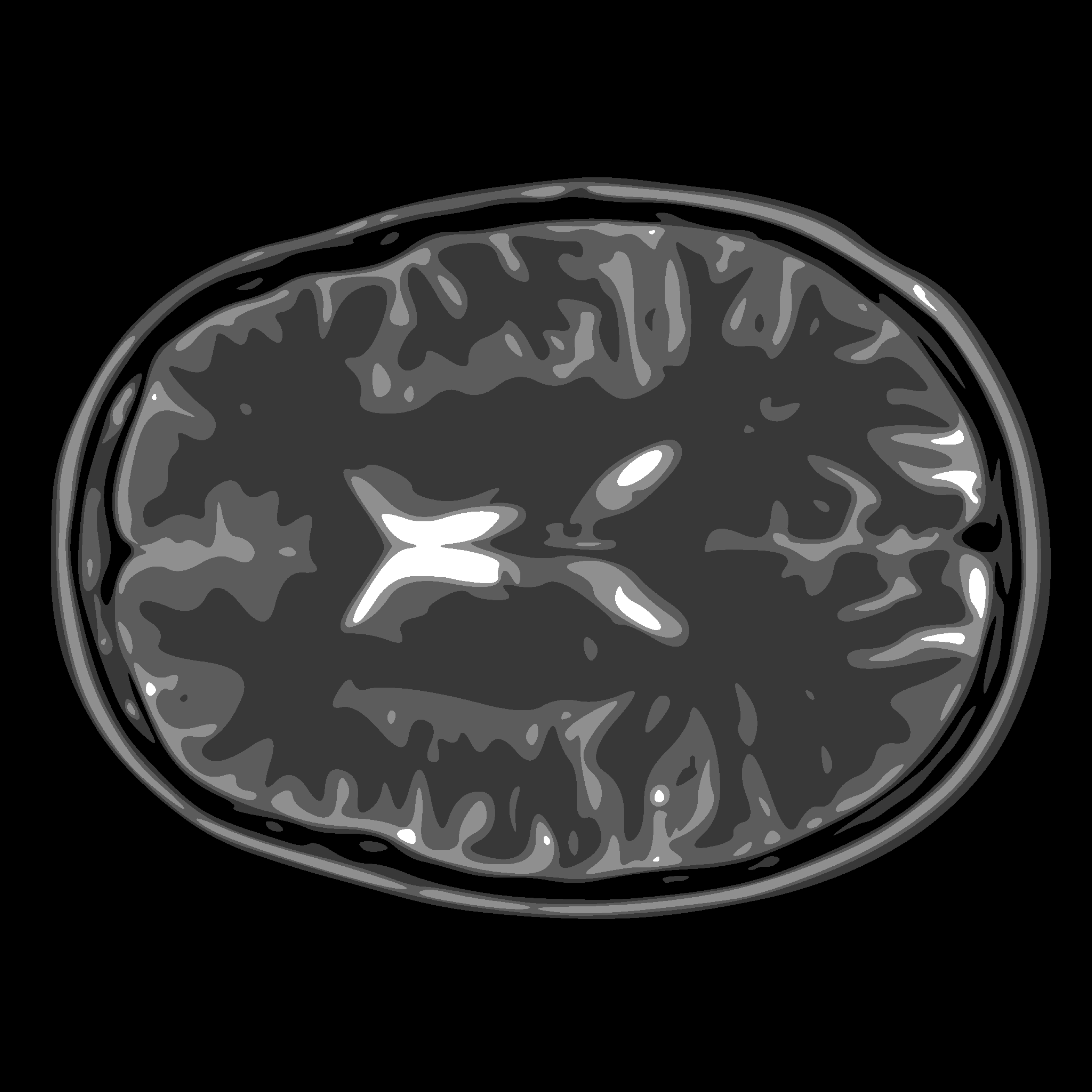} \\
{\small{Reference image}}
\btabu{@{}ccc}
\includegraphics[height=5cm]{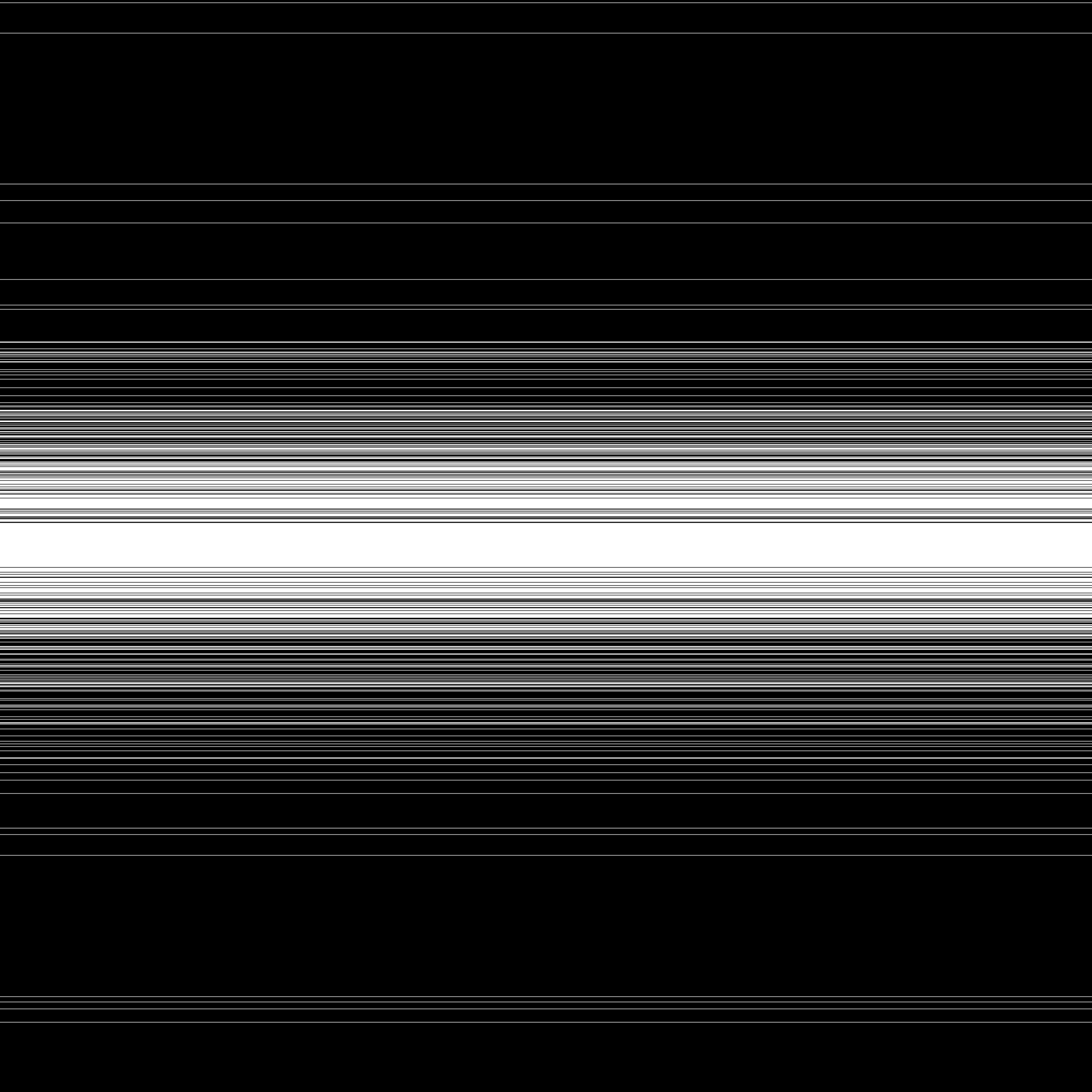}  &
\hspace{-0.3cm}\includegraphics[height=5cm]{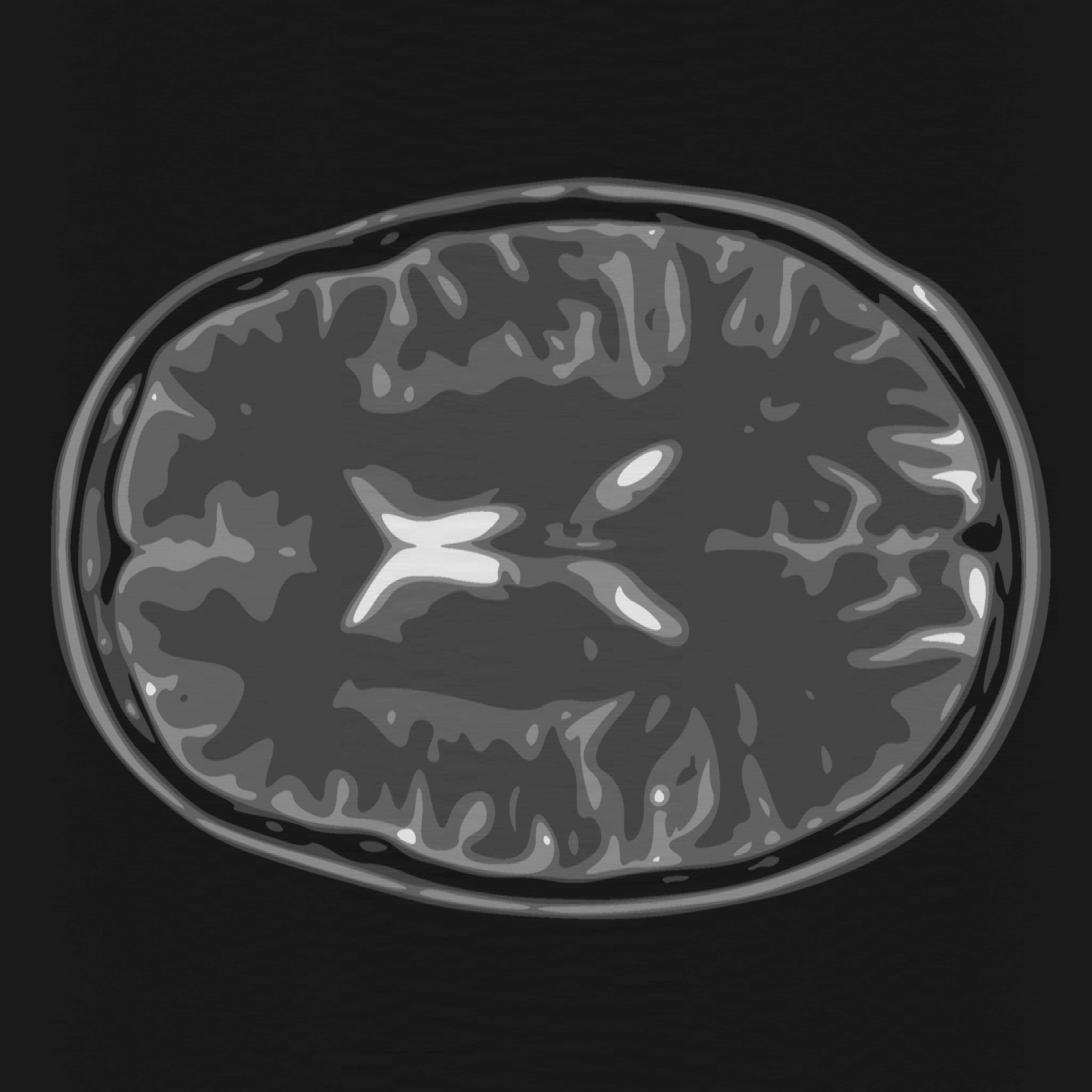} &
\hspace{-0.3cm}\includegraphics[height=5cm]{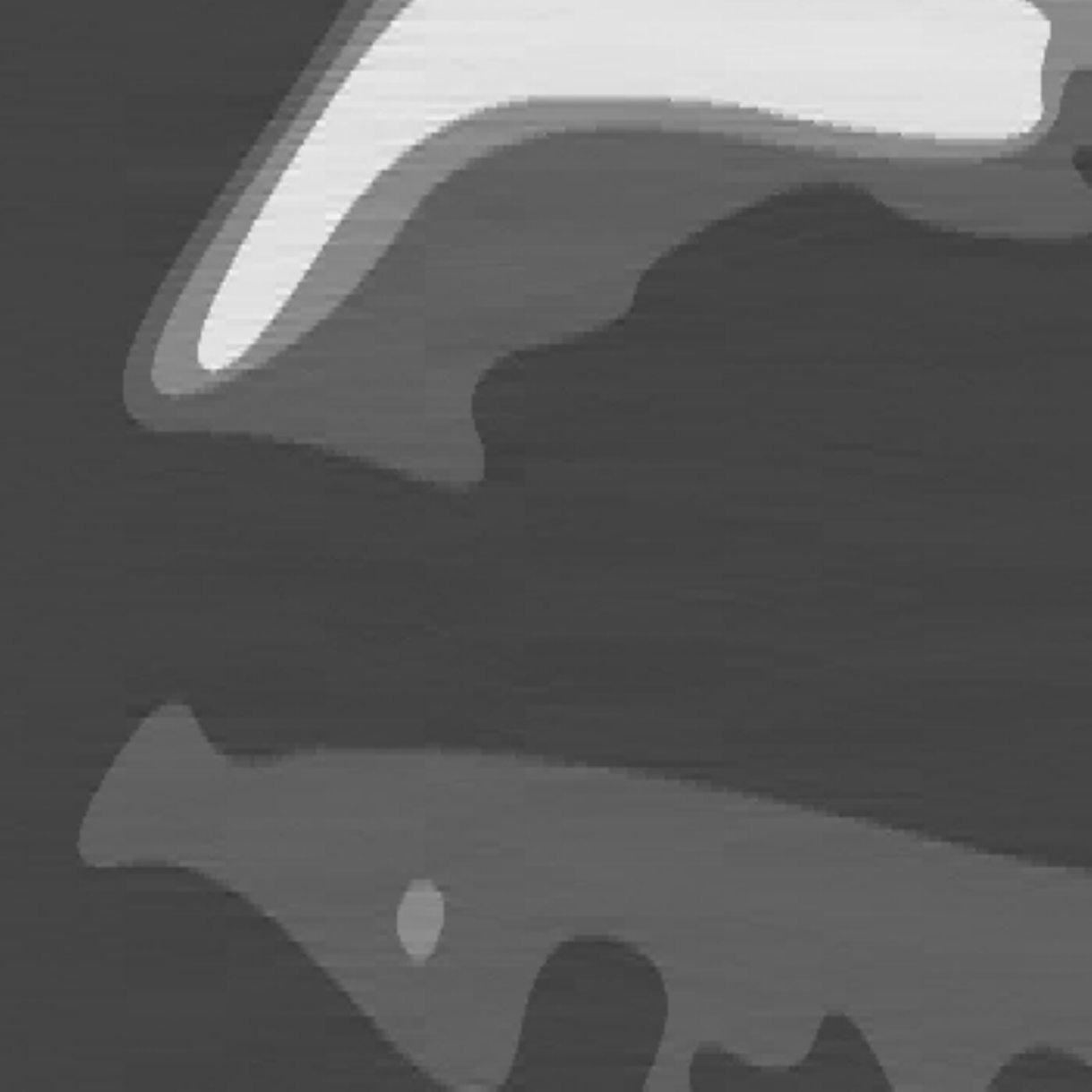} 
\\
{\small (a) Sampling scheme }& {\hspace{-0.3cm}\small (b) SNR = 24.47 dB} & {\hspace{-0.3cm}\small (c)} \\
\includegraphics[height=5cm]{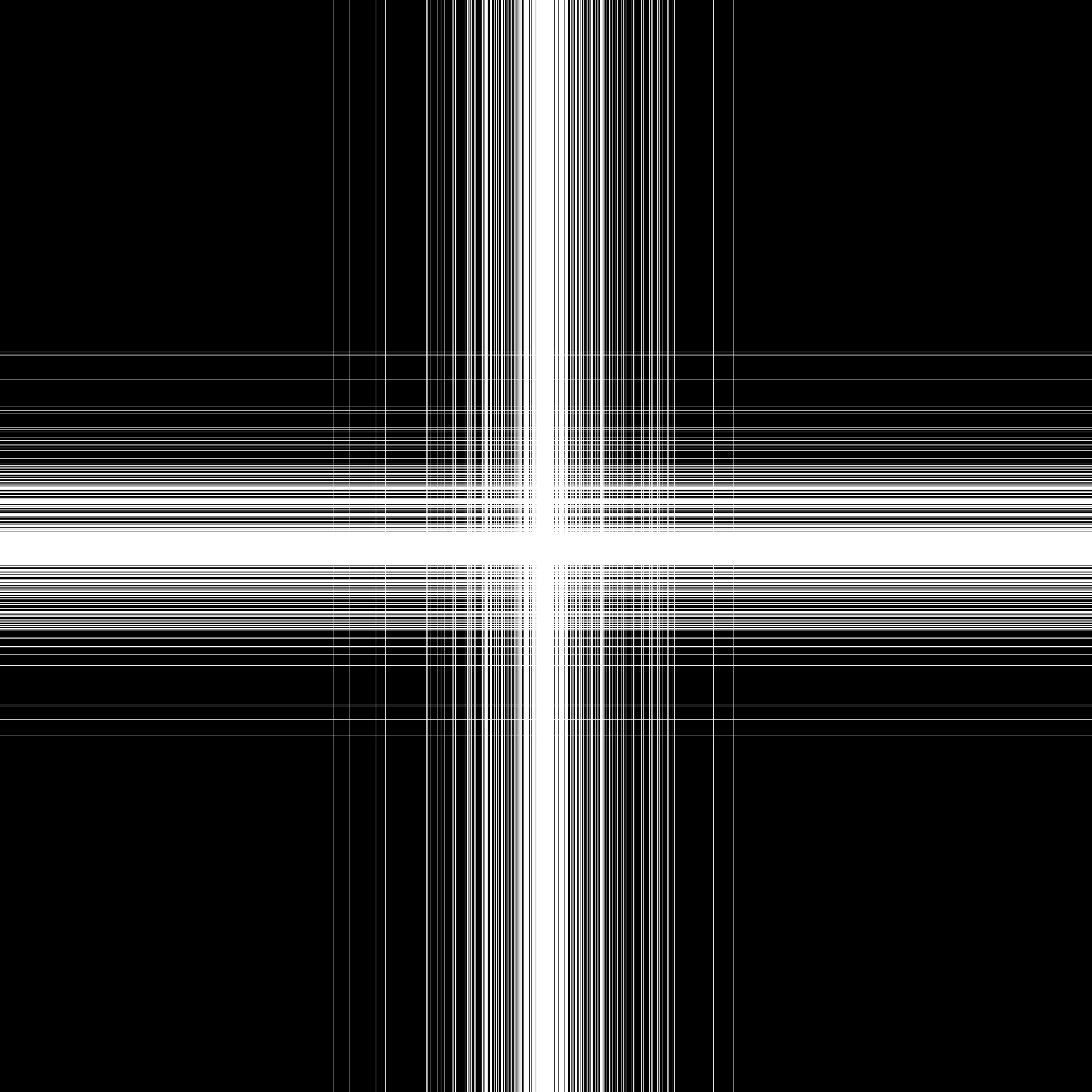} &
\hspace{-0.3cm}\includegraphics[height=5cm]{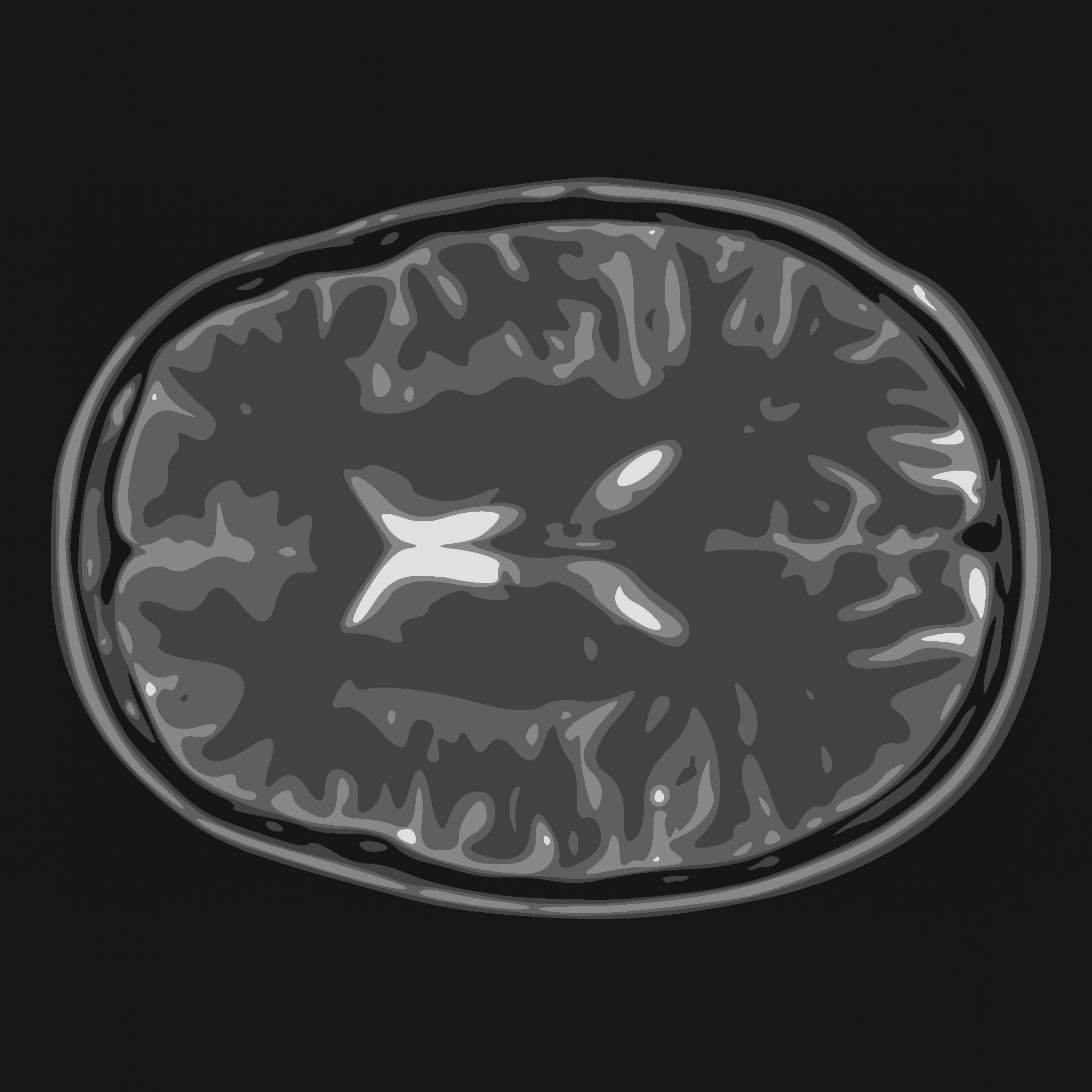} &
\hspace{-0.3cm}\includegraphics[height=5cm]{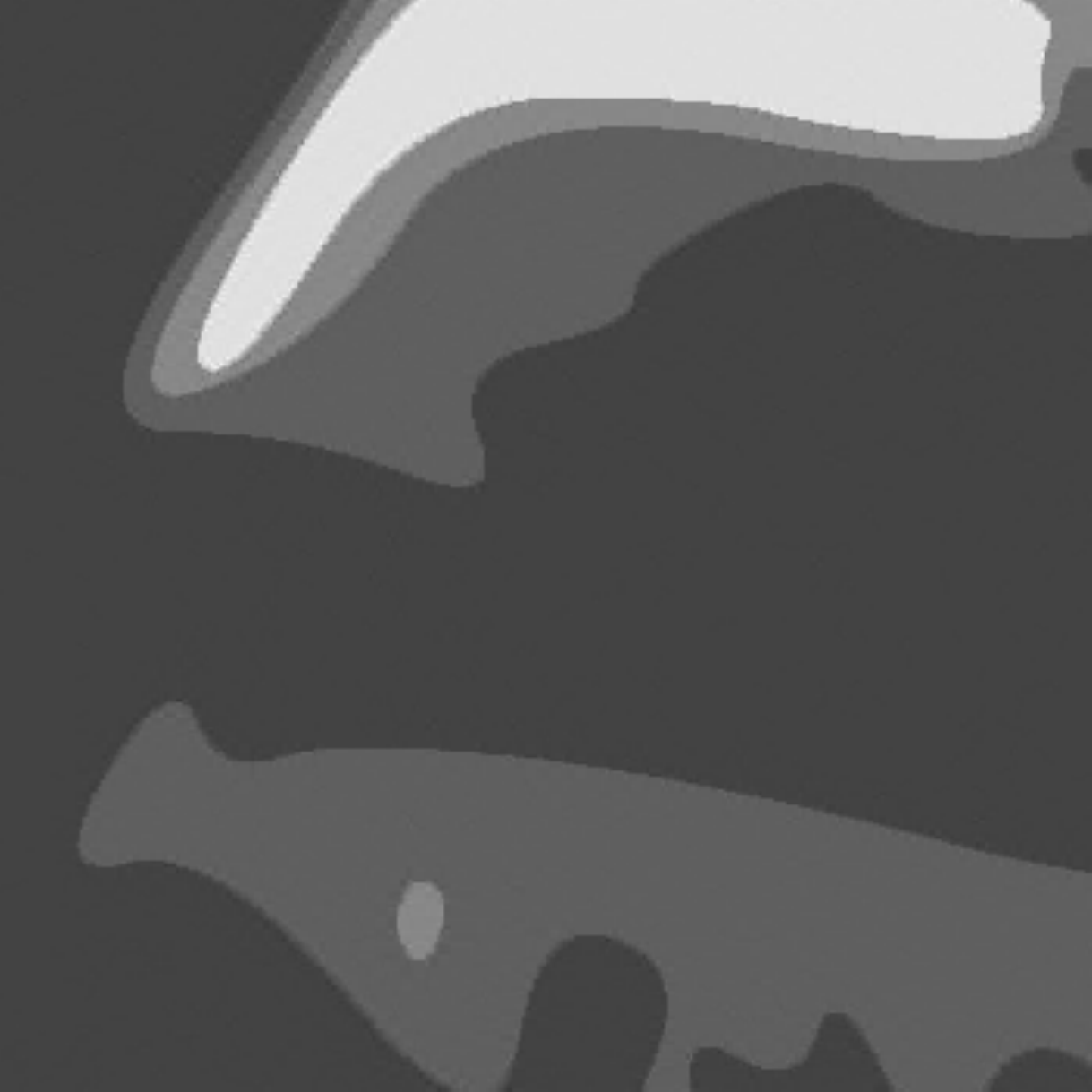} 
\\
{\small (d) Sampling scheme }  & {\hspace{-0.3cm}\small (e) SNR = 26.74 dB} & {\hspace{-0.3cm}\small (f)} \\
\etabu
\caption{\label{fig:cerveau} An example of MRI reconstruction of a $2048\times 2048$ phantom. The reference image to reconstruct is presented at the top of the figure. It is considered sparse in the wavelet domain. In (a) (d),  we present two kinds of sampling schemes with 20 \% of measurements: the samples are acquired in the 2D Fourier domain.
In (b) (e), we show the corresponding reconstruction via $\ell_1$-minimization. In (c) (f) we enhance the results by zooming on the reconstructed images. Note that the horizontal and vertical sampling scheme produces much better reconstruction results despite a smaller number of measurements since samples are overlapping. Moreover, the acquisition time would be exactly the same for an MRI.}
\end{center}
\end{figure}
\FloatBarrier


%
%

\section*{Acknowledgement}

The authors would like to thank Ben Adcock and Anders Hansen for their availibility for discussion. They are also grateful to Nicolas Chauffert for discussion.
This work was partially supported by the CIMI (Centre International de Math\'ematiques et d'Informatique)
Excellence program.

\appendix
\section{Proofs of the main results}

\subsection{Proof of Theorem \ref{thm:recovery}}
\label{sec:proof1}
In this section, we give sufficient conditions to guarantee that the vector $x$ is the unique minimizer of  \eqref{pb:minL1}, using an inexact dual certificate see \cite{candes2011probabilistic}.

\begin{lemme}[Inexact duality \cite{candes2011probabilistic}]
\label{lem:inexactDuality}
Suppose that $x \in \Rbb^{n}$ is supported on $S \subset \{1, \hdots , n\}$.  Assume that $A_S$ is full column rank and that
\begin{equation}
\| \left(A^*_S A_S \right)^{-1} \|_{2 \rightarrow 2} \leq 2 \qquad \text{and} \qquad  \max_{i\in {S^c}} \left\|  A_S^* A e_i \right\|_{2} \leq 1, \label{ass:dual1}
\end{equation}
where $\left(A^*_S A_S \right)^{-1} $ only makes sense on the set $\text{span}\{ e_i , i\in S\}$. Morever, suppose that there exists $v \in \Rbb^n$ in the row space of $A$ obeying
\begin{equation}
\| v_S - \sgn(x_S)\|_2 \leq 1/4 \qquad \text{and} \qquad  \|  v_{S^c}  \|_\infty\leq 1/4, \label{ass:dual2}
\end{equation}
Then, the vector $x$ is the unique solution of the minimization problem \eqref{pb:minL1}
\end{lemme}

First, let us focus on Conditions \eqref{ass:dual1}.  
Remark that $A_S^* A_S$ is invertible by assuming that $A_S$ is full column-rank.
Moreover,
$$ \| \left(A_S^* A_S\right)^{-1} \|_{2\rightarrow 2}  = \left\| \sum_{k=0}^\infty \left( A_S^* A_S -P_S \right)^k \right\|_{2\rightarrow 2}  \leq \sum_{k=0}^\infty  \left\|  A_S^* A_S - P_S  \right\|_{2\rightarrow 2} ^k.
$$
Therefore, if $\left\| A_S^* A_S- P_S \right\|_{2\rightarrow 2}  \leq 1/2$  is satisfied, then   $\| \left(A_S^* A_S\right)^{-1} \|_{2\rightarrow 2}  \leq 2$. Moreover, by Lemma \ref{lem:localIsometry}, $\| \left(A_S^* A_S \right)^{-1} \|_{2\rightarrow 2}  \leq 2$  with probability at least $1-\varepsilon$, provided that
$$
m \geq \frac{28}{3} \Theta (S, \pi) \ln \left( \frac{2 s  }{\varepsilon}  \right) .
$$
By definition of $\Gamma(S, \pi)$, the first inequality of Conditions \eqref{ass:dual1} is therefore ensured with probability larger than $1-\varepsilon$ if
\begin{equation}
\label{eq:condm2}
m \geq \frac{28}{3} \Gamma(S ,\pi)  \ln \left( \frac{2 s  }{\varepsilon}  \right).
\end{equation}
Furthermore, using Lemma \ref{lem:offSupportCoherence}, we obtain that 
$$
\max_{i\in {S^c}} \| A_S^* A e_{i} \|_{2}  \leq 1
$$
 with probability larger than $1-\varepsilon$ if
$$
m \geq \Theta (S ,\pi) \left( 1+ 4 \sqrt{\ln \left( \frac{n}{\varepsilon} \right)} + 4 {\ln \left( \frac{n}{\varepsilon} \right)} \right).
$$
Again by definition of $\Gamma(S ,\pi)$, the second part of Conditions \eqref{ass:dual2} is ensured if $n\geq 3$ and
\begin{equation}
\label{eq:condm3}
m \geq 9\Gamma(S ,\pi) \ln \left( \frac{n }{\varepsilon}\right). 
\end{equation}

Conditions \eqref{ass:dual2} remain to be verified. The rest of the proof of Theorem \ref{thm:recovery} relies on the construction of a vector $v$ satisfying the conditions described in Lemma \ref{lem:inexactDuality} with high probability. To do so, we adapt the so-called golfing scheme introduced by Gross \cite{gross2011recovering} to our setting. More precisely, we will iteratively construct a vector that converges to a vector $v$ satisfying \eqref{ass:dual2} with high probability.

Let us first partition the sensing matrix $A$ into blocks of blocks so that, from now on, we denote by $A^{(1)}$  the first $m_1$ blocks of $A$, $A^{(2)}$ the next $m_2$ blocks, and so on. The $L$ random matrices $\left\{ A^{(\ell)} \right\}_{\ell=1,\ldots,L}$ are independently distributed, and we have that $m = m_1+m_2+\hdots+m_L$. As explained before, $A^{(\ell)}_S$ denotes the matrix $A^{(\ell)} P_{S}$.

The golfing scheme starts by defining $v^{(0)}=0$, and then it iteratively defines
\begin{align}
v^{(\ell)} =  \frac{m}{m_\ell} A^{(\ell)^*} A^{(\ell)}_S \left( \sgn (x) - v^{(\ell-1)} \right) +v^{(\ell-1)},
\end{align}
for $\ell=1,\hdots, L$, where $\sgn(x_i)=0$ if $x_i=0$. In the rest of the proof, we set $v = v^{(L)}$. By construction, $v$ is in the row space of $A$. The main idea of the golfing scheme is then to combine the results from the various Lemmas in Section \ref{sec:aux} with an appropriate choice of $L$ to show that the random vector $v$ satisfies the assumptions of Lemma \ref{lem:inexactDuality} with large probability.
Using the shorthand notation $v_{S}^{(\ell)} = P_{S} v^{(\ell)}$, let us define
$$
w^{(\ell)} =  \sgn (x) -v_S^{(\ell)}, \; \ell = 1,\ldots,L,
$$
where $x \in \Cbb^{n}$ is the solution of Problem \eqref{pb:minL1}. 

From the definition of $v^{(\ell)}$, it follows that, for any $1 \leq \ell \leq L$,
\begin{equation}
\label{eq:qirec}
w^{(\ell)} = \left( P_S - \frac{m}{m_\ell}  A^{(\ell)*}_{S} A^{(\ell)}_{S} \right) w^{(\ell-1)} = \prod_{j=1}^\ell \left( P_S -  \frac{m}{m_j}  A^{(j)*}_{S} A^{(j)}_{S} \right)   \sgn(x),
\end{equation}
and
\begin{equation}
\label{eq:vfctq}
v = \sum_{\ell=1}^L  \frac{m}{m_\ell} A^{(\ell)*} A^{(\ell)}_{S} w^{(\ell-1)}.
\end{equation}
Note that in particular, $w^{(0)} =\sgn (x) $ and $w^{(L)} =\sgn (x)  -  v$.
In what follows, it will be shown  that the matrices $P_S - \frac{m}{m_\ell}  A^{(\ell)*}_S A^{(\ell)}_{S}$ are contractions and that the norm of the  vector $w^{(\ell)}$ decreases geometrically fast with $\ell$. Therefore, $v_S^{(\ell)}$ becomes close to $\sgn(x_S)$ as $\ell$ tends to $L$. In particular, we will prove that $\|w^{(L)} \|_2  \leq 1/4$ for a suitable choice of $L$. In addition, we also show that $v$ satisfies the condition $\| v_{S^c} \|_{\infty}  \leq 1/4$. All these conditions will be shown to be satisfied with a large probability (depending on $\varepsilon$). 

For all $1 \leq \ell \leq L$, assume that 
\begin{align}
\label{eq:control2}
\tag{C1-$\ell$}
\left\| w^{(\ell)} \right\|_2 &\leq  r_\ell  \left\| w^{(\ell-1)} \right\|_2 \\
\label{eq:control3}
\tag{C2-$\ell$}
\left\| \frac{m}{m_\ell} \left(A_{S^c}^{(\ell)}\right)^* A^{(\ell)}_{S} w^{(\ell-1)}\right\|_\infty &\leq t_\ell \| w^{(\ell-1)} \|_\infty \\
\label{eq:control4}
\tag{C3-$\ell$}
\left\| \left( \frac{m}{m_\ell} \left(A_{S}^{(\ell)}\right)^* A^{(\ell)}_{S} - P_S \right) w^{(\ell-1)}\right\|_\infty &\leq t_\ell' \| w^{(\ell-1)} \|_\infty,
\end{align}
with
\begin{enumerate}
\item $L = 2 + \left\lceil   \frac{\ln \left(s \right)}{ 2 \ln 2}   \right\rceil  $,
\item $r_\ell = \frac{1}{2},$ for $\ell=1,\hdots, L$,
\item $t_\ell = t_\ell' =  \frac{1}{5}$  for $\ell=1,\hdots, L$.
\end{enumerate}

Note that using \eqref{eq:control2}, we can write that
\begin{align}
\left\| \sgn (x_S)  - v_S \right\|_2 = \| w^{(L)}_S \|_2 \leq \| \sgn (x_S)  \|_2 \prod_{\ell=1}^L r_\ell \leq \sqrt{s} \prod_{\ell=1}^L r_\ell \leq \frac{  \sqrt{s}}{2^{L}} \leq \frac{1}{4}, 
\label{eq:conrand1}
\end{align}
where the last inequality follows from the  previously specified choice on $L$.

Furthermore, Equation \eqref{eq:control3} implies that
\begin{align}
 \| v_{S^c} \|_\infty &= \left\| \sum_{\ell=1}^L \frac{m}{m_\ell} \left(A_{S^c}^{(\ell)}\right)^* A^{(\ell)}_{S} w^{(\ell-1)} \right\|_\infty \nonumber \\
&\leq \sum_{\ell=1}^L  \left\|  \frac{m}{m_\ell}  \left(A_{S^c}^{(\ell)}\right)^*A^{(\ell)}_{S} w^{(\ell-1)} \right\|_\infty \nonumber  \\
&\leq \sum_{\ell=1}^L t_\ell \left\|   w^{(\ell-1)} \right\|_\infty
\nonumber \\ 
\nonumber
&\leq \sum_{\ell=1}^L t_\ell \prod_{j=1}^{\ell-1} t_j' \\
&=\left( \frac{1}{5}\right) \frac{1-(1/5)^L}{1-1/5} \leq \frac{1}{4}.
 \label{eq:conrand2}
\end{align}
Note that in Inequality \eqref{eq:conrand2}, the control of the operator norms $\infty \rightarrow \infty$ avoids the apparition of $\sqrt{s}$ as in the usual golfing scheme of \cite{candes2011probabilistic}. Indeed, in our proof strategy, we have used the fact that $\|w_0\|_\infty = \| \sgn (x_S )\|_\infty =1$, whereas in \cite{candes2011probabilistic} $\|w_0\|_2 = \| \sgn (x_S )\|_2 \leq \sqrt{s}$ is involved. This is a key step in the proof, since the absence of the degree of sparsity at this stage allows to derive results depending only on $S$ and not on its cardinality $s=|S|$.

We denote by $p_1(\ell)$, $p_2(\ell)$ and $p_3(\ell)$ the probabilities that the upper bounds \eqref{eq:control2}, \eqref{eq:control3} and \eqref{eq:control4} do not hold. 

Let us call "failure C" the event in which one of the $3L$ inequalities \eqref{eq:control2}, \eqref{eq:control3}, \eqref{eq:control4} is not satisfied. Then,
$$ \Pbb\left( \text{failure C} \right) \leq \sum_{\ell=1}^L \Pbb \left( \text{failure \eqref{eq:control2}}  \right) + \Pbb \left( \text{failure \eqref{eq:control3}}  \right) + \Pbb \left( \text{failure \eqref{eq:control4}}  \right).
$$
Therefore a sufficient condition for $\Pbb\left( \text{failure C} \right) \leq \varepsilon$ is $\sum_{\ell=1}^L p_1(\ell) + p_2(\ell) +p_3(\ell) \leq \varepsilon $ which holds provided that $p_1(\ell) \leq \varepsilon/3L$, $p_2(\ell) \leq \varepsilon/3L$ and $p_3(\ell) \leq \varepsilon/3L$ for every $\ell=1,\hdots,L$. 
By Lemma \ref{lem:lowDistortion}, condition $p_1(\ell) \leq \varepsilon/3L$ is satisfied if
$$ m_\ell \geq 32 \Gamma(S,\pi) \left( \ln\left(\frac{ 3L }{ \varepsilon} \right) +\frac{1}{4} \right).
$$
By Lemma \ref{lem:distortionInf2}, condition $p_2(\ell) \leq \varepsilon/3L$ is satisfied if
$$ m_\ell \geq 101 \Gamma(S,\pi)  \ln\left(\frac{ 12n L }{ \varepsilon} \right) .
$$
By Lemma \ref{lem:InfToInf}, condition $p_3(\ell) \leq \varepsilon/3L$ is satisfied if
$$ m_\ell \geq 101 \Gamma(S,\pi)  \ln\left(\frac{ 12n L }{ \varepsilon} \right) .
$$
Overall, condition 
\begin{align}
\label{ineq:condGammal}
m_\ell \geq 101 \Gamma(S,\pi)  \ln\left(\frac{ 12n L }{ \varepsilon} \right)
\end{align}
ensures that \eqref{eq:conrand1} and \eqref{eq:conrand2} are satisfied with probability $1-\varepsilon$.
Condition 
$$ m  = \sum_{\ell=1}^L m_\ell \geq 101 \left(\frac{\ln(s)}{2\ln(2)} + 3 \right) \Gamma(S ,\pi) \ln \left( 12nL \varepsilon^{-1}\right) 
$$
will imply \eqref{ineq:condGammal}. The latter condition can be simplified into
\begin{equation}
m \geq 73 \cdot {\Gamma(S ,\pi)} \ln(64s) \left(   \ln \left( \frac{9 n }{\varepsilon}\right) + \ln \ln(64s) \right).
\label{eq:condm1}
\end{equation}
 The latter condition ensures that the random vector $v$, defined by \eqref{eq:vfctq}, satisfies Assumptions \ref{ass:dual2} of Lemma \ref{lem:inexactDuality} with probability larger than $1-\varepsilon$.

Hence, we have thus shown that if  conditions  \eqref{eq:condm2}, \eqref{eq:condm3} and \eqref{eq:condm1} are satisfied, then the Assumptions  \ref{ass:dual1} and  \ref{ass:dual2} of Lemma \ref{lem:inexactDuality} simultaneously hold  with probability larger than $1-3\varepsilon$. Note that bound \eqref{eq:condm1} implies \eqref{eq:condm2} and \eqref{eq:condm3}. 

\section{Bernstein's inequalities}

\begin{thmchapter}[Scalar Bernstein Inequality]
\label{theo:scalBern}
Let $x_1, \hdots , x_m$ be independent real-valued, zero-mean, random variables such that $|x_\ell |\leq K$ almost surely for every $\ell \in \{ 1, \hdots , m \}$. Assume that $\Ebb | x_\ell |^2 \leq \sigma^2_\ell$ for $\ell \in \{ 1, \hdots , m \}$. Then for all $t>0$,
\begin{align*}
\Pbb \left( \left| \sum_{\ell=1}^m x_\ell \right| \geq t  \right) \leq 2 \exp\left(  -\frac{t^2/2}{\sigma^2 + Kt/3}\right),
\end{align*}
with $\sigma^2 \geq \sum_{\ell=1}^m \sigma_\ell^2$.
\end{thmchapter}

%
%

\begin{thmchapter}[Vector Bernstein Inequality (V1)] 
\label{theo:VectBern2}
\cite[Theorem 2.6]{candes2011probabilistic}
Let $\left( y_k \right)_{1\leq k \leq m}$ be a finite sequence of  independent random complex vectors of dimension $n$. Suppose that  $\Ebb y_{k} = 0$ and $\| y_{k} \|_{2} \leq K$ a.s.\ for some constant $K > 0$ and set $\sigma^2 \geq \sum_k \Ebb \| y_k \|_2^2$.  Let $Z = \left\| \sum_{k = 1}^{m} y_k \right\|_2 $. Then, for any $0< t \leq \sigma^2/K $, we have that
$$
 \Pbb \left( Z \geq  t \right) \leq  \exp\left( -\frac{\left( t / \sigma - 1 \right)^2}{ 4} \right) \leq \exp \left( - \frac{ t^2}{8\sigma^2 } +  \frac{1}{4}\right).
 $$
\end{thmchapter}


\begin{thmchapter}[Bernstein Inequality for self-adjoint matrices]
 \label{theo:MatBernSquare}
 Let $(Z_k)_{1 \leq k \leq n}$ be a finite sequence of  independent, random, self-adjoint matrices of dimension $d$, and let $a_{k}$ be a sequence of fixed self-adjoint matrices. Suppose that $Z_k$ is such that $\Ebb Z_k = 0$ and $\| Z_k \|_{2\rightarrow 2} \leq K$ a.s.\ for some constant $K > 0$ that is independent of $k$. Moreover, assume that  $\Ebb  Z_k^2  \preceq A_{k}^2$  for each $1 \leq k \leq n$. Define 
$$ \sigma^2 =  \left\|\sum_{k = 1}^{n}  A_{k}^2 \right\|_{2\rightarrow 2}  
$$
Then, for any $t > 0$, we have that
$$ \Pbb \left( \left\|\sum_{k = 1}^{n} Z_k  \right\|_{2\rightarrow 2} \geq t \right) \leq d \exp\left( -\frac{t^2/2}{\sigma^2 + Kt/3} \right).
$$
\end{thmchapter}

\begin{proof}
This result is as an application of the techniques developed in  \cite{tropp2012user} to obtain tail bounds for sum of random matrices. Our arguments follow those in the proof of Theorem 6.1 in \cite{tropp2012user}. We assume that $K=1$ since the general result follows by a scaling argument. Using the assumption that $\Ebb  Z_k^2  \preceq A_{k}^2$,  and by applying the arguments in the  proof of Lemma 6.7 in \cite{tropp2012user}, we obtain that
$$
\Ebb \exp \left( \theta  Z_k \right) \preceq \exp \left( g(\theta) A_{k}^2 \right), 
$$
for any real $\theta > 0$, where $g(\theta) = e^{\theta} - \theta - 1$, and the notation $\exp(A)$ denotes the matrix exponential of a self-adjoint matrix $A$ (see  \cite{tropp2012user} for further details). Therefore, by Corollary 3.7 in   \cite{tropp2012user}, it follows that 
\begin{equation}
 \Pbb \left( \left\|\sum_{k = 1}^{n} Z_k  \right\|_{2\rightarrow 2} \geq t \right) \leq d \inf_{\theta > 0} \left\{ e^{-\theta t + \sigma^2 g(\theta)} \right\}, \label{eq:proofbernstein}
\end{equation}
where $ \sigma^2 =  \left\|\sum_{k = 1}^{n}  A_{k}^2 \right\|_{2\rightarrow 2}$. To conclude, we follow the proof of  Theorem 6.1 in \cite{tropp2012user}. The function $\theta \mapsto -\theta t + \sigma^2 g(\theta)$ attains its minimum for $\theta = \ln(1+t/\sigma^2)$, which implies that the minimal value of the right-hand size of Inequality \eqref{eq:proofbernstein} is $d  \exp \left( - \sigma^2 h(t/\sigma^2) \right)$  where $h(u) = (1+u) \ln(1+u) - u$ for $u \geq 0$. To complete the proof, it suffices to use the standard lower bound $h(u) \geq \frac{u^2/2}{1+u/3}$ for $u \geq 0$.
\end{proof}

\section{Estimates: auxiliary results} \label{sec:aux}

Let $S$ be the support of the signal to be reconstructed such that $|S|=s$. We set 
$$ \Lambda(S , \pi) := \max_{1\leq k \leq M}  \frac{1}{\pi_k}  \left\| B_{k,S}^* B_{k,S} \right\|_{2\rightarrow 2}.
$$
Note that $  \left\| B_{k,S}^* B_{k,S} \right\|_{2\rightarrow 2} \leq  \left\| B_{k,S}^* B_{k,S} \right\|_{\infty\rightarrow \infty} \leq \left\| B_{k}^* B_{k,S} \right\|_{\infty\rightarrow \infty},$
therefore, 
$$ \Lambda(S,\pi)\leq \Theta(S,\pi).
$$
To make the notation less cluttered, we will write $\Lambda$, $\Theta$, $\Upsilon$ and $\Gamma$ instead of $\Lambda(S,\pi)$, $\Theta(S,\pi)$, $\Upsilon(S,\pi)$ and $\Gamma(S,\pi)$.

\begin{lemme}
\label{lem:localIsometry}
Let $S \subset \{1,\hdots , n\}$ be of cardinality of $s$. Suppose that $\Theta\geq 1$. Then, for any $\delta >0$, one has that
\begin{align}
\tag{E1}
\label{lem:E1eq}
\displaystyle \Pbb \left( \| A_S^* A_S - P_S \|_{2\rightarrow 2} \geq \delta \right) &\leq 2 s \exp \left( -  \frac{m\delta^2/2}{\Theta (1+\delta /3)} \right).
\end{align}
\end{lemme}
\begin{proof}
We decompose the matrix $A_S^* A_S-P_S$ as 
$$
A_S^* A_S- P_S= \frac{1}{m} \sum_{k=1}^m \frac{ B^*_{J_k,S} B_{J_k,S}}{\pi_{J_k}} - P_S  =\frac{1}{m} \sum_{k=1}^m X_{k},
$$
where $X_{k} :=\displaystyle \left( \frac{ B^*_{J_k,S} B_{J_k,S}}{\pi_{J_k}} - P_S\right)$.
It is clear that $\Ebb X_{k} = 0$, and since for all $1\leq k \leq M$, 
$
\frac{\left\|  B_{k,S} B^*_{k,S} \right\|_{2\rightarrow 2}}{\pi_k} \leq  \Lambda \leq \Theta,
$
we have that
$$\|X_{k} \|_{2\rightarrow 2} \leq  \max \left( \frac{\left\|  B^*_{J_k,S} B_{J_k,S} \right\|_{2\rightarrow 2}}{\pi_{J_k}}  -1 ,1\right)\leq   \Theta.$$
Lastly, we remark that
\begin{align*}
0 \quad \preceq \quad  \Ebb X_{k}^2 = \Ebb\left[  \frac{B^*_{J_k,S} B_{J_k,S}}{\pi_{J_k}} \right]^2 - P_S \quad &\preceq \quad \max_{1\leq k \leq M} \frac{\left\|  B^*_{k,S}  B_{k,S}\right\|_{2\rightarrow 2}}{\pi_k} \Ebb\left[  \frac{B^*_{J_k,S} B_{J_k,S}}{\pi_{J_k}} \right] \\
&\preceq \max_{1\leq k \leq M} \frac{\left\|  B^*_{k,S}  B_{k,S}\right\|_{2\rightarrow 2}}{\pi_k} P_S \preceq   \Lambda  P_S \\
&\preceq \Theta P_S.
\end{align*}
Therefore, using Theorem \ref{theo:MatBernSquare}, we can set $\sigma^2= \left\| \sum_{k=1}^{m} \Ebb X_k^2 \right\|_{2 \rightarrow 2} \leq m\Theta$.
Hence, inequality \eqref{lem:E1eq} immediately follows from  Bernstein's inequality for random matrices (see Theorem \ref{theo:MatBernSquare}).
\end{proof}

\begin{lemme}
\label{lem:lowDistortion}
Let $S \subset \{1, \hdots, n\}$, such that $|S|=s$. Let $w$ be a vector in $\Cbb^n$. Then, for any $0 \leq t \leq 1$, one has that
\begin{align}
\tag{E2}
\label{lem:E2eqt}
\Pbb & \left( \left\| \left( A_S^{*} A_S - P_S   \right) w \right\|_2\geq  t  \|w\|_2 \right)  \leq \exp \left( - \frac{m t^2}{8 \Theta} + \frac{1}{4} \right).\nonumber
\end{align}
\end{lemme}

\begin{proof}
Without loss of generality we may assume that $\| \wb \|_2 =1$. We remark that
$$
\left( A_S^* A_S - \Id_s \right) w_S =  \frac{1}{m} \sum_{k=1}^m \left( \frac{ B_{J_k,S}^* B_{J_k,S}}{\pi_{J_k}}  - P_S\right) w = \frac{1}{m} \sum_{k=1}^m  y_k,  
$$
where $ y_k = \left( \frac{ B_{J_k,S}^* B_{J_k,S}}{\pi_{J_k}}  - P_S\right) w$ is a random vector with zero mean.  Simple calculations yield that
\begin{align*}
 \left\| \frac{1}{m} y_k \right\|_2^2 &= \frac{1}{m^2} \left(  w^* \left( \frac{ B^*_{J_k,S} B_{J_k,S}}{\pi_{J_k}}  \right)^2 w   -2 w^* \frac{ B^*_{J_k,S} B_{J_k,S}}{\pi_{J_k}}  w +  w^*w \right) \\
 &\leq \frac{1}{m^2} \left( \Lambda w^*  \frac{ B^*_{J_k,S} B_{J_k,S}}{\pi_{J_k}}  w   -2 w^* \frac{ B^*_{J_k,S} B_{J_k,S}}{\pi_{J_k}}  w + 1\right) \\
 &= \frac{1}{m^2} \left( \left( \Lambda -2 \right) w^* \frac{ B^*_{J_k,S} B_{J_k,S}}{\pi_{J_k}}  w + 1  \right) \\
 &\leq \frac{1}{m^2} \left( \left( \Lambda -2 \right) \Lambda \| w\|_2^2  + 1  \right) = \frac{1}{m^2} \left( \left( \Lambda -2 \right) \Lambda+ 1  \right)  \\
 &\leq \frac{1}{m^2}  \left( \Lambda- 1  \right)^2 \leq \frac{1}{m^2}  \Lambda^2 \leq \frac{1}{m^2} \Theta^2.
\end{align*}
Now, let us define 
$
Z = \left\| \frac{1}{m} \sum_{k=1}^m  y_k\right\|_2.
$
By independence of the random vectors $y_{k}$, it follows that
\begin{align*}
\Ebb \left[ Z^2 \right] &= \frac{1}{m} \Ebb \left\|   y_1 \right\|^2_2 = \frac{1}{m} \Ebb \left[ \left\langle  \frac{ B^*_{J,S} B_{J,S}}{\pi_{J}} w ,   \frac{ B^*_{J,S} B_{J,S}}{\pi_{J}} w \right\rangle - 2 \left\langle \frac{ B^*_{J,S} B_{J,S}}{\pi_{J}}  w, w \right\rangle + \left\langle w , w \right\rangle   \right] \\
&= \frac{1}{m} \Ebb \left[ \left\langle \left(  \frac{ B^*_{J,S} B_{J,S}}{\pi_{J}}  \right)^2 w , w \right\rangle - 2 \frac{\left\|    B_{J,S}  w \right\|_2^2}{\pi_J} + 1   \right]. 
\end{align*}
To bound the first term in the above equality, one can write
\begin{align*} 
\Ebb &\left[ \left\langle \left(  \frac{ B^*_{J_1,S} B_{J_1,S}}{\pi_{J_1}} \right)^2 w ,w \right\rangle \right] = \left\langle \Ebb\left[\left(  \frac{ B^*_{J_1,S} B_{J_1,S}}{\pi_{J_1}} \right)^2\right] w , w \right\rangle \\
&\leq \Lambda \left\langle \Ebb\left[\left(  \frac{ B^*_{J_1,S} B_{J_1,S}}{\pi_{J_1}}  \right)\right] w , w \right\rangle \leq   \Lambda \| w\|_2^2 \leq  \Theta. 
\end{align*}
One immediately has that
$
\Ebb \frac{\left\| B_{J,S}  w  \right\|_2^2}{\pi_k}  =  \| w \|_2^2 = 1.
$
Therefore, one finally obtains that
$$  \Ebb \left[ Z^2 \right] \leq \frac{\Theta -1}{m} \leq \frac{\Theta}{m}.
$$
Using the above upper bounds, namely $ \left\| \frac{1}{m} y_k  \right\|_2 \leq  \frac{\Theta }{m} $ and $\Ebb \left[ Z^2 \right] \leq \frac{\Theta}{m}$, the result of the lemma is thus a consequence of the Bernstein's inequality for random vectors  (see Theorem \ref{theo:VectBern2}), which completes the proof.
 \end{proof}

\begin{lemme}
\label{lem:distortionInf2}
Let $S \subset \{1, \hdots, n\}$, such that $|S|=s$. Let $v$ be a vector of $\Cbb^n$. Then we have
\begin{align}
\tag{E3}
\label{lem:E3eq}
\Pbb\left( \left\| A_{S^c}^* A_S v \right\|_\infty \geq t \| v\|_\infty \right) \leq 4n \exp\left( -\frac{m t^2/4}{\Upsilon + \Theta t/3} \right).
\end{align}
\end{lemme}

\begin{proof}
Suppose without loss of generality that $\| v\|_\infty=1$.
Then,
\begin{align*} 
 \left\| A_{S^c}^* A_S v \right\|_\infty &= \max_{i \in S^c} \left| \left\langle e_i ,  A^* A_S v \right\rangle \right|  = \max_{i \in S^c } \frac{1}{m} \left| \sum_{k = 1}^m\left\langle e_i ,  \frac{B^*_{J_k} B_{J_k,S}}{\pi_{J_k}}  v \right\rangle \right|.
\end{align*}
Let us define $Z_k = \left\langle e_i ,  \frac{B^*_{J_k} B_{J_k,S}}{\pi_{J_k}}  v \right\rangle$. Note that $\Ebb  Z_k =0$, since for $i\in S^c$, $\Ebb \left\langle e_i ,  \frac{B^*_{J_k} B_{J_k,S}}{\pi_{J_k}} v \right\rangle= e_i^* \sum_{k=1}^M \pi_k \frac{B^*_{k} B_{k,S}}{\pi_{k}} v = e_i^* P_S v = 0$.
From Holder's inequality, we get
\begin{align*}
|Z_k | &= \left|  \left\langle e_i ,  \frac{B^*_{J_k} B_{J_k,S}}{\pi_{J_k}}  v \right\rangle \right|  = \left| e_i^*   \frac{B^*_{J_k} B_{J_k,S}}{\pi_{J_k}}   v  \right|  \leq  \max_{j\in S^c \atop 1\leq k \leq M}    \frac{1}{\pi_k}\left\|B_{k,S}^* B_k e_j \right\|_1  \|v\|_\infty \\
&\leq  \max_{j\in S^c \atop 1\leq k \leq M}  \frac{1}{\pi_k}  \left\|e_j^* B_k^*B_{k,S} \right\|_1 = \Theta.
\end{align*}
Furthermore,
\begin{align*} 
\Ebb | Z_k |^2 &=   \Ebb \left| \left\langle e_i ,  \frac{B^*_{J_k} B_{J_k,S}}{\pi_{J_k}}  v \right\rangle \right|^2 = \sum_{\ell=1}^M \frac{\left| e_i^* B^*_{\ell} B_{\ell,S} v \right|^2}{\pi_\ell}\\
 &\leq  \Upsilon.
\end{align*}
Therefore $\sum_{k=1}^m \Ebb | Z_k |^2 \leq  m \Upsilon$.
Using real-valued Bernstein's inequality \ref{theo:scalBern} in the case of complex random variables, we obtain
\begin{align*}
\Pbb &\left(  \frac{1}{m}\left| \sum_{k = 1}^m \left\langle e_i ,  \frac{B^*_{J_k} B_{J_k,S}}{\pi_{J_k}}  v \right\rangle \right| \geq t \right) \\
&\leq  \Pbb \left( \frac{1}{m}\left|  \sum_{k = 1}^m \text{Re} \left\langle e_i ,  \frac{B^*_{J_k} B_{J_k,S}}{\pi_{J_k}}  v \right\rangle \right| \geq t / \sqrt{2}\right) ...\\
& \qquad + \Pbb \left( \frac{1}{m}\left|  \sum_{k = 1}^m \text{Im} \left\langle e_i ,  \frac{B^*_{J_k} B_{J_k,S}}{\pi_{J_k}}  v \right\rangle \right| \geq t /\sqrt{2} \right)
\\
&\leq 4 \exp\left( -\frac{ m t^2/4}{\Upsilon + \Theta t/3} \right).
\end{align*}
Taking the union bound over $i\in S^c$ completes the proof.
\end{proof}

\begin{lemme}
\label{lem:InfToInf}
Let $S \subset \{1, \hdots, n\}$, such that $|S|=s$. Suppose that $\Theta\geq 1$.Let $v$ be a vector of $\Cbb^n$. Then we have
\begin{align}
\tag{E4}
\label{lem:E4eq}
\Pbb\left( \left\| \left( A_{S}^* A_S - P_S  \right) v \right\|_\infty \geq t \| v\|_\infty \right) \leq 4s \exp\left( -\frac{m t^2/4}{\Upsilon + \Theta t/3} \right).
\end{align}
\end{lemme}

\begin{proof}
Suppose without loss of generality that $\| v\|_\infty=1$.
Then,
\begin{align*} 
 \left\| \left( A_{S}^* A_S - P_S  \right) v \right\|_\infty &= \max_{i \in S} \left|\left\langle e_i ,  \left( A_S^* A_S - P_S  \right) v \right\rangle  \right|= \max_{i \in S } \frac{1}{m} \left| \sum_{k = 1}^m \left\langle e_i ,  \left( \frac{B^*_{J_k,S} B_{J_k,S}}{\pi_{J_k}} -P_S \right)  v \right\rangle \right|.
\end{align*}
Let us define $Z_k = \left\langle e_i ,  \left( \frac{B^*_{J_k,S} B_{J_k,S}}{\pi_{J_k}} -P_S\right)  v \right\rangle$. Note that $\Ebb Z_k =0$.
From Holder's inequality, we get
\begin{align*}
|Z_k | &= \left| \left\langle e_i ,  \left( \frac{B^*_{J_k,S} B_{J_k,S}}{\pi_{J_k}} -P_S\right)  v \right\rangle \right|  
\leq \left\| \frac{B^*_{J_k,S} B_{J_k,S}}{\pi_{J_k}} -P_S \right\|_{\infty \rightarrow \infty} {\leq \max(\Theta -1 , 1) \leq  \Theta,}
\end{align*}
since $\| B_{k,S}^* B_{k,S} \|_{\infty \rightarrow \infty} \leq \| B_{k}^* B_{k,S} \|_{\infty \rightarrow \infty}$, and using the same argument as in Lemma \ref{lem:distortionInf2}.
Furthermore,
\begin{align*}
\Ebb | Z_k |^2 &=  \Ebb \left| \left\langle e_i , \left( \frac{B^*_{J_k,S} B_{J_k,S}}{\pi_{J_k}} -P_S\right)  v \right\rangle \right|^2   \\
&= \Ebb \left| \left\langle e_i ,   \frac{B^*_{J_k,S} B_{J_k,S}}{\pi_{J_k}}   v \right\rangle \right|^2 -  \left\langle e_i , v \right\rangle \Ebb \left\langle e_i ,   \frac{B^*_{J_k,S} B_{J_k,S}}{\pi_{J_k}}  v \right\rangle -  \left\langle e_i , v \right\rangle^* \Ebb \left\langle e_i ,   \frac{B^*_{J_k,S} B_{J_k,S}}{\pi_{J_k}}  v \right\rangle \\
& \qquad +  \left|\left\langle e_i , v \right\rangle\right|^2\\
&= \Ebb \left| \left\langle e_i ,   \frac{B^*_{J_k,S} B_{J_k,S}}{\pi_{J_k}}  v \right\rangle \right|^2 - \left| \left\langle e_i , v \right\rangle\right|^2 \leq \Ebb \left|\left\langle e_i ,  
\frac{B^*_{J_k,S} B_{J_k,S}}{\pi_{J_k}}  v \right\rangle\right|^2 =\sum_{\ell=1}^M \frac{\left| e_i^*  {B^*_{\ell,S} B_{\ell,S}} v \right|^2}{\pi_\ell} \\
&\leq \Upsilon.
\end{align*}
Therefore, $\sum_{k=1}^m \Ebb |Z_k|^2  \leq m\Upsilon$, and using real-valued Bernstein's inequality \ref{theo:scalBern} in the case of complex random variables, we obtain
\begin{align*}
\Pbb &\left( \frac{1}{m}\left| \sum_{k = 1}^m \left\langle e_i ,  \left( \frac{B^*_{J_k,S} B_{J_k,S}}{\pi_{J_k}} -P_S\right) v \right\rangle \right| \geq t \right) \\
&\leq  \Pbb \left( \frac{1}{m} \left|  \sum_{k = 1}^m \text{Re} \left\langle e_i ,  \left( \frac{B^*_{J_k,S} B_{J_k,S}}{\pi_{J_k}} -P_S\right)  v \right\rangle \right| \geq t / \sqrt{2}\right)  + \Pbb \left( \frac{1}{m} \left|  \sum_{k = 1}^m \text{Im} \left\langle e_i , \left( \frac{B^*_{J_k,S} B_{J_k,S}}{\pi_{J_k}} -P_S\right) v \right\rangle \right| \geq t /\sqrt{2} \right)
\\
&\leq 4 \exp\left( -\frac{mt^2/4}{\Upsilon + \Theta t/3} \right).
\end{align*}
Taking the union bound over $i\in S$ completes the proof.
\end{proof}

\begin{lemme}
\label{lem:offSupportCoherence}
Let $S$ be a subset of $\{1, \hdots , n \}$. Then, for any $0 \leq t \leq  m$, one has that
\begin{align}
\tag{E5}
\label{lem:E5eq}
\Pbb & \left( \max_{i\in {S^c}} \left\| A_S^* A e_{i} \right\|_{2} \geq   t \right)  \leq n \exp \left( - \frac{\left( \sqrt{m} t / \sqrt{\Theta}  - 1\right)^2}{4} \right).
\end{align}
\end{lemme}

\begin{proof}
Let us fix some $i \in  {S^c}$. For $k=1, \hdots, M$, we define the random vector $$x_k :=  \frac{B^*_{J_k,S} B_{J_k}}{\pi_{J_k}} e_{i} .$$
Then, since  $i \in S^c$ one easily gets $\Ebb x_k = \sum_{\ell=1}^M B^*_{\ell,S} B_{\ell}e_{i} = \sum_{\ell=1}^M \left( B_{\ell} P_S\right)^* B_{\ell}e_{i}  = P_S \sum_{\ell=1}^M B^*_{\ell} B_{\ell}e_{i} = P_S e_i =0$ (note that $P_S$ is self-adjoint). In addition, we can write
\begin{align*}
\left\| A_S^* A e_{i} \right\|_{2}  = \left\|  \frac{1}{m} \sum_{k=1}^m \frac{B^*_{J_k,S} B_{J_k}}{\pi_{J_k}} e_{i} \right\|_{2 }   =   \left\| \frac{1}{m}\sum_{k=1}^M x_k \right\|_{2}.
\end{align*}
Then,
\begin{align*}
\|x_k\|_{2 }  &= \left\| 
\frac{B^*_{J_k,S} B_{J_k}}{\pi_{J_k}}e_{i} \right\|_{2}  \leq \left\| 
\frac{B^*_{J_k,S} B_{J_k}}{\pi_{J_k}}e_{i} \right\|_{1}  =  \left\| e_i^* \frac{B_{J_k}^* B_{J_k,S}}{\pi_{J_k}}\right\|_1 \\ 
&\leq \frac{1}{\pi_{J_k}} \left\| B_{J_k}^* B_{J_k,S} \right\|_{\infty\rightarrow\infty} \leq \Theta.
\end{align*}
Furthermore, one has that
\begin{align*}
\Ebb \left\| x_k \right\|_2^2 &=  \Ebb  \left\|\frac{B^*_{J_k,S} B_{J_k}}{\pi_{J_k}} e_i \right\|_2^2  \leq \Ebb  {\left\| \frac{B_{J_k,S}}{\sqrt{\pi_{J_k}}} \right\|_{2\rightarrow 2}^2} \left\|\frac{B_{J_k}}{\sqrt{\pi_{J_k}}} e_i \right\|_2^2\leq    \Lambda \Ebb \left\|\frac{B_{J_k}}{\sqrt{\pi_{J_k}}} e_i \right\|_2^2 = \Lambda \| e_i\|_2^2 = \Lambda, \\
\sum_{k=1}^m \Ebb \left\| x_k \right\|_2^2 &\leq m\Lambda  \leq m\Theta.
\end{align*}
Hence, using the above upper bounds, it follows from Bernstein's inequality for random vectors (see Theorem \ref{theo:VectBern2}) that
\begin{align*}
\Pbb \left( \left\| A_S^* A e_{i} \right\|_{2}  \geq  t\right) \leq  \exp \left( - \frac{\left(\sqrt{m} t / \sqrt{\Theta} - 1\right)^2}{4} \right),
\end{align*}
Finally, Inequality \eqref{lem:E4eq} follows from a union bound over  $i\in {S^c}$, which completes the proof.
\end{proof}

\section{Proof of results in Applications}

\subsection{Proof of Corollary \ref{corol:isolatedUniform}}
\label{app:proofIsolatedCandes}

The proof relies on the evaluation of $\Theta$ and $\Upsilon$ in the case of isolated measurements. 
In this case, we have $n$ blocks composed of isolated measurements. Then, each block corresponds to one of the rows $(a_k^*)_{1\leq k \leq n}$ of $A_0$. Recall that $\|a_k a_{k,S}^* \|_{\infty \rightarrow \infty} = \max_{1\leq i \leq n} \sup_{\| v \|_\infty \leq 1} | e_i^* a_k a_{k,S}^* v|$, so the norm $\| a_k a_{k,S}^* \|_{\infty \rightarrow \infty}$ is the maximum $\ell_1$-norm of the rows of the matrix $a_k a_{k,S}^*$. Therefore, the quantities in Definition \ref{def:quantities} can be rewritten as follows
\begin{align}
\notag
\Theta(S,\pi) &:= \max_{1\leq k \leq n} \frac{\| a_k a_{k,S}^* \|_{\infty \rightarrow \infty}}{\pi_k} \\
\label{eq:thetaIso}
&= \max_{1\leq k \leq n} \frac{\| a_k \|_{\infty} \| a_{k,S} \|_1}{\pi_k} \\
\notag
&\leq s \cdot \max_{1\leq k \leq n} \frac{\| a_k \|_{\infty}^2}{\pi_k}, \\
\label{eq:UpsiIso}
\Upsilon(S,\pi) &= \max_{1\leq i \leq n} \sup_{\|v \|_\infty\leq 1} \sum_{k=1}^n \frac{1}{\pi_k} |e_i^* a_k|^2 |a_{k,S}^*v|^2 \\
\notag
&\leq  \sup_{\|v \|_\infty\leq 1} \sum_{k=1}^n \frac{1}{\pi_k} \| a_k \|_\infty^2 |a_{k,S}^*v|^2 \\
\notag
& \leq \sup_{\|v \|_\infty\leq 1} \max_{1\leq \ell \leq n }    \frac{\| a_\ell \|_\infty^2}{\pi_\ell}  \sum_{k=1}^n  |a_{k,S}^*v|^2 =  \sup_{\|v \|_\infty\leq 1}  \| A_0 P_S v \|_2^2 \max_{1\leq \ell \leq n }    \frac{\| a_\ell \|_\infty^2}{\pi_\ell} \\
\notag
&=   \sup_{\|v \|_\infty\leq 1}  \|  P_S v \|_2^2 \max_{1\leq \ell \leq n }    \frac{\| a_\ell \|_\infty^2}{\pi_\ell} \\
\notag
&\leq s \cdot \max_{1\leq k \leq n} \frac{\| a_k \|_{\infty}^2}{\pi_k}.
\end{align}
Therefore we can choose $\Gamma(S,\pi) = s \cdot \max_{1\leq k \leq n} \frac{\| a_k \|_{\infty}^2}{\pi_k}$, and the result follows by Theorem \ref{thm:recovery}.

%
%

\subsection{Around Corollary \ref{corol:isolatedBlocksSparsity}}

\subsubsection{Proof of Corollary \ref{corol:isolatedBlocksSparsity}}
\label{app:proofIsolatedBlocksSparsity}

Again, this is all about evaluating $\Theta$ and $\Upsilon$ in this specific case. Concerning the evaluation of $\Upsilon$, we can use the expression \eqref{eq:UpsiIso}
to conclude that
$$ \Upsilon(S,\pi) =\max_{1\leq i \leq n} \sup_{\|v \|_\infty\leq 1} \sum_{k=1}^n \frac{1}{\pi_k} |e_i^* a_k| ^2 |a_{k,S}^*v|^2.
$$
To control $\Theta$, using \eqref{eq:thetaIso}, it suffices to write:
\begin{align*}
\Theta(S,\pi) &=  \max_{1\leq k \leq n} \frac{\| a_k a_{k,S}^* \|_{\infty \rightarrow \infty}}{\pi_k} \leq \max_{1\leq k \leq n} \frac{\| a_k \|_{\infty} \| a_{k,S} \|_1}{\pi_k} \\
&\leq \max_{1\leq k \leq n} \frac{\| a_k \|_{\infty}\sum_{\ell=1}^N \| a_{k,\Omega_\ell}\|_\infty s_\ell}{\pi_k}.
\end{align*}
By Theorem \ref{thm:recovery}, the two conditions 
\begin{eqnarray*}
m &\geq & C \left(\max_{1\leq k \leq n} \frac{\sum_{\ell=1}^N s_\ell \|a_{k,\Omega_\ell} \|_\infty \|a_k \|_\infty}{\pi_k}  \right)\ln(s) \ln \left( \frac{n }{\varepsilon}\right), \\
m &\geq & C \left(\max_{1\leq i \leq n}\sup_{\|v\|_\infty \leq 1 }\sum_{k=1}^n   \frac{1}{\pi_k}  \left| e_i^*  a_{k}\right|^2 \left| a_{k,S}^* v  \right|^2\right) \ln(s) \ln \left( \frac{ n }{\varepsilon}\right), 
\end{eqnarray*}
lead to the desired conclusion.

\subsubsection{Comparison of Corollary \ref{corol:isolatedBlocksSparsity} and the results in \cite{adcock2013breaking}.}
\label{app:comparisonAdcock}

Note that the sampling in \cite{adcock2013breaking} is based on Bernoulli drawings structured by level. Their results are then easily transposable to the case of i.i.d.\ sampling with constant probability by level.
The first condition on $m$ in Corollary \ref{corol:isolatedBlocksSparsity} is similar to condition (4.4) in Theorem 4.4 of \cite{adcock2013breaking}, since we recognize the term $\frac{\|a_{k,\Omega_\ell} \|_\infty \|a_k \|_\infty}{\pi_k}$ as the $(k,\ell)$-local coherence defined in \cite{adcock2013breaking}. 
Let us show that the second condition on $m$ is similar to equation (4.5) in \cite{adcock2013breaking}. 
First, observe that
\begin{align*}
\max_{1\leq i \leq n} \sup_{\|v\|_\infty \leq 1 }\sum_{k=1}^n   \frac{1}{\pi_k}  \left| e_i^*  a_{k}\right|^2 \left| a_{k,S}^* v  \right|^2 
&\leq \max_{1\leq \ell \leq N} \sup_{\|v\|_\infty \leq 1 }\sum_{k=1}^n   \frac{1}{\pi_k}  \left\|a_{k, \Omega_\ell}\right\|_\infty^2 \left| a_{k,S}^* v  \right|^2 \\
&\leq   \max_{1\leq \ell \leq N} \sup_{\|v\|_\infty \leq 1 }\sum_{k=1}^n   \frac{1}{\pi_k}  \left\|a_{k, \Omega_\ell}\right\|_\infty \left\|a_{k}\right\|_\infty \left| a_{k,S}^* v  \right|^2.
\end{align*}
Let $\tilde{v}$ denote the maximizer in the last expression, and define $\widetilde{s_k} = \left| a_{k,S}^* \tilde{v}  \right|^2$ for $1\leq k \leq n$. 
It follows,
 \begin{align}
 \label{eq:verifAdcock}
\max_{1\leq i \leq n} \sup_{\|v\|_\infty \leq 1 }\sum_{k=1}^n   \frac{1}{\pi_k}  \left| e_i^*  a_{k}\right|^2 \left| a_{k,S}^* v  \right|^2 
&\leq   \max_{1\leq \ell \leq N} \sum_{k=1}^n   \frac{1}{\pi_k}  \left\|a_{k, \Omega_\ell}\right\|_\infty \left\|a_{k}\right\|_\infty \widetilde{s_k},
\end{align}
and $\sum_{k=1}^n \widetilde{s_k} = \sum_{k=1}^n \left| a_{k,S}^* v  \right|^2 = \| A_0 P_S \widetilde{v} \|_2^2 = \| P_S \widetilde{v} \|_2^2 \leq \sum_{\ell=1}^N s_\ell$. 
The last inequality and Equation \eqref{eq:verifAdcock} for i.i.d\ sampling correspond to the condition (4.5) in Theorem 4.4 of \cite{adcock2013breaking} in the case of Bernoulli sampling.
This completes the comparison between Corollary \ref{corol:isolatedBlocksSparsity} and the results in \cite{adcock2013breaking}.

\subsection{Proof of Corollary \ref{corol:isolatedMRI}}
\label{app:proofIsolatedMRI}


Recall that $\left( \Omega_j\right)_{0\leq j \leq J}$ the dyadic partition of the set of indexes $\{1,\hdots ,n\}$. Recall also the function $j : \{ 1 , \hdots , n \} \rightarrow \{ 0 , \hdots , J\}$ defined by $j(u) = j$ if $u\in \Omega_{j}$. In the interests of simplifying notation, in this section, the symbol '$\gtrsim$' will be equivalent to '$\geq C \cdot$',  with $C$ a universal constant.
The following lemma will be useful to bound above the coefficients of $A_0$ in absolute value, and to derive Lemmas \ref{lem:thetaIsoMRI} and \ref{lem:upsilonIsoMRI}.

\begin{lemme}{\cite{adcock2014note}}
\label{lem:adcock1}
The magnitude of the coefficients of matrix $A_0=\Fc\phi^*$, where $\Fc$ is the 1D Fourier transform and $\phi$ is the 1D Haar transform,  satisfies
\begin{align}
\label{estimCoherence}
 \| P_{\Omega_j} A_0 P_{\Omega_\ell} \|_{1\rightarrow \infty}^2 \lesssim 2^{-j} 2^{-|j-\ell |}, \quad \text{for} \quad 0\leq j,\ell\leq J.
\end{align}
\end{lemme}

\begin{lemme}
\label{lem:thetaIsoMRI}
In the case of isolated measurements, with $A_0 = \Fc\phi^*$ with $\phi$ to be the inverse 1D Haar transform, suppose that the signal to reconstruct $x$ is sparse by levels, meaning that $\| P_{\Omega_j} x \|_0 \leq s_j$ for $0\leq j \leq J$. Then,
\begin{align}
\Theta \lesssim \max_{1\leq k \leq n} \frac{2^{-j(k)}}{\pi_k} \left( s_{j(k)} + \sum_{\ell = 0 \atop \ell \neq j(k)}^J s_\ell 2^{-|j(k) - \ell|/2} \right).
\end{align}
Choosing $\pi_k$ to be constant by level, i.e. $\pi_k = \tilde{\pi}_{j(k)}$, the last expression can be rewritten as follows
\begin{align}
\Theta \lesssim \max_{0\leq j \leq J} \frac{2^{-j}}{\tilde{\pi}_j} \left( s_{j} + \sum_{\ell = 0 \atop \ell \neq j}^J s_\ell 2^{-|j - \ell|/2} \right).
\end{align}
\end{lemme} 
 \begin{proof} Using \eqref{eq:thetaIso}, we can write
 \begin{align*}
 \Theta &= \max_{1\leq k \leq n} \frac{\| a_k \|_{\infty} \| a_{k,S} \|_1}{\pi_k} \leq \max_{1\leq k \leq n} \frac{\| a_k \|_{\infty} \sum_{\ell=0}^J \| a_{k,\Omega_\ell} \|_\infty s_\ell}{\pi_k} \\
 &\lesssim \max_{1\leq k \leq n} \frac{1}{\pi_k}  2^{-j(k)/2} \sum_{\ell=0}^J 2^{-j(k)/2} 2^{-|j(k)-\ell|/2}  s_\ell  \\
 &\lesssim \max_{1\leq k \leq n}\frac{1}{\pi_k}  2^{-j(k)} \sum_{\ell=0}^J  2^{-|j(k)-\ell|/2}  s_\ell ,
 \end{align*}
 where we use \eqref{estimCoherence} to bound above $\| a_{k,\Omega_\ell} \|_\infty$.
 \end{proof}
 
\begin{lemme}
\label{lem:upsilonIsoMRI}
In the case of isolated measurements, with $A_0 = \Fc\phi^*$ with $\phi$ to be the inverse Haar transform, suppose that the signal to reconstruct $x$ is sparse by levels, meaning that $\| P_{\Omega_j} x \|_0 \leq s_j$ for $0\leq j \leq J$. Choosing $\pi_k$ to be constant by level, i.e. $\pi_k = \tilde{\pi}_{j(k)}$, we have
\begin{align}
\Upsilon \lesssim \max_{0\leq j \leq J} \frac{1}{\tilde{\pi}_j} 2^{-j} \sum_{p=0}^J 2^{-|j-p|/2} s_p.
\end{align}
\end{lemme} 
\begin{proof}
Denoting $\tilde{v} = \tilde{v}(i)$ the argument of the supremum in the definition of $\Upsilon$, we get
\begin{align*}
\Upsilon &:= \max_{1\leq i \leq n} \sum_{k=1}^n \frac{1}{\pi_k} | e_i^* a_k |^2 | a_{k,S} \tilde{v} |^2 \leq \max_{0\leq \ell \leq J} \sum_{k=1}^n \frac{1}{\pi_k} \| a_{k,\Omega_\ell} \|_\infty^2 | a_{k,S} \tilde{v} |^2 \\
&\lesssim \max_{0\leq \ell \leq J} \sum_{k=1}^n \frac{1}{\pi_k} 2^{-j(k)} 2^{-|j(k) - \ell |}  | a_{k,S} \tilde{v} |^2 \lesssim \max_{0\leq \ell \leq J} \sum_{j=0}^J \frac{1}{\tilde{\pi}_j} 2^{-j} 2^{-|j - \ell |} \underbrace{\sum_{k\in \Omega_j}  | a_{k,S} \tilde{v} |^2}_{=:K_j} 
\end{align*}
We can rewrite $K_j$ as follows $K_j = \| P_{\Omega_j} A_0 P_S \tilde{v} \|_2^2$. Therefore, since $\|\tilde{v}\|_\infty \leq 1$,
\begin{align*}
\sqrt{K_j} &= \| P_{\Omega_j} A_0 P_S \tilde{v} \|_2 = \| P_{\Omega_j} A_0 \sum_{p=0}^J P_{\Omega_p} P_S \tilde{v} \|_2 \leq \sum_{p=0}^J \| P_{\Omega_j} A_0  P_{\Omega_p} P_S \tilde{v} \|_2 \\
&\leq \sum_{p=0}^J \| P_{\Omega_j} A_0  P_{\Omega_p} \|_{2\rightarrow 2} \| P_{\Omega_p} P_S\tilde{v} \|_{2 } \leq \sum_{p=0}^J \| P_{\Omega_j} A_0  P_{\Omega_p} \|_{2\rightarrow 2} \sqrt{s_p}.
\end{align*}
Using Lemma 4.3 of \cite{adcock2014note}, we have the following upper bound
$$ \| P_{\Omega_j} A_0  P_{\Omega_p} \|_{2\rightarrow 2}  \lesssim 2^{-|j-p|/2}, \qquad \text{for} \quad 0 \leq j,p \leq J.
$$
Then, $\sqrt{K_j} \lesssim \sum_{p=0}^J 2^{-|j-p|/2} \sqrt{s_p}$, and thus
\begin{align*}
K_j &\lesssim \left( \sum_{p=0}^J 2^{-|j-p|/2} \sqrt{s_p} \right)^2  \lesssim \left( \sum_{p=0}^J 2^{-|j-p|/2}  \right)\left( \sum_{p=0}^J 2^{-|j-p|/2} s_p \right) \\
& \lesssim \left( \sum_{p=0}^J 2^{-|j-p|/2} s_p \right)
\end{align*}
where in the second inequality we use Cauchy-Schwarz inequality. Therefore,
\begin{align*}
\Upsilon &\lesssim \max_{0\leq \ell \leq J} \sum_{j=0}^J  2^{-|j-\ell |} \frac{1}{\tilde{\pi}_j} 2^{-j} \sum_{p=0}^J 2^{-|j-p|/2} s_p \\
&\lesssim  \left( \max_{0\leq \ell \leq J} \sum_{j=0}^J  2^{-|j-\ell |} \right) \left( \max_{0\leq j\leq J} \frac{1}{\tilde{\pi}_j} 2^{-j} \sum_{p=0}^J 2^{-|j-p|/2} s_p \right)\\
&\lesssim \max_{0\leq j \leq J} \frac{1}{\tilde{\pi}_j} 2^{-j} \sum_{p=0}^J 2^{-|j-p|/2} s_p.
\end{align*}
\end{proof}
Note that the upper bounds given in Lemmas \ref{lem:thetaIsoMRI} and \ref{lem:upsilonIsoMRI} coincide. Therefore, we can apply Theorem \ref{thm:recovery} with the following upper bound for $\Gamma(S,\pi)$
$$ \Gamma(S,\pi)\lesssim  \max_{0\leq j \leq J} \frac{1}{\tilde{\pi}_j} 2^{-j} \sum_{p=0}^J 2^{-|j-p|/2} s_p,
$$
and conclude the proof for Corollary \ref{corol:isolatedMRI}.

\subsection{Proof of Corollary \ref{corol:Fourier2DLines}}
\label{app:proofFourier2DLines}

Recall that $A_0 = \phi \otimes \phi \in \Cbb^{n \times n}$, where $\phi \in \Cbb^{\sqrt{n}\times \sqrt{n}}$ is a 1D orthogonal transform. Consider a blocks dictionary made of $\sqrt{n}$ horizontal lines, i.e. for $1\leq k \leq \sqrt{n}$
$$ B_k = \left( \phi_{k,1} \phi ,  \hdots , \phi_{k,\sqrt{n} } \phi \right), \quad \text{ and thus } \quad B_k^*B_k = \left( \phi_{k,i}^* \phi_{k,j} \Id_{\sqrt{n}} \right)_{1\leq i,j \leq \sqrt{n} }.
$$
Now, let us fix that the signal support $S$ is concentrated on $q$ horizontal lines of the spatial plane. Formally, 
\begin{align}
S \subset \{ (j-1) \sqrt{n} + \{1, \hdots , \sqrt{n}\}, j\in J\}
\end{align} 
where $J\subset \{1, \hdots , \sqrt{n}\}$ and $|J |=q$. Therefore, 
$$ B_k^*B_{k,S} = \left( \delta_{j\in J}\phi_{k,i}^* \phi_{k,j} \Id_{\sqrt{n}} \right)_{1\leq i,j \leq \sqrt{n}},
$$
where $\delta_{j \in J}=1$ if $j\in J$ and $0$ otherwise.
In such a setting, the quantities in Definition \ref{def:quantities} can be rewritten as follows:
\begin{align}
\Theta(S,\pi) &=\max_{1\leq k \leq M} \max_{1\leq i \leq n} \frac{ \| e_i^* B_k^* B_{k,S} \|_1}{\pi_k} =  \max_{1\leq k \leq \sqrt{n}}  \max_{1\leq \tilde{i} \leq \sqrt{n}} \frac{|\phi_{k,\tilde{i}}| \sum_{j\in J} | \phi_{k,j}|}{\pi_k} \leq \max_{1\leq k \leq \sqrt{n}} q  \frac{\|\phi_{k,:}\|_\infty^2}{\pi_k}.
\end{align}

Recall that 
$$ \Upsilon(S,\pi)  :=\max_{1\leq i \leq n} \sup_{\|v\|_\infty \leq 1 }\sum_{k=1}^M   \frac{1}{\pi_k}  \left| e_i^*  B_{k}^* B_{k,S} v  \right|^2,
$$
and call $(i^{\star}, v)$ the argument of the supremum over $\{ 1, \hdots , n\}$ and $\{u, \|u\|_\infty \leq 1 \}$. 
Therefore,
$$\Upsilon(S,\pi)  =  \sum_{k=1}^M   \frac{1}{\pi_k}  \left| e_{i^\star}^*  B_{k}^* B_{k,S} v  \right|^2.
$$
We can decompose $i^\star = (i_1-1)\sqrt{n}+i_2$ with $i_1$, $i_2$ integers of $\{ 1, \hdots , \sqrt{n}\}$. We can write
\begin{align*}
\Upsilon(S,\pi)  &=  \sum_{k=1}^{\sqrt{n}} \frac{1}{\pi_k} \left|\sum_{j=1}^{\sqrt{n}} \delta_{j \in J} \phi_{k,i_1}^* \phi_{k,j} e_{i_2}^* {v}[j] \right|^2 =  \sum_{k=1}^{\sqrt{n}} \frac{1}{\pi_k} |\phi_{k,i_1} |^2\left|\sum_{j=1}^{\sqrt{n}} \delta_{j \in J} \phi_{k,j} w_j \right|^2,
\end{align*}
where $w\in \Cbb^{\sqrt{n}}$ such that $w_j = e_{i_2}^* {v}[j]$ and ${v}[j] \in \Cbb^{\sqrt{n}}$ is the restriction of ${v}$ to the $j$-th horizontal line, i.e. to the components of $v$ indexed by $\{(j-1)\sqrt{n}+1,\hdots, j\sqrt{n}\}$. 
We can rewrite the last expression as follows
\begin{align*}
\Upsilon(S,\pi)  &=  \sum_{k=1}^{\sqrt{n}} \frac{1}{\pi_k} |\phi_{k,i_1} |^2\left| \left\langle e_k , \phi P_J w \right\rangle \right|^2 \leq  \max_{1\leq \ell \leq \sqrt{n}}  \frac{1}{\pi_\ell} |\phi_{\ell,i_1} |^2   \sum_{k=1}^{\sqrt{n}} \left| \left\langle e_k , \phi P_J w \right\rangle \right|^2 \\
&= \max_{1\leq \ell \leq \sqrt{n}}   \frac{1}{\pi_\ell} |\phi_{\ell,i_1} |^2    \| \phi P_J w\|_2^2 = \max_{1\leq \ell \leq \sqrt{n}}  \frac{1}{\pi_\ell} |\phi_{\ell,i_1} |^2    \| P_J w\|_2^2 \\
&\leq   \max_{1\leq \ell \leq \sqrt{n}} \frac{1}{\pi_\ell} |\phi_{\ell,i_1} |^2  \cdot q,
\end{align*}
where in the last expression we use that $\|w\|_\infty\leq 1$.
Choosing $\phi$ as the 1D Fourier transform gives $\| \phi_{\ell,:}\|_\infty = \frac{1}{n^{1/4}}$ and choosing a uniform sampling among the $\sqrt{n}$ horizontal lines, i.e. $\pi_\ell^\star = 1/\sqrt{n}$ for $1\leq \ell \leq \sqrt{n}$, leads to
\begin{align*}
\Gamma(S, \pi^{\star}) \leq q,
\end{align*} 
which ends the proof of Corollary \ref{corol:Fourier2DLines}.

 \subsection{Proof of Corollary \ref{corol:linesShannonMRI}}
 \label{app:proofLinesShannonMRI}

We recall that the sampling matrix is then constructed from the full sampling matrix $A_0\in \Cbb^{n\times n}$, in the 2D setting, where $A_0 = \Fc_{2D} \Psi^*$ with $\Fc_{2D} \in \Cbb^{n\times n}$ the 2D Fourier transform and $\Psi^* \in \Cbb^{n\times n}$ the 2D inverse wavelet transform. Since both transforms are separable, $\Fc_{2D} = \Fc \otimes \Fc$, $\Psi = \psi \otimes \psi$, with $\otimes$ the Kronecker product and $\Fc, \psi \in \Cbb^{\sqrt{n}\times \sqrt{n}}$ the corresponding 1D transforms. Then $A_0$ can also be rewritten as $A_0=\phi \otimes \phi$, the Kronecker product of the 1D transforms $\phi :=\Fc\psi^*\in \Cbb^{\sqrt{n} \times \sqrt{n}}$. 

In this section, in order to avoid any confusion, we will denote by $\left( e_i^{(n)} \right)_{1\leq i \leq n}$ the canonical basis in dimension $n$.


In Corollary \ref{corol:linesShannonMRI}, we focus on the case where $A_0 = \phi \otimes \phi \in \Cbb^{n \times n}$ is the 2D Fourier-Shannon wavelet transform, then $\phi\in \Cbb^{\sqrt{n} \times \sqrt{n}}$ is the 1D Fourier-Shannon wavelets transform. Therefore, $\phi$ and $A_0$ are block-diagonal orthogonal matrices.
The sensing schemes are based on horizontal lines on the 2D plane, meaning that
$$ B_k = \left( \phi_{k,1} \phi \hdots \phi_{k,\sqrt{n}} \phi \right),
$$
for $k=1,\hdots , \sqrt{n}$.
By defintion of the Fourier-Shannon transform, we have that
 $$ B_k^* B_k =  \left( \phi_{k,\ell}^* \phi_{k,m} \Id_{\sqrt{n}} \right)_{1\leq \ell, m \leq \sqrt{n}}   = \frac{1}{2^{j(k)}} \left( \delta_{\ell \in \tau_{j(k)}} \delta_{m \in \tau_{j(k)}} \Id_{\sqrt{n}} \right)_{1\leq \ell, m \leq \sqrt{n}}, 
 $$
 for $k = 1, \hdots , \sqrt{n}$, where  $\delta_{\ell \in \tau_{j}} = 1$ if $\ell \in \tau_{j}$, and 0 otherwise.
 
 First let us start with the evaluation of $\Theta$. 
By definition of $\|\cdot \|_{\infty \rightarrow \infty}$, we have
$$ \| B_k^* B_{k,S} \|_{\infty \rightarrow \infty} = \max_{1\leq \ell \leq n} \sup_{\| v \|_\infty \leq 1 \atop v \in \Cbb^n} 
\left| \left( e_\ell^{(n)}\right)^* B_k^* B_k P_S {v} \right|.
$$ 
Setting $\tilde{v}=\tilde{v}(k)$ the argument of the supremum in the last expression, then
 \begin{align*}
 \Theta := \max_{1\leq k \leq \sqrt{n}} \max_{1 \leq \ell \leq n} \frac{1}{\pi_k} \left| \left( e_\ell^{(n)}\right)^* B_k^* B_k P_S \tilde{v} \right|,
 \end{align*}
Note that $\| \tilde{v} \|_\infty \leq 1$. The index $\ell$ can be rewritten as $\ell = (\ell_1 - 1)\sqrt{n} + \ell_2$, with $1\leq \ell_1 , \ell_2 \leq \sqrt{n}$.
  \begin{align*}
 \Theta &:= \max_{1\leq k \leq \sqrt{n}} \max_{1 \leq \ell_1 , \ell_2 \leq \sqrt{n}} \frac{1}{\pi_k} \left| \phi_{k,\ell_1}^* \left( \phi_{k,m}  \left( e_{\ell_2}^{(\sqrt{n})} \right)^*\right)_{1\leq m \leq \sqrt{n}} P_S \tilde{v} \right|, \\
 & = \max_{1\leq k \leq \sqrt{n}} \max_{1 \leq \ell_1 , \ell_2 \leq \sqrt{n}} \frac{1}{\pi_k} \left| \phi_{k,\ell_1}^* \sum_{m=1}^{\sqrt{n}} \phi_{k,m}  \left( e_{\ell_2}^{(\sqrt{n})} \right)^* \left( P_S \tilde{v} \right)[m] \right|,\\
 \end{align*}
where $\left( v \right)[m] \in \Cbb^{\sqrt{n}}$ is the restriction of the vector $v$ to the $m$-th horizontal line, i.e. to the components indexed by $\{(m-1)\sqrt{n}+1,
\hdots, m\sqrt{n}\}$.  Set $w^{\mid (m)} :=  \left( P_S \tilde{v} \right)[m] 
\in \Cbb^{\sqrt{n}}$, the restriction of $P_S \tilde{v}$ to the $m$-th horizontal line. Then the $\ell_2$-th component of $w^{\mid (m)}$, 
written as $w^{\mid (m)}_{\ell_2}$ is equal to $\left( e_{\ell_2}^{(\sqrt{n})} 
\right)^*  \left( P_S \tilde{v} \right)[m]$. Note that $\left| w^{\mid (m)}_{\ell_2}\right| \leq 1$ if $(m-1)\sqrt{n} + \ell_2 \in S$, and it is equal to 0 otherwise. Then,
  \begin{align}
  \label{exp:ThetaShannon}
 \Theta 
 & \leq \max_{1\leq k \leq \sqrt{n}} \max_{1 \leq \ell_1 , \ell_2 \leq \sqrt{n}} \frac{1}{\pi_k} \left| \phi_{k,\ell_1}^* \sum_{m=1}^{\sqrt{n}} \phi_{k,m}  w^{\mid (m)}_{\ell_2}\right|.
 \end{align}
 By the properties of block-diagonality of the Fourier-Shannon transform, we have
   \begin{align}
   \label{exp:normInf}
 \Theta 
 & \leq \max_{1\leq k \leq \sqrt{n}} \max_{1 \leq \ell_1 , \ell_2 \leq \sqrt{n}} \frac{1}{\pi_k} \left| \phi_{k,\ell_1}^*\sum_{m\in \tau_{j(k)}} \phi_{k,m}  w^{\mid (m)}_{\ell_2}\right| \\
 \notag
 & \leq \max_{1\leq k \leq \sqrt{n}} \max_{1 \leq \ell_2 \leq \sqrt{n}} \frac{1}{\pi_k} \left\| \phi_{k,:} \right\|_\infty^2\left| \sum_{m\in \tau_{j(k)}}  w^{\mid (m)}_{\ell_2}\right| \\
 \notag
  & \leq \max_{1\leq k \leq \sqrt{n}} \max_{1 \leq \ell_2 \leq \sqrt{n}} \frac{1}{\pi_k} \left\| \phi_{k,:} \right\|_\infty^2\sum_{m\in \tau_{j(k)}} \left| w^{\mid (m)}_{\ell_2} \right| \\
  \label{eq:thetaUp}
  & \lesssim \max_{1\leq k \leq \sqrt{n}}  \frac{1}{\pi_k} \frac{1}{2^{j(k)}} s^{c}_{j(k)}.
 \end{align}
 Indeed, $\sum_{m\in \tau_{j(k)}} \left| w^{\mid (m)}_{\ell_2} \right|$ is bounded above by $ \sum_{m \in \tau_{j(k)}} \delta_{(m-1)\sqrt{n} + \ell_2 \in S}$, which counts the number of intersections between $S$, the $\ell_2$th-column and the $j(k)$ (horizontal) level, see the blue line in Figure \ref{fig:illusEns}. Taking the maximum over $1\leq \ell_2 \leq \sqrt{n}$ leads to $\sum_{m \in \tau_{j(k)}} \delta_{(m-1)\sqrt{n} + \ell_2 \in S} \leq s_{j(k)}^c$.

Secondly, let us evaluate $\Upsilon$. We have that
 \begin{align*}
\Upsilon := \max_{1\leq \ell \leq n} \sum_{k=1}^{\sqrt{n}} \frac{1}{\pi_k} \left| \left( e_\ell^{(n)} \right)^* B_k^* B_k \tilde{v} \right|^2,
\end{align*} 
where $\tilde{v}=\tilde{v}(\ell)$ is the argument of the supremum on the $\ell_\infty$ unit-ball. Using \eqref{exp:normInf}, with $\ell = (\ell_1-1)\sqrt{n}+\ell_2$, we can rewrite 
\begin{align*}
\Upsilon =  \max_{1\leq \ell_1, \ell_2 \leq \sqrt{n}} \sum_{k=1}^{\sqrt{n}} \frac{1}{\pi_k} \left| \phi_{k,\ell_1 }^* \sum_{m=1}^{\sqrt{n}} \phi_{k,m} w^{\mid (m)}_{\ell_2}\right|^2,
\end{align*} 
 where $w^{\mid (m)}_{\ell_2} := \left( e_{\ell_2}^{(\sqrt{n})} \right)^*  \left( P_S \tilde{v} \right)[m]$. Note again that $\left| w^{\mid (m)}_{\ell_2} \right| \leq 1$ if $(m-1)\sqrt{n} + \ell_2 \in S$, and it is equal to 0 otherwise.
 By denoting $w^{|(:,\ell_2)}$ the vector with components 
 \begin{align}
 \label{def:w}
 w^{|(:,\ell_2)} &:=\left( w^{\mid (1)}_{\ell_2} , \quad w^{\mid (2)}_{\ell_2} ,  \hdots, w^{\mid (\sqrt{n})}_{\ell_2}  \right)^*,
 \end{align}
 we can rewrite the previous quantity as follows
\begin{align}
\label{exp:UpsiShannon}
\Upsilon &=  \max_{1\leq \ell_1, \ell_2 \leq \sqrt{n}} \sum_{k=1}^{\sqrt{n}} \frac{1}{\pi_k} \left| \phi_{k,\ell_1 }^* \left\langle \phi_{k,:}^* , w^{|(:,\ell_2)} \right\rangle  \right|^2 \\
\notag
&= \max_{1\leq \ell_1, \ell_2 \leq \sqrt{n}} \sum_{k=1}^{\sqrt{n}} \frac{1}{\pi_k} \left| \phi_{k,\ell_1 } \right|^2 \left| \left\langle \phi_{k,:}^* , w^{|(:,\ell_2)} \right\rangle  \right|^2.
\end{align}  
Since $\phi$ is an orthogonal block-diagonal transform, we have
\begin{align*}
\Upsilon  &= \max_{1\leq \ell_1, \ell_2 \leq \sqrt{n}} \sum_{k\in\tau_{j(\ell_1)}} \frac{1}{\pi_k} \left| \phi_{k,\ell_1 } \right|^2 \left| \left\langle \phi_{k,:}^* , w^{|(:,\ell_2)} \right\rangle  \right|^2.
\end{align*}
Choosing $\pi_k = \tilde{\pi}_j$ for $k \in \tau_j$ meaning that the probability of drawing lines is constant by levels, we can write that
\begin{align*}
\Upsilon  &= \max_{1\leq \ell_1, \ell_2 \leq \sqrt{n}} \frac{1}{\tilde{\pi}_{j(\ell_1)}} \sum_{k\in\tau_{j(\ell_1)}}  \left| \phi_{k,\ell_1 } \right|^2 \left| \left\langle \phi_{k,:}^* , w^{|(:,\ell_2)} \right\rangle  \right|^2, \\
&\leq  \max_{1\leq \ell_1, \ell_2 \leq \sqrt{n}} \frac{1}{\tilde{\pi}_{j(\ell_1)}} \sum_{k\in \tau_{j(\ell_1)}} \| \phi_{k,:} \|_\infty^2 \left| \left\langle \phi_{k,:}^* ,w^{|(:,\ell_2)}\right\rangle  \right|^2, \\
&\lesssim  \max_{1\leq \ell_1, \ell_2 \leq \sqrt{n}} \frac{2^{-j(\ell_1)} }{\tilde{\pi}_{j(\ell_1)}} \sum_{k\in \tau_{j(\ell_1)}}\left| \left\langle \phi_{k,:}^* , w^{|(:,\ell_2)}\right\rangle  \right|^2, \\
&=  \max_{1\leq \ell_1, \ell_2 \leq \sqrt{n}} \frac{2^{-j(\ell_1)} }{\tilde{\pi}_{j(\ell_1)}} \left\| P_{\tau_{j(\ell_1)}} \phi w^{|(:,\ell_2)} \right\|_2^2.
\end{align*}
Since $\phi$ is orthogonal and block diagonal we have $\left\| P_{\tau_{j(\ell_1)}} \phi w^{|(:,\ell_2)} \right\|_2^2 = \| P_{\tau_{j(\ell_1)}} w^{|(:,\ell_2)} \|_2^2$. Then,
\begin{align}
\notag
\Upsilon 
&\lesssim  \max_{1\leq \ell_1, \ell_2 \leq \sqrt{n}} \frac{2^{-j(\ell_1)} }{\tilde{\pi}_{j(\ell_1)}} \left\| P_{\tau_{j(\ell_1)}}  w^{|(:,\ell_2)} \right\|_2^2, \\
  \label{eq:UpsiUp}
&\lesssim \max_{1\leq \ell_1 \leq \sqrt{n}} \frac{2^{-j(\ell_1)} }{\tilde{\pi}_{j(\ell_1)}}  s^c_{j(\ell_1)},
\end{align}
where the last step invokes that $\| P_{\tau_{j(\ell_1)}} w^{|(:,\ell_2)} \|_2^2 \leq \sum_{m \in \tau_{j(\ell_1)}} \delta_{(m-1)\sqrt{n} + \ell_2 \in S} \leq s^c_{j(\ell_1)}$.
Note that the upper bounds \eqref{eq:thetaUp} and \eqref{eq:UpsiUp} on $\Upsilon$ and $\Theta$ coincide.
They lead to the following choice for $1\leq k \leq \sqrt{n}$,
$$ \pi_k = \tilde{\pi}_{j(k)} =\frac{s^{c}_{j(k)}2^{-j(k)}}{\sum_{\ell=1}^{\sqrt{n}} s^{c}_{j(\ell)}2^{-j(\ell)}} = \frac{s^{c}_{j(k)}2^{-j(k)}}{\sum_{j=0}^{J} \sum_{\ell \in \tau_j} s^{c}_{j}2^{-j}} = \frac{s^{c}_{j(k)}2^{-j(k)}}{\sum_{j=0}^{J}  s^{c}_{j}}.
$$
Then for this particular choice, we can rewrite
\begin{align*}
\max (\Theta, \Upsilon ) & \lesssim \sum_{j=0}^{J}  s^{c}_{j}.
\end{align*}
To conclude, by Theorem \ref{thm:recovery}, a lower bound on the required number of horizontal lines to acquire is thus
$$ m \gtrsim \sum_{j=0}^{J}  s^{c}_{j}\ln(s) \ln(n/\varepsilon).
$$

\subsection{Proof of Corollary \ref{corol:linesHaarMRI}}
\label{app:prooflinesHaarMRI}

In this part, using the formalism introduced in the last section, $\psi$ is the 1D Haar transform, and $\phi$ is then the Fourier-Haar's wavelet transform.
In such a case, we can reuse \eqref{exp:ThetaShannon} in Section \ref{app:proofLinesShannonMRI} to evaluate $\Theta$:
\begin{align*}
 \Theta 
 & \leq \max_{1\leq k \leq \sqrt{n}} \max_{1 \leq \ell_1 , \ell_2 \leq \sqrt{n}} \frac{1}{\pi_k} \left| \phi_{k,\ell_1} \sum_{m=1}^{\sqrt{n}} \phi_{k,m}  w_m[\ell_2] \right|.
\end{align*}
Using Lemma \ref{lem:adcock1}, we have for $1\leq k,m \leq \sqrt{n}$,
$$\left| \phi_{k,m} \right|  \lesssim 2^{-j(k)/2} 2^{-|j(k) - j(m)|/2}.$$
Therefore,
\begin{align}
\notag
 \Theta 
 & \leq \max_{1\leq k \leq \sqrt{n}} \max_{1 \leq \ell_1 , \ell_2 \leq \sqrt{n}} \frac{1}{\pi_k} \left| \phi_{k,\ell_1}^* \right|  \sum_{m=1}^{\sqrt{n}} \left| \phi_{k,m} w^{\mid (m)}_{\ell_2} \right| \\
 \notag
  & \leq \max_{1\leq k \leq \sqrt{n}} \max_{1 \leq \ell_1 , \ell_2 \leq \sqrt{n}} \frac{1}{\pi_k} \left| \phi_{k,\ell_1}^* \right|  \sum_{j=0}^{J} \sum_{m \in \tau_{j}} \left| \phi_{k,m} \right| \left|  w^{\mid (m)}_{\ell_2}  \right| \\
  \notag
    & \leq \max_{1\leq k \leq \sqrt{n}} \max_{1 \leq \ell_1 , \ell_2 \leq \sqrt{n}} \frac{1}{\pi_k} \left| \phi_{k,\ell_1}^* \right|  \sum_{j=0}^{J} \sum_{m \in \tau_{j}} \left| \phi_{k,m} \right| \left| w^{\mid (m)}_{\ell_2} \right| \\
    \notag
    &\lesssim \max_{1\leq k \leq \sqrt{n}} \max_{1 \leq  \ell_2 \leq \sqrt{n}} \frac{1}{\pi_k} 2^{-j(k)} \sum_{j=0}^{J} 2^{-|j(k) - j|/2} \sum_{m \in \tau_{j}}    \left|  w^{\mid (m)}_{\ell_2} \right| \\
    \label{exp:ThetaUpHaar}
        &\lesssim \max_{1\leq k \leq \sqrt{n}}  \frac{1}{\pi_k} 2^{-j(k)} \sum_{j=0}^{J} 2^{-|j(k) - j|/2} s^c_{j}.
\end{align}
 
 Now let us study $\Upsilon$.
Recall the definition of $w^{|(:,\ell_2)}$ depending on $\ell_2$ in \eqref{def:w}, we can reuse \eqref{exp:UpsiShannon} to have
\begin{align*}
\Upsilon &=  \max_{1\leq \ell_1, \ell_2 \leq \sqrt{n}} \sum_{k=1}^{\sqrt{n}} \frac{1}{\pi_k} \left| \phi_{k,\ell_1 }^* \left\langle \phi_{k,:}^* , w^{|(:,\ell_2)}\right\rangle  \right|^2 \\
&= \max_{1\leq \ell_1, \ell_2 \leq \sqrt{n}} \sum_{k=1}^{\sqrt{n}} \frac{1}{\pi_k} \left| \phi_{k,\ell_1 }^* \right|^2 {\left| \left\langle \phi_{k,:}^* , w^{|(:,\ell_2)} \right\rangle  \right|^2}, \\
&= \max_{1\leq \ell_1 , \ell_2 \leq \sqrt{n}} \sum_{j=0}^{J} \frac{1}{\tilde{\pi}_j} \sum_{k\in\tau_j}  \left| \phi_{k,\ell_1 }^* \right|^2 {\left| \left\langle \phi_{k,:}^* , w^{|(:,\ell_2)} \right\rangle  \right|^2}, 
\end{align*}  
by choosing $\pi_k = \tilde{\pi}_j$ for $k\in \tau_j$, meaning that the drawing probability is constant by level. Since for $k\in \tau_j$, we have $\left| \phi_{k,\ell_1 }^* \right|^2 \leq 2^{-j} 2^{-|j - j(\ell_1)|}$ by Lemma \ref{lem:adcock1}. Then,
\begin{align*}
\Upsilon 
&= \max_{1\leq \ell_1, \ell_2 \leq \sqrt{n}} \sum_{j=0}^{J} \frac{1}{\tilde{\pi}_j} 2^{-j} 2^{-|j - j(\ell_1)|} \underbrace{\sum_{k\in\tau_j}  \left| \left\langle \phi_{k,:}^* , w^{|(:,\ell_2)} \right\rangle  \right|^2}_{=:K_j}.
\end{align*}  
Dealing with $K_j$, we can derive that
\begin{align*}
\sqrt{K_j} &= \left\| P_{\tau_j} \phi^* w^{|(:,\ell_2)}\right\|_2 = \left\| P_{\tau_j} \phi^* \sum_{r=0}^J P_{\tau_r} w^{|(:,\ell_2)} \right\|_2
 &\leq \sum_{r=0}^J \left\| P_{\tau_{j}} \phi^* P_{\tau_r} \right\|_{2 \rightarrow 2} \left\| P_{\tau_r} w^{|(:,\ell_2)} \right\|_2 \\
 &\lesssim \sum_{r=0}^J 2^{-| j - r|/2} \sqrt{s^c_{r}},
\end{align*}
where the upper bound $\left\| P_{\tau_{j}} \phi^* P_{\tau_r} \right\|_{2 \rightarrow 2}  \lesssim 2^{-| j - r|/2} $ can be found in \cite[Lemma 4.3]{adcock2014note}.
Then,
\begin{align*}
K_k &\lesssim  \left( \sum_{r=0}^J 2^{-| j  - r|/2} \sqrt{s^c_{r}} \right)^2 \lesssim \left( \sum_{r=0}^J 2^{-| j  - r|/2} \right) \left( \sum_{r=0}^J 2^{-| j  - r|/2} {s^c_{r}} \right) \\
&\lesssim \sum_{r=0}^J 2^{-| j - r|/2} {s^c_{r}}.
\end{align*}
Therefore, 
\begin{align}
\notag
\Upsilon 
&\lesssim \max_{1\leq \ell_1\leq \sqrt{n}} \sum_{j=0}^{J} \frac{1}{\tilde{\pi}_j} 2^{-j} 2^{-|j-j(\ell_1)|}  \sum_{r=0}^J 2^{-| j - r|/2} {s^c_{r}} \\
\notag
&\lesssim \left( \max_{1\leq \ell_1\leq \sqrt{n}}   \sum_{j=0}^J 2^{-|j-j(\ell_1)|}  \right) \left(  \max_{0\leq j \leq J}  \frac{2^{-j}}{\tilde{\pi}_j}  \sum_{r=0}^J 2^{-| j - r|/2} {s^c_{r}} \right) \\
    \label{exp:UpsiUpHaar}
&\lesssim  \max_{0\leq j \leq J} \frac{2^{-j}}{\tilde{\pi}_j}  \sum_{r=0}^J 2^{-| j - r|/2} {s^c_{r}} .
\end{align}   
The upper bounds \eqref{exp:ThetaUpHaar} and \eqref{exp:UpsiUpHaar} give
$$ \max(\Theta, \Upsilon) \lesssim \max_{0\leq j \leq J} \frac{2^{-j}}{\tilde{\pi}_j}  \sum_{r=0}^J 2^{-| j - r|/2} {s^c_{r}}.
$$
Therefore, by Theorem \ref{thm:recovery}, a lower bound on the required number of horizontal lines is
$$ m \gtrsim \max_{0\leq j \leq J}\frac{2^{-j}}{\tilde{\pi}_j}  \sum_{r=0}^J 2^{-| j - r|/2} {s^c_{r}} \ln(n/\varepsilon) \ln(s).
$$
By choosing
$$ {\pi}_k = \tilde{\pi}_{j(k)} = \frac{2^{-j(k) \sum_{r=0}^J 2^{-|j(k)-r|/2} s^c_{r}}}{ 
\sum_{\ell=1}^{\sqrt{n}}  
2^{-j(\ell)} \sum_{r=0}^J 
2^{-|j(\ell) -r|/2} s^c_{r}  },
$$
for $1 \leq k \leq \sqrt{n}$, the lower bound on the required number of horizontal lines can be rewritten as
\begin{align*} 
m & \gtrsim \sum_{\ell=1}^{\sqrt{n}}  
2^{-j(\ell)} \sum_{r=0}^J 
2^{-|j(\ell) -r|/2} s^c_{r}  \cdot  \ln(n/\varepsilon) \ln(s) \\
& \gtrsim \sum_{j=0}^J \sum_{\ell \in \tau_j} 2^{-j} \sum_{r=0}^J 
2^{-|j -r|/2} s^c_{r}  \cdot  \ln(n/\varepsilon) \ln(s) \\
&\gtrsim \sum_{j=0}^J \sum_{r=0}^J 
2^{-|j -r|/2} s^c_{r}  \cdot  \ln(n/\varepsilon) \ln(s) \\
&\gtrsim \sum_{j=0}^J \left(  s^c_{j}  + \sum_{r=0 \atop r\neq j}^J 
2^{-|j -r|/2} s^c_{r}  \right) \cdot  \ln(n/\varepsilon) \ln(s),
\end{align*}
which concludes the proof of Corollary \ref{corol:linesHaarMRI}.

\bibliographystyle{alpha}
\bibliography{mybib}
\end{document}